\documentclass[sts]{imsart}

\RequirePackage{amsthm,amsmath,amsfonts,amssymb, algorithm, multirow, url}
\RequirePackage[authoryear]{natbib}%% uncomment this for author-year citations
\RequirePackage[colorlinks,citecolor=blue,urlcolor=blue]{hyperref}%% uncomment this for coloring bibliography citations and linked URLs
\RequirePackage{graphicx}%% uncomment this for including figures

\startlocaldefs
\theoremstyle{plain}
\newtheorem{theorem}{Theorem}
\newtheorem{lemma}{Lemma}
\newtheorem{corollary}{Corollary}
\theoremstyle{remark}                
\newtheorem{property}{Property}
\newfloat{Algorithm}{thp}{lop}
\floatname{Algorithm}{Algorithm}
\DeclareMathOperator*{\argmax}{arg\,max}

\newcommand\diag{\textnormal{diag}}
\newcommand\tr{\textnormal{tr}}
\newcommand\E{\textnormal{E}}
\newcommand\N{\textnormal{N}}
\renewcommand\P{\textnormal{P}}
\newcommand\new{\textnormal{new}}
\newcommand\old{\textnormal{old}}
\newcommand\var{\textnormal{var}}
\newcommand\cov{\textnormal{cov}}
\newcommand\obs{\textnormal{obs}}
\newcommand\augm{\textnormal{aug}}
\newcommand\mis{\textnormal{mis}}

\newcommand\e{\textnormal{e}}
\newcommand\opt{\textnormal{opt}}

\newcommand\vech{\textnormal{vech}}

\newcommand\bfa{{\bf{a}}}
\newcommand\bfc{{\bf{c}}}

\newcommand\bfm{{\bf{m}}}
\newcommand\bfr{{\bf{r}}}
\newcommand\bfs{{\bf{s}}}
\newcommand\bfu{{\bf{u}}}
\newcommand\bfw{{\bf{w}}}
\newcommand\bfx{{\bf{x}}}
\newcommand\bfy{{\bf{y}}}
\newcommand\bfz{{\bf{z}}}
\newcommand\bfalpha{{\pmb{\alpha}}}
\newcommand\bfeps{{\pmb{\epsilon}}}
\newcommand\bfeta{{\pmb{\eta}}}
\newcommand\bfmu{\pmb{\mu}}
\newcommand\bfphi{\pmb{\phi}}
\newcommand\bftheta{{\pmb{\theta}}}
\newcommand\bfsigma{{\pmb{\sigma}}}
\newcommand\bflambda{{\pmb{\lambda}}}
\newcommand\bfzero{{\bf{0}}}
\newcommand\bone{{\bf{1}}}

\newcommand\rmm{r_-}
\newcommand\rpp{r_+}
\newcommand\vp{\varphi_+}
\newcommand\vm{\varphi_-}
\graphicspath{{Figures/}}
\endlocaldefs

\begin{document}

\begin{frontmatter}
\title{Efficient data augmentation techniques for some classes of state space models}
%\title{A sample article title with some additional note\thanksref{T1}}
\runtitle{Efficient data augmentation for state space models}
%\thankstext{T1}{A sample of additional note to the title.}

\begin{aug}
\author[A]{\fnms{Linda S. L.} \snm{Tan}\ead[label=e1]{statsll@nus.edu.sg}}
%%%%%%%%%%%%%%%%%%%%%%%%%%%%%%%%%%%%%%%%%%%%%%
%% Addresses                                %%
%%%%%%%%%%%%%%%%%%%%%%%%%%%%%%%%%%%%%%%%%%%%%%
\address[A]{Linda S. L. Tan is Assistant Professor at Department of Statistics and Data Science, National University of Singapore, \printead{e1}.}
\end{aug}

\begin{abstract}
Data augmentation improves the convergence of iterative algorithms, such as the EM algorithm and Gibbs sampler by introducing carefully designed latent variables. In this article, we first propose a data augmentation scheme for the first-order autoregression plus noise model, where optimal values of working parameters introduced for recentering and rescaling of the latent states, can be derived analytically by minimizing the fraction of missing information in the EM algorithm. The proposed data augmentation scheme is then utilized to design efficient Markov chain Monte Carlo (MCMC) algorithms for Bayesian inference of some non-Gaussian and nonlinear state space models, via a mixture of normals approximation coupled with a block-specific reparametrization strategy. Applications on simulated and benchmark real datasets indicate that the proposed MCMC sampler can yield improvements in simulation efficiency compared with centering, noncentering and even the ancillarity-sufficiency interweaving strategy. 
\end{abstract}


\end{frontmatter}
%%%%%%%%%%%%%%%%%%%%%%%%%%%%%%%%%%%%%%%%%%%%%%
%%%% Main text entry area:

\section{Introduction}
Data augmentation is a powerful approach for improving the efficiency of iterative algorithms such as the EM algorithm \citep{Dempster1977} and Markov chain Monte Carlo (MCMC) methods like the Gibbs sampler \citep{Geman1984}, via the introduction of carefully designed latent variables. Let $p(Y_\obs|\bftheta)$ be the density of observed data $Y_\obs$ conditional on unknown parameter $\bftheta \in \mathbb{R}^d$. A data augmentation scheme constructs unobserved or missing data, $Y_\mis$, such that 
\[
p(Y_\obs|\bftheta) = \int p(Y_\augm|\bftheta) \; dY_\mis,
\]
where $Y_\augm = (Y_\obs, Y_\mis)$ is the augmented data. 

An important development is the {\em working parameter} approach proposed by \cite{Meng1997, Meng1998} to improve the convergence rate of the EM algorithm for the multivariate $t$, Poisson and mixed effects models. In {\em conditional} augmentation, optimal working parameters are found by minimizing the fraction of missing information governing the EM algorithm convergence rate. \cite{Tan2007} applied this approach to quadratic optimization problems. Efficient MCMC algorithms for posterior sampling can also be constructed using {\em marginal} augmentation \citep{vanDyk2001}, where working parameters are assigned prior distributions and then marginalized out in the sampling scheme. \cite{Goplerud2021} extends marginal augmentation to variational Bayes \citep{Blei2017} to improve posterior approximation of non-nested binomial hierarchical models. Alternatively, PX-EM \citep{Liu1998} accelerates EM by expanding the set of model parameters to adjust the covariance among parameters in the M-step. Recently, \cite{Tak2020} proposed a data transforming augmentation scheme that reduces heteroscedastic models into homoscedastic ones.

In Bayesian inference, data augmentation is also known as reparametrization, and two well-known parametrizations are (hierarchical) centering and noncentering. Convergence rates of MCMC algorithms \citep{Roberts1997} for these schemes have been studied for random effect models \citep{Gelfand1995, Gelfand1996}, Gaussian state space models \citep{Pitt1999, Fruhwirth-Schnatter2004}, hierarchical models \citep{Papaspiliopoulos2003, Papaspiliopoulos2007} and Gaussian process based models \citep{Bass2017}. \cite{Zanella2021} investigate convergence of the Gibbs sampler for Gaussian multilevel models of depth more than two, lending insight on the choice of parametrization, that is also applicable beyond Gaussian models.  

Centering and noncentering play complementary roles as the Gibbs sampler often converges much faster under one parametrization than the other. \cite{Yu2011} developed ASIS (ancillarity-sufficiency interweaving strategy) to exploit this contrasting feature. \cite{Kastner2014} evaluate the performance of ASIS on stochastic volatility (SV) models \citep{Taylor1982}, while \cite{Kastner2017} proposed shallow and deep interweaving strategies for multivariate factor SV models, which can boost efficiency by several orders of magnitude. \cite{Simpson2017} apply ASIS to dynamic linear models using data augmentations based on (wrongly) scaled errors and disturbances. ASIS was also used to achieve shrinkage for time-varying parameter models \citep{Bitto2019}, perform model selection for factor SV models incorporating leverage, asymmetry and heavy tails \citep{Li2020}, and combined with elliptical slice sampling for posterior estimation in nonlinear state space models with univariate autoregressive state equation \citep{Kreuzer2020}. 

Another alternative is partial noncentering, which lies on the continuum between centering and noncentering, and is capable of utilizing the information in the data to determine an almost optimal parametrization. Partial noncentering has been shown to yield better convergence in MCMC methods than both centering and noncentering for random effect models \citep{Papaspiliopoulos2003} and spatial generalized linear mixed models \citep{Christensen2006}. It can also increase efficiency and provide more accurate posterior approximations in variational Bayes for generalized linear mixed models \citep{Tan2013, Tan2021}.

\subsection{Proposed data augmentation scheme}
We propose a data augmentation scheme for state space models with univariate observations $y_t$ and latent states $x_t$, where $y_t$ is generated from $p(y_t|x_t)$ independently and $(x_t)$ is a stationary AR(1) (first order autoregressive) process given by
\begin{equation} \label{AR(1)}
\begin{aligned}
x_{t+1} &= \mu + \phi (x_t - \mu)+ \sigma_\eta \eta_t,  && t \geq 1, \\
x_1 &\sim \N(\mu, \sigma_\eta^2/(1-\phi^2)),
\end{aligned}
\end{equation}
$\eta_t  \overset{\text{iid}}{\sim} \N(0,1)$, $\mu \in \mathbb{R}$, persistence $|\phi| < 1$ and $\sigma_\eta > 0$. Many influential models fall inside this framework, such as the SV model  with leverage \citep{Omori2007}, skewness or heavy tails \citep{Abanto2014}, stochastic conditional duration model \citep{Bauwens2004} and stochastic copula autoregressive model \citep{Almeida2012}. Our scheme introduces working parameters, $a$ and $w_t$, to rescale and recenter $x_t$ so that the transformed latent state is 
\begin{equation} \label{PNCP}
\alpha_t = (x_t - w_t \mu)/ \sigma_\eta^{a}, \quad a \in \mathbb{R}, \;  w_t \in \mathbb{R}.
\end{equation}

\subsection{Maximum likelihood parameter estimation}
First, we consider maximum likelihood estimation of the parameters of the AR(1) plus noise model using the EM algorithm. The observation equation is 
\begin{align*}
y_t &= x_t + \sigma_\epsilon \epsilon_t, &&  \epsilon_t \overset{\text{iid}}{\sim} \N(0,1), \; t \geq 1,
\end{align*}
$\sigma_\epsilon > 0$, and $(\epsilon_t)$ and $(\eta_s)$ are independent for all $t$ and $s$, and of $(x_t)$. In this Gaussian context, optimal working parameters in different settings can be derived analytically by minimizing the fraction of missing information. Studies of their large sample properties reveal features distinct from random effect models. To incorporate the optimal schemes for inferring each parameter, an alternating expectation-conditional maximization algorithm \citep{Meng1997} is designed and shown to be more attractive than centering and noncentering. 

The EM algorithm is a natural approach for maximum likelihood estimation of parameters in linear Gaussian state space models \citep{Shumway1982} as the latent states can be treated as missing data and the E-step can be performed using smoothed estimates from the Kalman filter \citep{Kalman1960}. Compared with gradient-ascent (scoring) methods, EM is simple to implement, numerically stable (likelihood increases monotonically) and guaranteed to converge to a local maximum \citep{Wu1983}. It is preferred when the M-step is tractable or parameters are high-dimensional as computational costs tend to be lower. While Newton and quasi-Newton algorithms converge quadratically or superlinearly, the EM algorithm converges linearly or sublinearly. Slow convergence often occurs in later stages, especially in high signal-to-noise ratio settings \citep{Olsson2006}. Besides data augmentation, techniques proposed to accelerate EM include quasi-Newton \citep{Jamshidian1997, Zhou2011}, extrapolation \citep{Saadaou2010i}, parabolic \citep{Berlinet2009} and Anderson acceleration \citep{Henderson2019} and linearly preconditioned nonlinear conjugate gradient \citep{Zhou2021}. Both EM and scoring can also be combined with particle methods in an online or offline manner to obtain maximum likelihood parameter estimates of general state space models \citep{Kantas2015}.

\subsection{Bayesian parameter estimation and smoothing}
As the rate of convergence of EM algorithms and Gibbs samplers are closely linked \citep{Sahu1999}, data augmentation schemes optimized for the EM algorithm can potentially be utilized to design efficient MCMC algorithms. Here we focus on parameter estimation and smoothing in a Bayesian framework but do not address online filtering. We consider two non-Gaussian and nonlinear state space models; the SV model for financial returns and the stochastic conditional duration model for time intervals between transactions, whose observation equations can be linked to the AR(1) plus noise model via mixture of normals approximations \citep{Kim1998}. A block-specific reparametrization (BSR) strategy, which allows the parametrization of latent states to vary across blocks in a MCMC sampler, is proposed to incorporate the optimal schemes derived previously. Experiments on simulated data and real applications indicate that BSR always performs better than the worse of centering and noncentering, and often surpasses both and even ASIS. 

Sequential Monte Carlo or particle filters \citep{Doucet2001} are widely used for online state estimation, and have also become popular for smoothing and parameter estimation. The particle learning and smoothing algorithm \citep{Carvalho2010} incorporates parameter estimation via a fully adapted filter and smoothing via recursive backward sampling. \cite{Yang2018} propose a modification of this algorithm to capture dependence between states and parameters, and a new smoothing algorithm called refiltering which can be parallelized easily.

For nonlinear and non-Gaussian state space models, it is difficult to design efficient proposal densities for sampling from the (conditional) posterior of high-dimensional latent states. Particle MCMC \citep{Andrieu2010} overcome this by using sequential Monte Carlo to build efficient proposals. Particle marginal Metropolis-Hastings, for instance, uses the marginal likelihood estimated from a particle filter run at each iterate to compute acceptance probabilities. For good mixing, the number of particles should scale linearly with the number of observations. \cite{Fearnhead2016} suggest using data augmentation to reduce Monte Carlo error in the particle filter. In contrast, gradient-based sampling methods such as Hamiltonian Monte Carlo \citep{Neal2011} and Metropolis adjusted Langevin \citep{Girolami2011} construct proposal densities adaptively based on the local geometry of the target, and gradients must be computed or approximated numerically at each iteration. Compared to standard MCMC, particle MCMC and gradient-based methods have higher computational cost, but are more generally applicable as they require little design effort. The efficient tuning of parameters in Hamiltonian Monte Carlo is a challenging task, and the No-U-Turn Sampler \citep{Hoffman2014} in Stan \citep{Stan2019} employs automatic tuning based on the leapfrog method. \cite{Kleppe2019} introduce dynamically rescaled Hamiltonian Monte Carlo where the target distribution is reparametrized to have constant scaling, while \cite{Osmundsen2019} modify the target to be close to Gaussian using transport map. We provide some comparisons of BSR with particle marginal Metropolis-Hastings and Hamiltonian Monte Carlo implemented in the R package {\tt nimble} \citep{Valpine2017} and Stan respectively in one of the real applications.

\subsection{Extensions}
The general linear Gaussian state space model (dynamic linear model) can be written as 
\begin{align*}
\bfy_t &= Z_t \bfx_t + \bfeps_t,  && \bfeps_t \sim \N(\bfzero, H_t), \; t \geq 1, \\
\bfx_{t+1} - \bflambda &= \Phi_t (\bfx_t - \bflambda)+ \bfeta_t,  && \bfeta_t \sim \N(\bfzero, Q_t),
\end{align*}
with $\bfx_1 \sim \N(\bfmu, P)$. Here $\bfy_t \in \mathbb{R}^p$ is the observation vector, $\bfx_t \in \mathbb{R}^m$ is the state vector and the error sequences $(\bfeps_t)$, $(\bfeta_t)$ and $\bfx_1$ are assumed to be independent. The proposed scheme in \eqref{PNCP} is not applicable generally to the dynamic linear model as the idea of partial rescaling and recentering using fixed means and variances is usually relevant when $(\bfx_t)$ is a stationary process. However, it is still useful in some special cases where $\bflambda=\bfmu$, $\Phi_t$ and $Q_t$ do not vary with time, and $Z_t$ is known. 

If $m=1$, $\Phi_t=\phi \in (-1,1)$, $\bfmu = \mu$, $P = \sigma_\eta^2/(1-\phi^2)$ and $Q_t= \sigma_\eta^2$, then $(\bfx_t)$ reduces to the stationary AR(1) process and the proposed scheme in \eqref{PNCP} can be applied. Optimal values of working parameters for multivariate Gaussian observations can be derived similarly as for univariate $y_t$ and details are given in the Supplement \citep{Tan2022}. 

Suppose $m > 1$. Let scalar functions on vectors be defined elementwise, $\bfsigma_\eta = (\sigma_{\eta1}, \dots, \sigma_{\eta m})^T$ and $\bfphi = (\phi_1, \dots, \phi_m)^T$. If $\Phi_t = \diag(\bfphi)$ where each $|\phi_i| <1$, $P = \diag(\bfsigma_\eta^2/(1-\bfphi^2))$ and $Q_t = \diag(\bfsigma_\eta^2)$, then each element in $(\bfx_t)$ is an AR(1) process \citep[see, e.g. multivariate factor SV models in ][]{Li2020}. In this case, a data augmentation scheme similar to \eqref{PNCP}, where
\[
\bfalpha_t = (\bfx_t - \bfw_t \odot \bfmu)/ \bfsigma_{\eta}^\bfa, \quad \bfa \in \mathbb{R}^m, \;  \bfw_{t} \in \mathbb{R}^m,
\]
may be used. More generally, if $\Phi_t$, $Q_t$ and $P$ are defined such that $(\bfx_t)$ is a stationary process with mean $\bfmu$ and covariance matrix $\Sigma = LDL^T$, where $L$ is a unit lower triangular matrix and $D=\diag(\bfsigma_\eta^2)$, a scheme such as
\[
\bfalpha_t = L^{-1}(\bfx_t - \bfw_t \odot  \bfmu)/\bfsigma_\eta^\bfa
\]
may be useful. Optimal values of the working parameters can be found similarly by optimizing the convergence rate of the EM algorithm, albeit with higher complexity.

\section{EM algorithm for partially noncentered AR(1) plus noise model} \label{sec_EM}
The AR(1) plus noise model can be expressed in terms of $(\alpha_t)$ as 
\begin{align*}
y_t &= \sigma_\eta^a \alpha_t + w_t\mu + \sigma_\epsilon \epsilon_t,  && t\geq 1,\\
\sigma_\eta^a \alpha_{t+1} &= \bar{w}_{t+1} \mu + \phi ( \sigma_\eta^a \alpha_t  - \bar{w}_t \mu )+ \sigma_\eta \eta_t,  \nonumber \\
\sigma_\eta^a \alpha_1 &\sim \N\left(  \bar{w}_1 \mu ,\, {\sigma_\eta^2}/{(1-\phi^2)} \right), \nonumber
\end{align*}
where $\bar{w}_t=1-w_t$. Now $\mu$ and $\sigma_\eta$ appear in both the observation and state equations. Suppose we have observations $\bfy=(y_1, \ldots, y_n)^T$. Let $\bfx = (x_1, \ldots, x_n)^T$ be the latent states, $\bfalpha = (\alpha_1, \dots, \alpha_n)^T$, $\bfw =(w_1, \dots, w_n)^T$ and $\bar{\bfw} = \bone - \bfw$, where $\bone$ and $\bf0$ denote the vectors of all ones and zeros respectively (dimension inferred from context). When $a=0$ and $\bfw = \bf0$, we recover the fully centered parametrization (CP) in \eqref{AR(1)}, so called as $x_t$ is centered around the a priori expected value $\mu$, and the parameters $\mu$ and $\sigma_\eta^2$ appear only in the state equation. The fully noncentered parametrization (NCP) is obtained when $a=1$ and $\bfw = \bone$. Let $\bftheta = (\mu, \sigma_\eta^2, \phi, \sigma_\epsilon^2)^T$ be the model parameters, and the signal-to-noise ratio be defined as $\gamma = \sigma_\eta^2/\sigma_\epsilon^2$. In matrix notation, the AR(1) plus noise model can be expressed as
\begin{equation} \label{mod_G} 
\begin{gathered}
\bfy | \bfalpha, \bftheta  \sim \N(\sigma_\eta^a \bfalpha + \mu \bfw, \sigma_\epsilon^2 I), \\
\sigma_\eta^{a} \bfalpha | \bftheta \sim \N( \mu \bar{\bfw}, \sigma_\eta^{2}\Lambda^{-1} ),
\end{gathered}
\end{equation}
where $\Lambda$ is a symmetric {\it tridiagonal} matrix with off-diagonal elements equal to $-\phi$ and diagonal given by $(1, 1+\phi^2, \dots, 1+\phi^2, 1)^T$. For $|\phi| < 1$, $\Lambda$ is invertible. Regardless of the parametrization, the marginal distribution of $\bfy$ is $\N(\mu\bone, S)$ where $S= \sigma_\epsilon^2 I + \sigma_\eta^2 \Lambda^{-1}$.

Suppose we wish to find $\bftheta^* = (\mu^*, \sigma_\eta^{2*}, \phi^*, {\sigma_\epsilon^2}^*)^T$ that maximizes the observed data log-likelihood $\log p(\bfy|\bftheta) = \log \int p(\bfalpha,\bfy|\bftheta) \,d\bfalpha$ using an EM algorithm. We consider $\bfalpha$ as missing data and $(\bfy,\bfalpha)$ as augmented data. Given an initial estimate $\bftheta^{(0)}$, the algorithm performs an E-step and M-step at each iteration $i$, where the E-step computes $Q(\bftheta|\bftheta^{(i)}) = \E_{\bfalpha|\bfy, \bftheta^{(i)}} [\log p(\bfalpha, \bfy | \bftheta)]$, and the M-step maximizes $Q(\bftheta|\bftheta^{(i)}) $ with respect to $\bftheta$. The conditional distribution $p(\bfalpha | \bfy,\bftheta)$ is $\N(\bfm_{a\bfw}, V_a)$, where 
\begin{equation} \label{eq_conddistn}
\begin{gathered}
V_a =  \sigma_\eta^{-2a}(\sigma_\epsilon^{-2} I + \sigma_\eta^{-2}\Lambda)^{-1}, \\
\bfm_{a\bfw} = \sigma_\eta^{a} V_a [(\bfy -  \mu \bfw)/\sigma_\epsilon^2 + \mu \sigma_\eta^{-2}\Lambda \bar{\bfw}].
\end{gathered}
\end{equation}
The subscripts of $V_a$ and $\bfm_{a\bfw}$ represent their dependence on the values of $a$ and $\bfw$ in the scheme. For instance, when $a=0$ and $\bfw = \bone$, 
\begin{align*}
V_0 = (\sigma_\epsilon^{-2} I + \sigma_\eta^{-2}\Lambda)^{-1}, \;\; 
\bfm_{0\bone} =  V_0 (\bfy - \mu\bone)/\sigma_\epsilon^2.
\end{align*}

It can be derived that 
\begin{align*}
&Q(\bftheta | \bftheta^{(i)}) \negmedspace = \negmedspace   - n\log(2\pi) \negmedspace - \negmedspace \tfrac{1}{2} \big\{ [\sigma_\eta^{2a} \tr( V_{a}^{(i)}) + z^{(i)T}z^{(i)}]/ \sigma_\epsilon^2  \\
& \; + n(1-a) \log\sigma_\eta^2 + \sigma_\eta^{2(a-1)} \tr(\Lambda V_a^{(i)}) - \log(1-\phi^2)\\
& \; + n\log \sigma_\epsilon^2 + \sigma_\eta^{-2} (\sigma_\eta^a \bfm_{a\bfw}^{(i)} -\mu \bar{\bfw})^T \Lambda (\sigma_\eta^a \bfm_{a\bfw}^{(i)} -\mu \bar{\bfw}) \big\},
\end{align*}
where $z^{(i)} = \bfy - \mu \bfw- \sigma_\eta^a \bfm_{a\bfw}^{(i)}$, and $\bfm_{a\bfw}^{(i)}$ and $V_a^{(i)}$ are evaluated at $\bftheta^{(i)}$. At each iteration, the EM algorithm updates $\bfm_{a\bfw}$ and $V_{a}$ given $\bftheta^{(i)}$ and sets $\bftheta^{(i+1)} = \argmax_\bftheta Q(\bftheta | \bftheta^{(i)})$. The expectation-conditional maximization algorithm \citep{Meng1993} reduces the complexity of the M-step by replacing it with a sequence of conditional maximization steps. If we maximize $Q(\bftheta | \bftheta^{(i)})$ with respect to each element $\theta_s$ of $\bftheta$ with the remaining elements held fixed at current values, then the update of $\theta_s$ can be obtained by setting $\nabla_{\theta_s} Q(\bftheta | \bftheta^{(i)}) = 0$. This yields the following closed form updates for $\mu$ and $\sigma_\epsilon^2$:
\begin{equation} \label{update1}
\begin{gathered}
\mu = \frac{\sigma_\epsilon^{-2} (\bfy -  \sigma_\eta^a \bfm_{a\bfw}^{(i)})^T \bfw + \sigma_\eta^{a-2} \bfm_{a\bfw}^{(i)T} \Lambda \bar{\bfw} }{ \sigma_\epsilon^{-2} \bfw^T \bfw + \sigma_\eta^{-2} \bar{\bfw}^T \Lambda \bar{\bfw}}, \\
\sigma_\epsilon^2 = \tfrac{1}{n} [\sigma_\eta^{2a} \tr(V_a^{(i)}) + z^{(i)T} z^{(i)} ].
\end{gathered}
\end{equation}
Setting $\nabla_{\sigma_\eta^2}  Q(\bftheta|\bftheta^{(i)}) = 0$, we obtain closed form updates for the CP and NCP,
\begin{equation*}
{\sigma_\eta^2} = 
\begin{cases}
\frac{(\bfm_{0\bfzero}^{(i)} - \mu \bone)^T \Lambda (\bfm_{0\bfzero}^{(i)} - \mu \bone)+ \tr(\Lambda V_0^{(i)})}{n} & \text{if CP} , \\
\left(\frac{ (\bfy- \mu\bone)^T \bfm_{1\bone}^{(i)} }{\tr(V_1^{(i)})  + \bfm_{1\bone}^{(i)T} \bfm_{1\bone}^{(i)}}\right)^2 & \text{if NCP}.
\end{cases}
\end{equation*}
However, there are no closed form updates for $\sigma_\eta^2$ and $\phi$ for arbitrary values of $\{a, \bfw\}$. For $\phi$, closed form updates do not exist for the CP and NCP as well. We use the {\tt optimize} function from the {\tt Optim} package \citep{Mogensen2018} in {\tt Julia} to maximize $Q(\bftheta | \bftheta^{(i)})$ with respect to these parameters. Brent's method is used to find $\phi \in (-1,1)$ while the unconstrained L-BFGS method is used to find $\nu = \log \sigma_\eta^2$.

We are interested in finding values of $\{a, \bfw\}$ that optimize the rate of convergence of the EM algorithm. Define the information matrices as $I_\obs(\bftheta) =-  \nabla_\bftheta^2 \log p(Y_\obs | \bftheta)$, 
\[
I_{\augm}(\bftheta) = -\E_{Y_\mis | Y_\obs,\bftheta}[ \nabla_\bftheta^2  \log p(Y_\augm| \bftheta)],
\]
and $I_\mis(\bftheta)  = I_{\augm}(\bftheta) - I_\obs(\bftheta)$, where $\nabla_\bftheta^2(\cdot) = \partial^2(\cdot)/\partial \bftheta \partial\bftheta^T$. In a neighborhood of $\bftheta^*$, 
\[
\bftheta^{(i+1)} - \bftheta^* = J(\bftheta^*) (\bftheta^{(i)} - \bftheta^* ), 
\] 
where $J(\bftheta^*)$, the matrix rate of convergence of the EM algorithm, is given by
\[
J(\bftheta^*)= I - I_{\augm}(\bftheta^*)^{-1} I_{\obs}(\bftheta^*).
\]
The actual observed (global) rate of convergence is
\[
r = \lim_{i \rightarrow \infty} \| \bftheta^{(i+1)} - \bftheta^* \|/ \\\bftheta^{(i)} - \bftheta^*\|,
\]
where $\| \cdot\|$ is any norm on $\mathbb{R}^d$. Under some regularity conditions, $r$ is given by the largest eigenvalue of $J(\bftheta^*)$, and a larger value indicates slower convergence. If $I_\obs(\bftheta^*)$ is positive semidefinite, then $r \in [0,1]$. The EM algorithm converges rapidly ($r$ is close to 0) if the observed information is close to the augmented information. Since $I_{\obs}(\bftheta^*)$ depends only on the observed data and is independent of the parametrization, the rate of convergence can be optimized by minimizing $I_{\augm}(\bftheta^*)$ with respect to $\{a, \bfw\}$ in the sense of a positive semidefinite ordering \cite[see Theorem 1 of][]{Meng1997}. Let $I_{\theta_i, \theta_j}$ denote the $(i,j)$ element in $I_\augm(\bftheta^*)$. We first consider cases where only $\mu$ or $\sigma_\eta^2$ is unknown (matrix and actual observed rate of convergence are identical), followed by the case where all parameters are unknown. The proofs of all results in Sections \ref{sec_loc}, \ref{sec_scale} and \ref{sec_all} are given in the Supplement \citep{Tan2022}.

\subsection{Unknown location parameter} \label{sec_loc}
Suppose $\sigma_\eta^2$, $\phi$ and $\sigma_\epsilon^2$ are known and only $\mu$ is unknown. The EM algorithm for this case (Algorithm 1) alternately updates $\bfm_{a\bfw}^{(i)}$ and $\mu$ as in \eqref{eq_conddistn} and \eqref{update1}. 

\begin{theorem} \label{thm_alg1}
Rate of convergence of Algorithm 1 is 
\begin{equation*} 
\rho(\bfw)^T V_0 \rho(\bfw)/\tau(\bfw),
\end{equation*}
where $\tau(\bfw) = \sigma_\epsilon^{-2} \bfw^T \bfw + \sigma_\eta^{-2} \bar{\bfw}^T \Lambda \bar{\bfw}$ and $\rho(\bfw) = V_0^{-1} \bfw - \sigma_\eta^{-2} \Lambda \bone$. This rate is minimized to zero at
\begin{equation*}
\bfw^\opt = \sigma_\eta^{-2} V_0 \Lambda \bone \quad \text{or} \quad \bar{\bfw}^\opt = \sigma_\epsilon^{-2} V_0 \bone.
\end{equation*}
\end{theorem} 

From Theorem \ref{thm_alg1}, the rate of convergence of Algorithm 1 for the CP and NCP are $(\bone^T \Lambda V_0 \Lambda \bone)/(\sigma_\eta^2 \bone^T \Lambda \bone) $ and $(\bone^T V_0 \bone)/(n\sigma_\epsilon^2)$ respectively, which are strictly positive. However, when $\bfw = \bfw^\opt$, Algorithm 1 {\it converges instantly} to 
\[
\mu^* = \frac{\bfy^T S^{-1} \bone}{\bone^T S^{-1} \bone}.
\]
This can be seen by plugging $\bfw^\opt$ into the update in \eqref{update1}. Moreover, it is possible to compute $\bfw^\opt$ in advance for Algorithm 1 as $\bfw^\opt$ does not depend on $\mu$, and the rate of convergence is also independent of $a$, so that $a$ can be set to any convenient value. 

We investigate the range of elements in $\bfw^\opt$ and their dependence on $\sigma_\eta^2$, $\phi$ and $\sigma_\epsilon^2$ by deriving an explicit expression for $\bfw^{\opt}$ using results on the inverse of tridiagonal matrices in \cite{Tan2019}. From the expression of $\bfw^{\opt}$ in Theorem S1 (see Supplement \citep{Tan2022}), $\bfw^\opt$ depends on $\phi$, and on $\sigma_\eta^2$ and $\sigma_\epsilon^2$ only through the signal-to-noise ratio $\gamma$. Corollary \ref{cor1} presents bounds for $\bfw^\opt$ which are tight when $\phi=0$, showing that $0 < w_t^\opt < 1$ if $\phi \in [0,1)$ and $0 < w_t^\opt < 2$ if $\phi \in (-1,0)$. This is unlike normal hierarchical models, where $w_t^\opt$ always lies in $[0,1]$ \citep{Papaspiliopoulos2003}. From Corollary \ref{cor2}, the location centered parametrization ($\bfw = \bf0$) is increasingly preferred as the persistence $\phi$ and signal-to-noise ratio $\gamma$ increase when $0<\phi<1$. However, if $-1 <\phi <0$, $w_t^\opt$ may not be strictly decreasing with either $\phi$ or $\gamma$. Figure \ref{fig_woptloc} shows the values of elements in $\bfw_\opt$ when $n=10$ for different values of $\gamma$ and $\phi$. 

\begin{corollary}\label{cor1}
For $t=1, \dots, n$, 
\begin{equation*}
\begin{gathered}
0 < 1- B_1 \leq w_t^\opt \leq 1 - B_2 < 1
\quad (0\leq \phi < 1), \\
0 < 1 - B_2 \leq w_t^\opt \leq 1 + B_2 - 2B_1 < 2
\quad (-1 < \phi < 0),
\end{gathered}
\end{equation*}
where $B_1 = \gamma/ \{(1-\phi)^2 + \gamma\}$ and $B_2 = \gamma/(1-\phi^2 + \gamma)$. 
\end{corollary}

\begin{corollary}\label{cor2}
If $0 < \phi < 1$, each element of $\bfw^\opt$ decreases strictly as $\phi$ and the signal-to-noise ratio $\gamma$ increase. As $\phi$ approaches 1, $\bfw^\opt$ approaches $\bf0$.
\end{corollary}

\begin{figure}[htb!]
\includegraphics[width=0.48\textwidth]{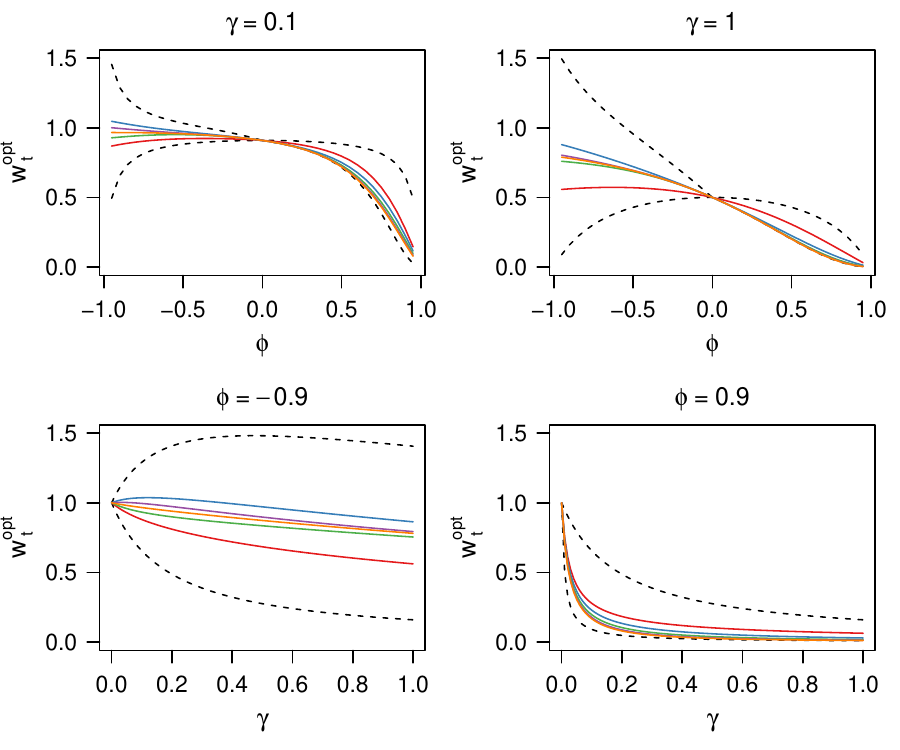}
\caption{Plots of $\bfw^\opt$ for different values of $\phi$ and $\gamma$ when $n=10$. Each solid line represents an element in $\bfw^\opt$. Only 5 distinct lines are seen as $\bfw^\opt$ is symmetric. Dashed lines represent bounds in Corollary \ref{cor1}.}
\label{fig_woptloc}
\end{figure}

From Theorem \ref{thm_alg1}, instant convergence is achievable only if the portion of $\mu$ subtracted from $x_t$ in \eqref{PNCP} is allowed to vary with $t$. However, Theorem \ref{thm_alg1_largen} shows that for large $n$, the convergence rate of Algorithm 1 goes to zero even if $w_t$ is common for all $t$, that is, if we restrict $\bfw = w \bone$. The optimal value of $w^\opt$ in Theorem \ref{thm_alg1_largen} is easy to compute and store, and can be used in place of that in Theorem \ref{thm_alg1} when $n$ is large.
\begin{theorem} \label{thm_alg1_largen}
If $\bfw = w\bone$, the rate of convergence of Algorithm 1 is optimized at $w^\opt =1/(1+u)$, where $u = n\gamma /(\bone^T \Lambda \bone)$. As $n \rightarrow \infty$, $w^\opt \rightarrow 1/(1+u^*)$ where $u^* = \gamma/(1-\phi)^2$ and the optimal rate of convergence of Algorithm 1 goes to zero.
\end{theorem}

\subsection{Unknown scale parameter}\label{sec_scale}
Next, suppose $\mu$, $\phi$ and $\sigma_\epsilon^2$ are known and only $\sigma_\eta^2$ is unknown. We assume $\mu \neq 0$ as this is the case of interest. The EM algorithm for this case (Algorithm 2), alternately updates $\bfm_{a\bfw}$ and $V_a$ at the E-step as in \eqref{eq_conddistn} and $\sigma_\eta^2$ at the M-step.

\begin{theorem}\label{thm_Alg2_opt_aw}
If $\mu \neq 0$, the rate of convergence of Algorithm 2 is jointly minimized at
\begin{equation} \label{eq_awopt}
a^\opt  = 1- \frac{\tr(V_0)}{n \sigma_\epsilon^2}, \quad
\bar{\bfw}^\opt = \frac{1}{\mu} \Big( \frac{2 V_0 \Lambda}{a^\opt \sigma_\eta^2} - I \Big)\bfm_{0\bone},
\end{equation}
where $\sigma_\eta^2 = \sigma_\eta^{2*}$. 
\end{theorem}

From Theorem \ref{thm_Alg2_opt_aw}, $\{a^\opt, \bfw^\opt\}$ depend on the unknown $\sigma_\eta^{2*}$, while $\bfw^\opt$ also depends on observed data $\bfy$. Thus, even if the true $\sigma_\eta^2$ is known, $\bfw^\opt$ will still vary across datasets generated from the AR(1) plus noise model due to sampling variability. We recommend updating $\{a^\opt, \bfw^\opt\}$ based on the latest estimate of $\sigma_\eta^2$. In Corollary \ref{cor3}, we show that $a^\opt \in (0,1)$. However, the same does not apply to $w_t^\opt$, which has been observed to be negative or exceed 1 in simulations. 

\begin{corollary} \label{cor3}
$a^\opt$ depends on $\phi$, and on $\sigma_\eta^2$ and $\sigma_\epsilon^2$ only through the signal-to-noise ratio $\gamma$. In addition, $a^\opt \in (0,1)$ and decreases strictly as $\gamma$ increases.
\end{corollary}

\begin{corollary} \label{cor4}
For $\bfw^\opt$ in \eqref{eq_awopt}, $\E(\bfw^\opt) = \bone$ and $\cov(\bfw^\opt) = AS^{-1} A/\mu^2$, where $A = 2V_0/a^\opt - \sigma_\eta^2 \Lambda^{-1}$. 
\end{corollary}

\begin{theorem} \label{thm_ahat}
As $n \rightarrow \infty$, $a^\opt$ converges to
\begin{equation} \label{ahatopt}
\hat{a}^\opt = 1- \gamma/ \sqrt{\{(1-\phi)^2 + \gamma\}\{(1+\phi)^2 + \gamma\}}.
\end{equation}
\end{theorem}

Theorem \ref{thm_ahat} derives a large-sample estimate $\hat{a}^\opt$ of $a^\opt$, which lends insight on how $a^\opt$ varies with $\gamma$ and $\phi$. From \eqref{ahatopt}, $\hat{a}_\opt$ is symmetric about $\phi = 0$ and decreases strictly with $\gamma$ as $\nabla_\gamma \hat{a}^\opt < 0$, but the relationship of $\hat{a}_\opt$ with $|\phi|$ is not monotone. Figure \ref{fig_aopt} shows the values of $a^\opt$ and $\hat{a}^\opt$ in different parameter settings, and $a^\opt$ is indistinguishable from $\hat{a}^\opt$ for $n \geq 200$. In addition, $a^\opt$ decreases as $\gamma$ increases, is symmetric about $\phi=0$, and varies with $\phi$ in the form of a (sometimes inverted) U-shape. We can compute $a^\opt$ efficiently using a Kalman filter and smoother using $\mathcal{O}(n)$ time steps, but $\hat{a}^\opt$ is even easier to compute and can be used in place of $a^\opt$ for large $n$.
\begin{figure}[htb!]
\includegraphics[width=0.48\textwidth]{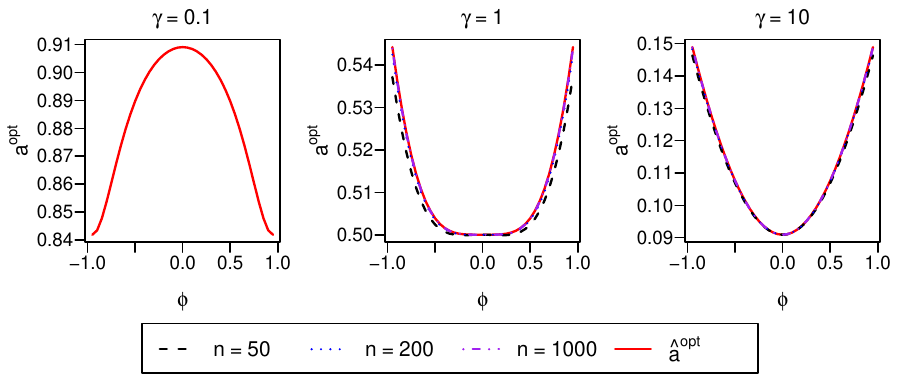}
\caption{Values of $a^\opt$ and $\hat{a}^\opt$ under different parameter settings.}
\label{fig_aopt}
\end{figure}

To investigate the dependence of $\bfw^\opt$ on $\phi$ and signal-to-noise ratio $\gamma$, we compute $\bfw^\opt$ for 1000 data sets simulated under different parameter settings. We set $\mu = 1$, $\sigma_\epsilon^2 = 0.1$, $\sigma_\eta^2 \in \{0.01, 0.1, 1\}$, $\phi \in \{-0.95, -0.90, \ldots, 0.95\}$ and $n\in \{50, 200, 1000\}$. Figure \ref{fig_wopt} shows the mean, 5th and 95th quantiles of the first element of $\bfw^\opt$ evaluated across 1000 simulated data sets. Plots of other elements look similar and are not shown. The mean is approximately 1 as is consistent with Corollary \ref{cor4}, but the variance increases with $\gamma$ and as $|\phi| \rightarrow 1$. The behavior of $\bfw^\opt$ does not appear to vary with $n$ and thus we cannot find an asymptotic estimate of $\bfw^\opt$ that is valid for large $n$ as we have done for $a^\opt$ in Theorem \ref{thm_ahat}.
\begin{figure}[htb!]
\includegraphics[width=0.48\textwidth]{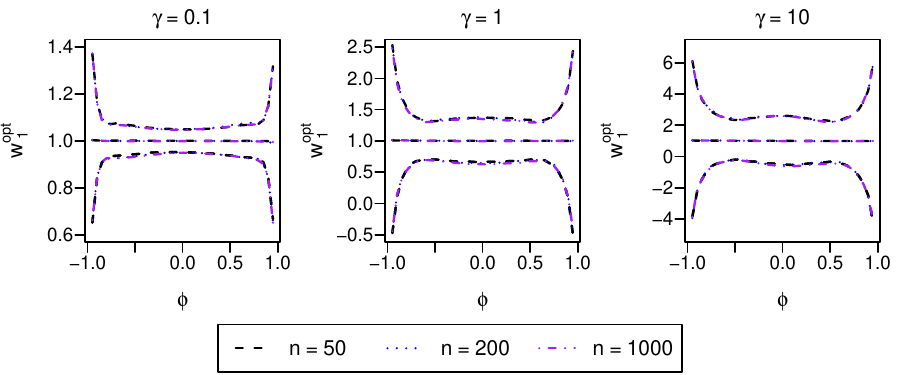}
\caption{Mean, 5th and 95th quantiles of $w_1^\opt$ over 1000 simulations.}
\label{fig_wopt}
\end{figure}

Theorem \ref{thm_aymp_converge_rate_Alg2} gives the approximate rate of convergence of Algorithm 2 for large $n$ and Figure \ref{fig_asymprate} shows how this rate varies for different combinations of $\gamma$ and $\phi$. The asymptotic rate of convergence is always less than 0.5 and symmetric about $\phi=0$. Algorithm 2 converges faster for smaller $|\phi|$ and larger signal-to-noise ratio $\gamma$. It converges most slowly when $|\phi|$ is close to 1 and $\gamma$ is close to 0.

\begin{theorem} \label{thm_aymp_converge_rate_Alg2}
The asymptotic rate of convergence of Algorithm 2 as $n \rightarrow \infty$ is approximately $(c - \kappa_0)/(2c-\kappa_0)$, where $\kappa_0=\sqrt{c^2-4}$ and $c = (1+\gamma +\phi^2)/|\phi|$. This rate decreases as $\gamma$ increases and $|\phi|$ decreases.
\end{theorem}
\begin{figure}[htb!]
\includegraphics[width=0.3\textwidth]{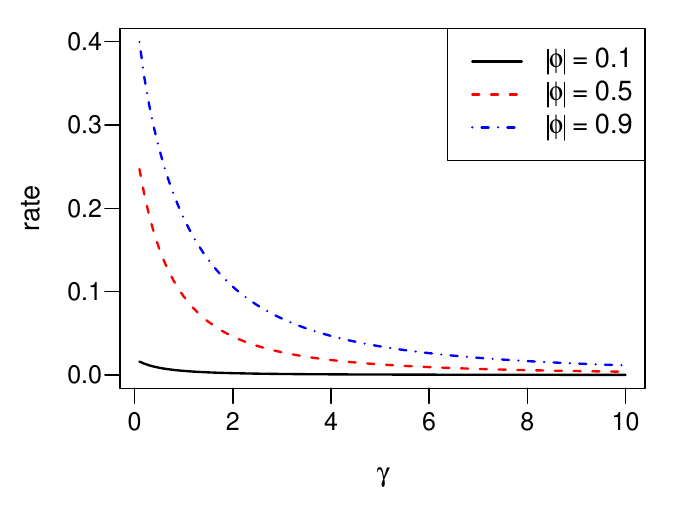}
\caption{Asymptotic rate of convergence of Algorithm 2.}
\label{fig_asymprate}
\end{figure}

\subsection{All parameters are unknown} \label{sec_all}
When all four parameters are unknown, the augmented information matrix $I_\augm(\bftheta^*)$ is given by
\begin{equation*}
\begin{bmatrix}
I_{ \mu, \mu} (\bfw) & I_{\mu, \sigma_\eta^2}(a, \bfw) & I_{ \mu, \phi} (\bfw) & I_{\mu, \sigma_\epsilon^2} (\bfw)  \\
\cdot & I_{\sigma_\eta^2, \sigma_\eta^2}(a, \bfw) & I_{\sigma_\eta^2, \phi}(a, \bfw) & I_{\sigma_\eta^2, \sigma_\epsilon^2}(a, \bfw) \\
\cdot & \cdot & I_{ \phi, \phi} & I_{\phi, \sigma_\epsilon^2}   \\
\cdot & \cdot & \cdot & I_{ \sigma_\epsilon^2, \sigma_\epsilon^2}
\end{bmatrix}.
\end{equation*}
Each element is dependent only on the working parameters stated in brackets and explicit expressions are given in the Supplement \citep{Tan2022}. For instance, $I_{ \mu, \mu}$ is independent of $a$ while $I_{\sigma_\eta^2, \sigma_\eta^2}$ depends on both $a$ and $\bfw$. Note that $I_{\phi, \sigma_\epsilon^2} = 0$, and $I_{ \phi, \phi}$, $I_{\sigma_\epsilon^2, \sigma_\epsilon^2}$ are both independent of $\{a, \bfw\}$. We also prove that $I_{\mu, \phi}/n$, $I_{ \mu, \sigma_\epsilon^2}/n$, $I_{\sigma_\eta^2, \phi}/n$ and $I_{\sigma_\eta^2, \sigma_\epsilon^2}/n$ converge to 0 almost surely as $n \rightarrow \infty$. Hence these elements do not change substantially with $\{a, \bfw\}$ for large $n$, and it suffices to focus on 
\[
\widetilde{I_\augm}(\bftheta^*) = \begin{bmatrix} I_{ \mu, \mu} (\bfw) & I_{\mu, \sigma_\eta^2} (a, \bfw) \\
\cdot & I_{\sigma_\eta^2, \sigma_\eta^2} (a, \bfw) \end{bmatrix},
\]
when optimizing the EM algorithm convergence rate. 

From Sections \ref{sec_loc} and \ref{sec_scale}, the parametrization optimal for inferring $\mu$ and $\sigma_\eta^2$ are different, and it is impossible to find $\{a, \bfw\}$ that jointly minimizes $I_{ \mu, \mu}$ and $I_{\sigma_\eta^2, \sigma_\eta^2}$. Hence when all parameters are unknown, we employ the {\em alternating expectation-conditional maximization} algorithm, which inserts an E-step before each conditional update, thus allowing the data augmentation scheme to vary across parameters. This enables us to use the optimal augmentation scheme for each parameter while conditioning on the rest, at the cost of more computation per iteration. To minimize the number of additional E-steps, we group $\{\phi,\sigma_\epsilon^2\}$ with $\sigma_\eta^2$ and perform only one E-step based on the optimal parametrization for $\sigma_\eta^2$ before the conditional updates of $\{\sigma_\eta^2, \phi, \sigma_\epsilon^2\}$. This is feasible as any convenient parametrization can be used for $\{\phi,\sigma_\epsilon^2\}$. The updates of $\{\phi, \sigma_\epsilon^2\}$ is a joint update as $\phi$ and $\sigma_\epsilon^2$ are independent in $Q(\bftheta | \bftheta^{(i)})$. After that, we update $\mu$ using its optimal parametrization. Instead of performing the E-step and M-step separately, we note that the update for $\mu$ in \eqref{update1} for $\bfw^\opt = \sigma_\eta^{-2} V_0 \Lambda \bone$ simplifies to $(\bfy^T S^{-1} \bone)/(\bone^T S^{-1} \bone)$, and $S^{-1} \bone = \sigma_\epsilon^{-2} \bfw^\opt$. Hence we can compute $\bfw^\opt$ followed by the $\mu$ update directly without computing $\bfm_{a\bfw}$. This choice of $\bfw$ corresponds to having no missing data since the $\mu$ update is that obtained by maximizing $p(\bfy | \bftheta)$ directly with respect to $\mu$. Algorithm 3 (for inferring all the parameters) is outlined below, and each iteration has two cycles. Algorithm 1 (for inferring $\mu$ only) omits step 2 while Algorithm 2 (for inferring $\sigma_\eta^2$ only) omits step 3 and updates of $\{\phi,\sigma_\epsilon^2\}$ in step 2(c). We say that Algorithms 1, 2 and 3 are using the partially noncentered parametrization (PNCP) as both the mean and scale are partially noncentered.

\begin{algorithm}
\renewcommand{\thealgorithm}{}
\floatname{algorithm}{}
\caption{Algorithm 3}
\begin{normalsize}
\begin{flushleft}
Initialize $\bftheta^{(0)}$, $i = 0$ and $r=1$. While $i \leq M$ and $r <\delta$,
\begin{enumerate}
\item $i \leftarrow i+1$.
\item \begin{enumerate}
\item Set $a=1 - \tr(V_0)/(n\sigma_\epsilon^2)$ and $\bar{\bfw} = \{2V_0 \Lambda /(a\sigma_\eta^2) - I\} \bfm_{0\bone}/\mu$.
\item Update $\bfm_{a\bfw}^{(i)}$ and $V_a^{(i)}$ as in \eqref{eq_conddistn}.
\item Update $\sigma_\eta^2$, $\sigma_\epsilon^2$ and $\phi$. 
\end{enumerate}
\item 
\begin{enumerate}
\item Set $\bfw = \sigma_\eta^{-2} V_0 \Lambda \bone$. 
\item Update $\mu = (\bfy^T \bfw)/(\bone^T  \bfw)$.
\end{enumerate}
\item Compute $L^{(i)} = \log p(\bfy | \bftheta^{(i)})$. \\
If $i \geq 2$, compute $s = (L^{(i)} - L^{(i-1)})/L^{(i-1)}$.
\end{enumerate}
\end{flushleft}
\end{normalsize}
\end{algorithm}

\subsection{Initialization and diagnosing convergence} \label{sec: diagnose convergence}
The EM algorithm can be sensitive to initialization and there is also a risk of getting stuck in local modes. To aid convergence, we initialize using the strategy below. Recall that $\bfy \sim \N(\mu\bone, S)$, where $S = \sigma_\epsilon^2 I + \sigma_\eta^2 \Lambda^{-1}$. Let $\gamma_h$ and $\hat{\gamma}_h$ denote the autocovariance and sample autocovariance respectively at lag $h$.  Then
\begin{align*}
\gamma_0 &= \sigma_\epsilon^2 + \sigma_\eta^2/(1-\phi^2),   && t=1, \dots, n, \\
\gamma_h &= \phi^h \sigma_\eta^2/(1- \phi^2),                  && 1 \leq h \leq n-t.
\end{align*}
We initialize $\mu$ as the sample mean $\bar{y} =  \sum_{t=1}^n y_i/n$. Since $\gamma_2/\gamma_1 = \phi$, \cite{Shumway2017} suggest initializing $\phi$ as $\hat{\gamma}_2/ \hat{\gamma}_1$. This strategy works well if $|\phi|$ is close to 1 but can be poor if $|\phi|$ is close to 0. Let $\rho_1 = \gamma_1/\gamma_0$ and $\hat{\rho}_1 = \hat{\gamma}_1/\hat{\gamma}_0$. As $|\phi| > |\rho_1|$, we consider values of $\phi$ taking the sign of $\hat{\gamma}_1$, such that $|\phi| \in \{0.1, 0.2, \dots, 0.9\}$ and $|\phi| > |\hat{\rho}_1|$. If this set is empty (e.g. if $|\hat{\rho}_1| \geq 0.9$), we consider $\phi = (\hat{\rho}_1 + \text{sign}(\hat{\rho}_1))/2$. For each plausible value of $\phi$, we compute ${\sigma_\eta^2} = \hat{\gamma}_1(1 - \phi^2)/ \phi$ and $\sigma_\epsilon^2 = \hat{\gamma}_0 - \hat{\gamma}_1/ \phi$. Then we evaluate the log-likelihood at each $\bftheta$ and choose the one that maximizes the likelihood as initial values. For diagnosing convergence, we consider the relative increment $s$ in the log-likelihood (defined in Algorithm 3), and terminate the algorithm when $s$ is less than a tolerance $\delta$ or if the maximum number of iterations $M$ is reached. We set $\delta = 10^{-9}$ and $M=10^5$ in all experiments unless stated otherwise. Efficient computation of the log-likelihood and optimal values of $\{a, \bfw\}$ using the Kalman filter are discussed in the Supplement \citep{Tan2022}. The cost per iteration of Algorithm 3 is $\mathcal{O}(n)$, where $n$ is the number of observations as the cost of computing the working parameters, log likelihood and parameter updates all scale linearly in $n$.

\section{Experimental results for EM algorithm}
We compare Algorithms 1--3 with the expectation-conditional maximization algorithms using the CP or NCP. For the expectation-conditional maximization algorithms, only one E-step is used at the beginning of each iteration followed by the CM-steps. In contrast, Algorithm 3 requires more computation in each iteration to update working parameters and find numerically the value of $\sigma_\eta^2$ that maximizes $Q(\bftheta | \bftheta^{(i)})$. To make Algorithms 2 and 3 more competitive, the updates in steps 2a and 3 are performed only in the first five iterations and thereafter only at $i=1000, 2000, \dots$. After convergence, we perform a final update of $\mu$ and the log-likelihood. These modifications do not affect the convergence rate of Algorithms 2 and 3 but are helpful in reducing the computation cost.

\subsection{Simulations}
First, we simulate data from the AR(1) plus noise model by setting $n=10^4$, $\mu = -1$, $\sigma_\epsilon^2 = 0.1$, $\phi \in \{-0.9, -0.6, -0.3, 0.3, 0.6, 0.9\}$ and $\sigma_\eta^2 \in \{0.01, 0.1, 1\}$. Twenty datasets are generated in each setting. 

Algorithm 1 considers $\mu$ as the only unknown. When $\delta = 10^{-9}$, it is difficult to differentiate among different parametrizations as all algorithms converge on average in less than five iterations and the averaged estimates of $\mu$ are almost identical. To magnify the differences, we reduce $\delta$ to $10^{-11}$. The average number of iterations NCP and CP take to converge are shown in Figure \ref{Alg1expt}, whereas Algorithm 1 converges instantly. With this lower tolerance, NCP achieves the same averaged estimates of $\mu$ as CP and PNCP, but becomes increasingly inefficient as $\gamma$ increases and as $\phi$ approaches 1. On the other hand, CP is least efficient when $\gamma$ is small and $\phi = -0.9$, and improves as $\gamma$ increases and $\phi \rightarrow 1$. Hence the performance of NCP and CP are complementary. These observations can also be explained using Theorem \ref{thm_alg1_largen}, which states that $w^\opt  \rightarrow (1-\phi)^2/\{\gamma + (1-\phi)^2 \}$ as $n \rightarrow \infty$. From this expression, $w^\opt \rightarrow 0$ as $\phi \rightarrow 1$ and $\gamma \rightarrow +\infty$. Hence the CP is strongly preferred in these conditions. 
\begin{figure}[htb!]
\centering
\includegraphics[width=0.4\textwidth]{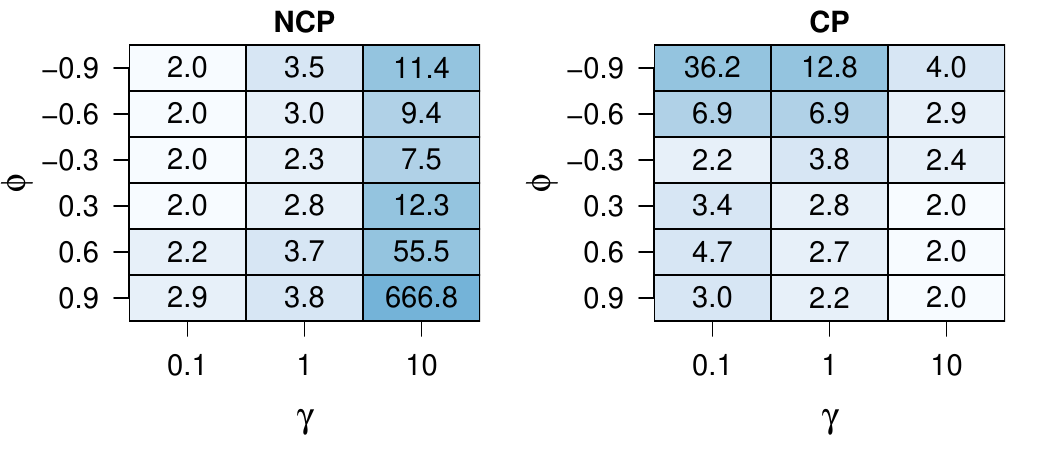}
\caption{(Unknown $\mu$) Average number of iterations required for convergence.}
\label{Alg1expt}
\end{figure}

For Algorithm 2 where $\sigma_\eta^2$ is the only unknown parameter, the average number of iterations required for each algorithm to converge are shown in Figure \ref{Alg2expt}. The average runtime in seconds is also given in brackets. The rate of convergence appears to be symmetric about $\phi=0$ for each parametrization. As is consistent with Theorem \ref{thm_aymp_converge_rate_Alg2}, PNCP becomes more efficient as $\gamma$ increases and $|\phi|$ decreases. The efficiency of CP worsens as $\gamma$ decreases and $|\phi| \rightarrow 0$, while NCP performs best when $\gamma$ is close to 1 and worsens as $|\phi|$ increases. The estimates of $\sigma_\eta^2$ obtained using different parametrizations are virtually identical except $\phi\in \{-0.9, -0.6, 0.9\}$ and $\gamma=10$. In these cases, NCP converges slowly and yields slightly different estimates from CP and PNCP.
\begin{figure}[htb!]
\centering
\includegraphics[width=0.48\textwidth]{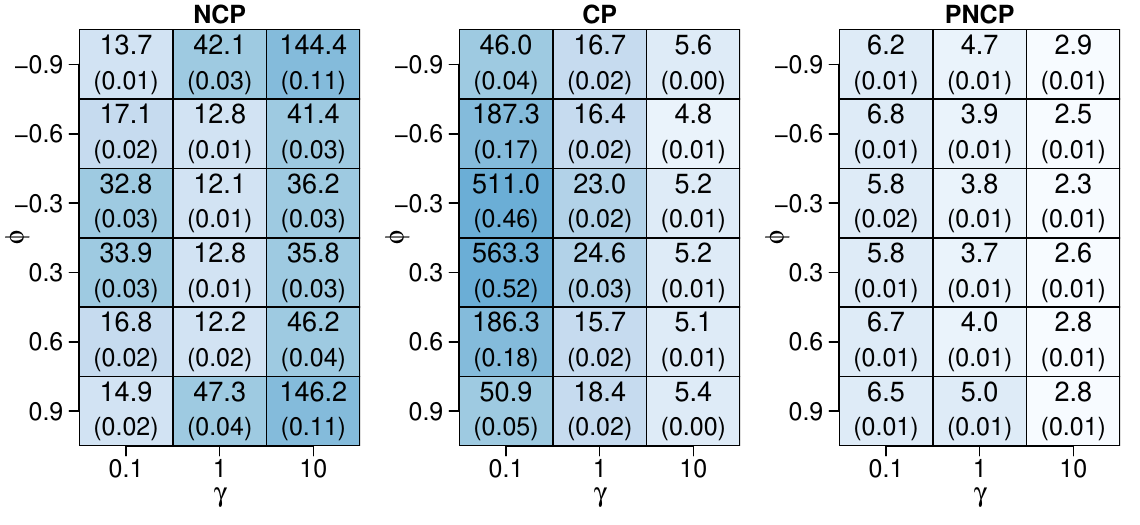}
\caption{(Unknown $\sigma_\eta^2$) Average number of iterations required for convergence, and average runtime (seconds) in brackets.}
\label{Alg2expt}
\end{figure}

Algorithm 3 treats all parameters as unknown. The average number of iterations required and runtime of each algorithm are shown in Figure \ref{Alg3chart}. 
\begin{figure}[tb!]
\centering
\includegraphics[width=0.48\textwidth]{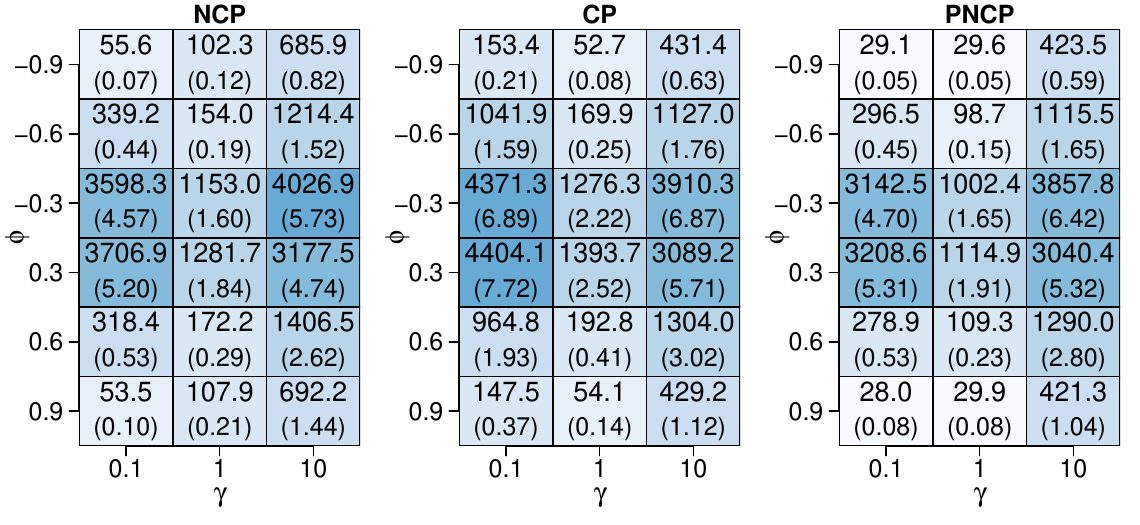}
\caption{(All parameters unknown) Average number of iterations required for convergence, and average runtime (seconds) in brackets.}
\label{Alg3chart}
\end{figure}
Again, there is symmetry about $\phi=0$. In each setting, PNCP converges using the least number of iterations and achieves a higher log-likelihood on average. Its average runtime is always close to, if not faster than the better of NCP and CP. Table \ref{Alg3speedup} shows that PNCP is able to provide speedup over CP and NCP in most cases, of up to 4.9 times. 
\begin{table*}[htb!]
\caption{(All parameters unknown) Speedup factor in average runtime provided by PNCP over NCP and CP.}
\label{Alg3speedup}
\begin{tabular}{@{}ccccccccccccccccccc@{}}  \hline
$\phi$ &\multicolumn{3}{c}{$-$0.9} &\multicolumn{3}{c}{$-$0.6} &\multicolumn{3}{c}{$-$0.3} &\multicolumn{3}{c}{0.3} &\multicolumn{3}{c}{0.6} &\multicolumn{3}{c}{0.9} \\
$\gamma$ & 0.1 & 1 & 10 & 0.1 & 1 & 10 & 0.1 & 1 & 10 & 0.1 & 1 & 10 & 0.1 & 1 & 10 & 0.1 & 1 & 10 \\   \hline
NCP & 1.3 & 2.3 & 1.4 & 1.0 & 1.3 & 0.9 & 1.0 & 1.0 & 0.9 & 1.0 & 1.0 & 0.9 & 1.0 & 1.3 & 0.9 & 1.4 & 2.6 & 1.4 \\
CP &  4.2 & 1.6 & 1.1 & 3.5 & 1.7 & 1.1 & 1.5 & 1.3 & 1.1 & 1.5 & 1.3 & 1.1 & 3.6 & 1.8 & 1.1 & 4.9 & 1.6 & 1.1 \\ \hline
\end{tabular}
\end{table*}
These properties make PNCP a great option as it is always faster than the worst of CP and NCP, often able to provide speedup and achieves a higher log likelihood on average. In most cases, the estimates obtained using different parametrizations are practically identical (more details are given in the Supplement \citep{Tan2022}).

\subsection{Real data}
The first dataset contains IBM stock prices from 1962 to 1965 ($n=1008$) studied in \cite{Pal2017}, and it is available in the code bundle for the book on GitHub. The second dataset involves an industrial robot making a sequence of movements and the time series ($n=324$) records the distance in inches from a final target. It is available in the R package {\tt TSA} \citep{Cryer2008} as {\tt data(robot)}. We scale up all measurements in {\tt robot} up by 1000 as the magnitude of the original measurements are very small with a maximum of 0.0083. 

Tables \ref{tabibm} and \ref{tabrobot} show the results for the IBM data and robot data respectively. Iterations refers to the number of iterations required for the EM algorithm to converge, where convergence is determined using the criteria described in Section \ref{sec: diagnose convergence}. The runtime reported is in seconds and the log-likelihood is the maximum value of $\log p(\bfy|\bftheta)$ attained at convergence. Maximum likelihood estimates of $\bftheta$ at convergence are also given. 

For the IBM data, $\hat{\phi}$ is very close to 1, while $\hat{\gamma} = 44.275/0.135$ is large. From Figure \ref{Alg3chart}, we expect CP to perform better than NCP which is indeed the case. PNCP automatically leans towards CP and provides a speedup of about 2.2 compared to NCP. The maximum log-likelihood attained by NCP is lower than CP and PNCP, and the estimates of $\mu$,  $\sigma_\eta^2$ and $\sigma_\epsilon^2$ also differ noticeably. This is likely because the NCP is stuck at a poorer local mode, or it is converging so slowly and the relative increment in the log-likelihood per iteration is so small that the algorithm was terminated before full convergence is reached.
\begin{table}[ht]
\caption{IBM data.}
\label{tabibm}
\begin{tabular}{@{}lrrr@{}}
\hline
& NCP & CP & PNCP \\ \hline
Iterations & 23504 & 9036 & 9030 \\
Runtime (s) & 4.86 & 2.30 & 2.25 \\
Max. log-likelihood & $-$3346.113 & $-$3345.929 & $-$3345.929 \\
$\mu$ & 463.125 & 482.043 & 482.043 \\   
$ \sigma_\eta^2$ & 43.946 & 44.275 & 44.275 \\
$\phi$ & 0.995 & 0.995 & 0.995 \\
$\sigma_\epsilon^2$ &  0.153 & 0.135 & 0.135 \\ 
\hline
\end{tabular}
\end{table}

For the robot data, $\hat{\phi}$ is also close to 1, but $\hat{\gamma} = 0.209/5.062$ is small, and NCP is expected to perform better based on Figure \ref{Alg3chart}. Indeed, NCP converges faster than CP, while PNCP reduces the number of iterations of NCP further by about half. The estimates returned by NCP and PNCP are identical while CP differs slightly.
\begin{table}[ht]
\caption{Robot data.}
\label{tabrobot}
\begin{tabular}{@{}lrrr@{}}
\hline
& NCP & CP & PNCP \\   \hline
Iterations & 93 & 326 & 42 \\
Runtime (s) & 0.00 & 0.03 & 0.00 \\
Max. log-likelihood & $-$748.809 & $-$748.809 & $-$748.809 \\
$\mu$ & 1.486 & 1.486 & 1.486 \\   
$ \sigma_\eta^2$ & 0.209 & 0.210 & 0.209, \\
$\phi$ & 0.947 & 0.947 & 0.947 \\
$\sigma_\epsilon^2$ &  5.062 & 5.061 & 5.062 \\ 
\hline
\end{tabular}
\end{table}

Figure \ref{IBMrobotloglikelihood} tracks the log-likelihood at each iteration, and optimization of the working parameters has provided PNCP with a good headstart. Overall, PNCP provides an attractive alternative to NCP and CP as it automatically gravitates towards the better parametrization and is often able to outperform both in terms of speed and accuracy.
\begin{figure}[ht]
\includegraphics[width=0.48\textwidth]{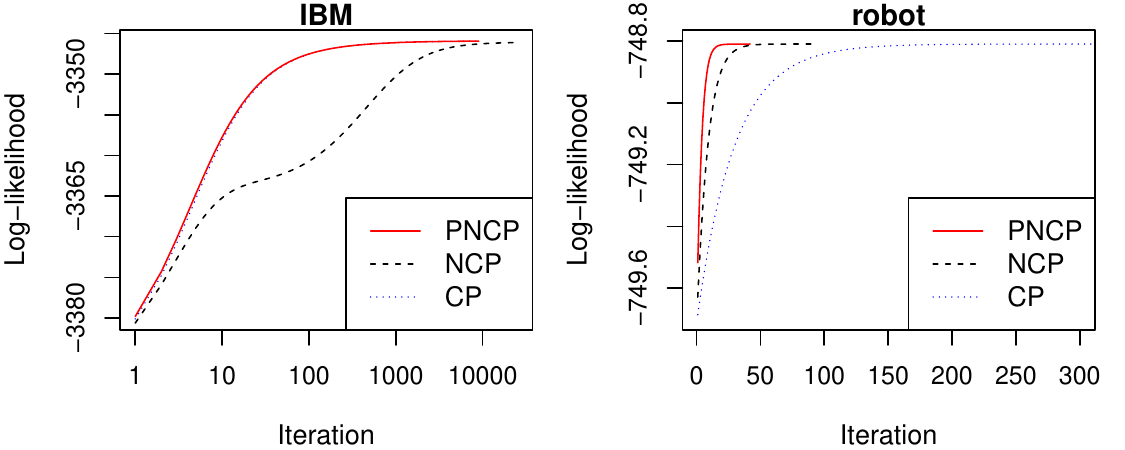}
\caption{Log-likelihood at each iteration.}
\label{IBMrobotloglikelihood}
\end{figure}

\section{Extensions}
We have proposed a data augmentation scheme for the AR(1) plus noise model in which working parameters $\{a, \bfw\}$ are introduced for rescaling and recentering. Optimal values of $\{a, \bfw\}$ that optimize the convergence of the EM algorithm in maximum likelihood estimation have been derived. However, the EM algorithm can also be used in the Bayesian framework to find the mode of a posterior distribution. Will the optimal values of $\{a, \bfw\}$ vary depending on the prior specified on $\bftheta$? It is known that the rate of convergence of the Gibbs sampler is closely linked to the EM algorithm. Will the proposed data augmentation scheme also improve convergence for the Gibbs sampler? We seek to address these issues in this section and also touch briefly on the connection to variational Bayes.

\subsection{EM Algorithm for finding posterior mode} \label{sec_EMBayesian}
Suppose a prior $p(\bftheta)$ is placed on $\bftheta$, and we wish to use the EM algorithm to find the posterior mode $\bftheta^\star$. That is, we wish to find $\bftheta$ maximizing $\log p(\bfy, \bftheta) = \log p(\bfy|\bftheta) + \log p(\bftheta)$. As before, let $\bfalpha$ denote the latent states dependent on $\{a, \bfw\}$. In this case,
\begin{equation*}
\begin{aligned}
Q(\bftheta|\bftheta^{(i)}) 
%&= \E_{\bfalpha|\bfy, \bftheta^{(i)}}[\log p(\bfy, \bfalpha, \bftheta)] \\
&= \E_{\bfalpha|\bfy, \bftheta^{(i)}}[\log p(\bfy, \bfalpha|\bftheta)] + \log p(\bftheta)
\end{aligned}
\end{equation*}
has an additional term $\log p(\bftheta)$. Hence the E-step is unchanged but the M-step must be modified. As $p(\bfalpha|\bftheta, \bfy) p(\bftheta|\bfy) = p(\bftheta|\bfalpha, \bfy)\; p(\bfalpha|\bfy)$, taking the logarithm, differentiating twice with respect to $\bftheta$ and taking expectation with respect to $p(\bfalpha|\bftheta, \bfy)$ leads to
\begin{multline*}
- \E_{\bfalpha|\bftheta, \bfy} [\nabla_\bftheta^2 \log p(\bfalpha|\bftheta, \bfy)] -\nabla_\bftheta^2 \log p(\bftheta|\bfy) \\
= - \E_{\bfalpha|\bftheta, \bfy} [\nabla_\bftheta^2 \log p(\bftheta|\bfalpha, \bfy)].
\end{multline*}
This yields the Bayesian analogue of the missing information principle, 
\[
I_\mis(\bftheta) + I_\obs^B(\bftheta) = I_\augm^B(\bftheta),
\]
where $I_\obs^B(\bftheta) = - \nabla_\bftheta^2 \log p(\bftheta|\bfy) = - \nabla_\bftheta^2 \log p(\bfy|\bftheta) - \nabla_\bftheta^2 \log p(\bftheta)$ and
\begin{equation*}
\begin{aligned}
I_\augm^B(\bftheta) 
&= - \E_{\bfalpha|\bftheta, \bfy} [\nabla_\bftheta^2 \log p(\bftheta|\bfalpha, \bfy)] \\
&= - \E_{\bfalpha|\bftheta, \bfy} [ \nabla_\bftheta^2 \log p(\bfy,\bfalpha |\bftheta)] - \nabla_\bftheta^2 \log p(\bftheta) \\
&= I_\augm(\bftheta) - \nabla_\bftheta^2 \log p(\bftheta).
\end{aligned}
\end{equation*}
To optimize the rate of convergence of the EM algorithm, we seek values of $\{a, \bfw\}$ that minimizes $I_\augm^B(\bftheta^\star)$. Since $p(\bftheta)$ is independent of $\{a, \bfw\}$, it suffices to minimize $I_\augm(\bftheta^*)$ with respect to $\{a, \bfw\}$. Hence the only difference is that $I_\augm(\bftheta)$ is evaluated at the posterior mode $\bftheta^\star$ instead of maximum likelihood estimate $\bftheta^*$. As $I_{\mu, \mu}$ is independent of $\mu$, the optimal value of $\bfw$ derived in Section \ref{sec_loc} remains valid. As for $I_{\sigma_\eta^2, \sigma_\eta^2}$, it is shown in the Supplement \citep{Tan2022} that the optimal values of $\{a, \bfw\}$ derived in Section \ref{sec_scale} miraculously remain valid for any prior distribution $p(\sigma_\eta^2)$. Thus we can still use the optimal values of $\{a, \bfw\}$ derived in Sections \ref{sec_loc} and \ref{sec_scale} to optimize the rate of convergence of the EM algorithm when finding the posterior mode of the AR(1) plus noise model.

\subsection{Gibbs sampler and variational Bayes}
It is well known that the convergence rates of the EM algorithm and Gibbs sampler are closely related. When the target distribution $p(\bftheta, \bfx | \bfy)$ is Gaussian, \cite{Sahu1999} showed that the rate of convergence of the Gibbs sampler that alternately updates $\bftheta$ and $\bfx$ is equal to that of the corresponding EM algorithm. Specifically, if the precision matrix of $p(\bftheta, \bfx | \bfy)$ is
\begin{equation*}
H = \begin{bmatrix}
H_{11} & H_{12}  \\  H_{12}^T & H_{22}
\end{bmatrix},
\end{equation*}
where $H_{11}$ and $H_{22}$ are blocks corresponding to $\bftheta$ and $\bfx$ respectively, then the common rate of convergence is $H_{11}^{-1} H_{12} H_{22}^{-1} H_{21}$. 

As illustration, suppose $\mu$ is the only unknown parameter and a flat prior $p(\mu) \propto 1$ is used. We show that the strategy of partially noncentering $\mu$ and results in Section \ref{sec_loc} can be transferred to the Gibbs sampler corresponding to Algorithm 1 given below. 
\begin{tabbing}
\qquad \enspace Initialize $\mu^{(0)}$. For $i=1, \ldots, N$, \\
\qquad \qquad Step 1. Sample $\bfalpha^{(i)}$ from $p(\bfalpha | \mu^{(i-1)}, \bfy)$.\\
\qquad \qquad Step 2. Sample $\mu^{(i)}$ from $p(\mu | \bfalpha^{(i)}, \bfy)$. 
\end{tabbing}
The conditional distribution $p(\bfalpha|\mu, \bfy)$ is given in \eqref{eq_conddistn} and
\begin{equation*}
\mu | \bfalpha, \bfy \sim \N( \tau(\bfw)^{-1} \{\sigma_\epsilon^{-2} \bfy^T \bfw - \sigma_\eta^a \bfalpha^T \rho(\bfw) \},  \tau(\bfw)^{-1} ).
\end{equation*}
The joint posterior density $p(\mu, \bfalpha | \bfy)$ is Gaussian with precision matrix $H$, where $H_{11} = \tau(\bfw)$, $H_{12} =  \sigma_\eta^a\rho(\bfw)^T$ and $H_{22} =  \sigma_\eta^{2a} V_0^{-1}$. By \cite{Sahu1999}, the rate of convergence of the Gibbs sampler is $\tau(\bfw)^{-1} \rho(\bfw)^T V_0\rho(\bfw)$, the same as that stated in Theorem \ref{thm_alg1}. Hence the Gibbs sampler converges instantly and produces independent draws when $\bfw = \bfw^\opt$. Alternatively, if we combine the updates in steps 1 and 2, the update of $\mu$ at the $i$th iteration is 
\begin{equation*}
\begin{aligned}
\mu^{(i)} &= \tau(\bfw)^{-1} \left[ \bfy^TS^{-1} \bone+\rho(\bfw)^T V_0\rho(\bfw) \mu^{(i-1)} \right.\\
& \left. \quad + \sqrt{\rho(\bfw)^T V_0\rho(\bfw) + \tau(\bfw)} Z \right],
\end{aligned}
\end{equation*}
where $Z \sim \N(0,1)$. The rate of convergence can be optimized by minimizing the autocorrelation at lag 1, $\tau(\bfw)^{-1} \rho(\bfw)^T V_0\rho(\bfw)$, and the result in Theorem \ref{thm_alg1} follows.

Next, we demonstrate that variational Bayes can also benefit from partial noncentering. Consider a variational Bayes approximation to $p(\mu, \bfalpha | \bfy)$ of the form $q(\bfalpha,\mu) = q(\bfalpha)q(\mu)$. Subjected to this product density restriction, the optimal $q(\bfalpha)$ and $q(\mu)$, obtained by minimizing the Kullback-Leibler divergence from $p(\mu, \bfalpha | \bfy)$ to $q(\bfalpha,\mu)$, are $\N(m_\bfalpha^q, \sigma_\eta^{-2a}V_0)$ and $\N(m_\mu^q, \tau(\bfw)^{-1})$ respectively \cite[see, e.g.][]{Ormerod2010}, where 
\begin{equation*}
\begin{gathered}
m_\bfalpha^q = \sigma_\eta^{-a}V_0\{ \sigma_\epsilon^{-2}\bfy- \rho(\bfw) m_\mu^q \}, \\
m_\mu^q = \tau(\bfw)^{-1} \{\sigma_\epsilon^{-2} \bfy^T\bfw - \sigma_\eta^a \rho(\bfw)^T m_\bfalpha^q \}.
\end{gathered}
\end{equation*}
The variational Bayes algorithm thus iterates between updating $m_\bfalpha^q$ and $m_\mu^q$. Combining the two updates, we have at the $i$th iteration,
\begin{equation*}
{m_\mu^q}^{(i)} = \tau(\bfw)^{-1} (\bfy ^T S^{-1} \bone+  \rho(\bfw)^T V_0\rho(\bfw) \, {m_\mu^q}^{(i-1)} ).
\end{equation*}
Hence the rate of convergence is also given by $\tau(\bfw)^{-1} \rho(\bfw)^T V_0\rho(\bfw)$, which is minimized to zero when $\bfw = \bfw^\opt$. Moreover, $\tau(\bfw^\opt) = (\bone^TS^{-1}\bone)^{-1}$. Hence $m_\mu^q = (\bfy^T S^{-1} \bone)/(\bone^TS^{-1}\bone) $ and $q(\mu)$ is able to capture the true marginal distribution of $\mu$.

In summary, when $\mu$ is the only unknown parameter, the rate of convergence of the EM algorithm, Gibbs sampler and variational Bayes are all equal to $\tau(\bfw)^{-1} \rho(\bfw)^T V_0\rho(\bfw)$, which is optimized at $\bfw^\opt$. This outcome is in line with the result of \cite{Tan2014}, who showed the equivalence of the rates of convergence between the EM algorithm, Gibbs sampler and variational Bayes, when the target density is Gaussian. Intuitively, $\bfw^\opt$ minimizes the fraction of missing information for the EM algorithm, and the autocorrelation at lag 1 for the Gibbs sampler. For variational Bayes, $\bfw^\opt$, besides optimizing the rate of convergence, also produces a more accurate posterior approximation. This is because $\cov(\bfalpha, \mu|\bfy)$ is zero when $\bfw = \bfw^\opt$, and $q(\bfalpha, \mu) = q(\bfalpha) q(\mu)$ is then an accurate reflection of the dependence structure between $\bfalpha$ and $\mu$.

\section{Non-Gaussian state space models}
Next, we consider MCMC algorithms for Bayesian inference of some nonlinear and non-Gaussian state space models and demonstrate how the optimal augmentation schemes derived for the AR(1) plus noise model can be utilized to improve the mixing and convergence of the MCMC sampler. We focus on the stochastic volatility (SV) and stochastic conditional duration models. 

\subsection{Stochastic volatility model} \label{sect_SV model}
The returns of financial time series often have variances that vary with time, and the SV model seeks to capture this behavior by modeling the log of the squared volatility as a hidden AR(1) process. The observation equation is 
\begin{align*}
y_t &= \e^{x_t/2} \epsilon_t,  && \epsilon_t \sim \N(0,1),\; t \geq 1,
\end{align*} 
where $\e^{x_t/2}$ represents the volatility, while $(x_t)$ is defined in \eqref{AR(1)}. The $(\epsilon_t)$ and $(\eta_s)$ are independent sequences and the unknown parameter $\bftheta = (\mu, \sigma_\eta^2, \phi)^T$. We can transform the SV model into a linear model by writing
\begin{equation*}
\tilde{y}_t = \log y_t^2 = x_t + \log \epsilon_t^2,
\end{equation*}
where $\log \epsilon_t^2$ follows a log $\chi_1^2$ distribution with mean, $\psi(0.5) - \log 0.5 \approx -1.27$, and variance, $\pi^2/2 \approx 4.93$. \cite{Harvey1994} approximate the distribution of $\log \epsilon_t^2$ using $\N(-1.27, 4.93)$, thereby recovering the AR(1) plus noise model with $y_t$ replaced by $\tilde{y}_t +1.2704$ and $\sigma_\epsilon^2 = 4.93$. They then estimated $\bftheta$ by maximizing the quasi-likelihood using Kalman filtering. However, \cite{Kim1998} showed that the distribution of $\log \epsilon_t^2$ is highly skewed and poorly approximated by a normal distribution, and they approximate it using a 7-component mixture of normal densities instead. Here we consider the improved mixture approximation with $K=10$ components by \cite{Omori2007}. Let  $r_t \in \{1, \dots, K\}$ be mixture component indicators such that $\P(r_t=k) = p_k$ and $\log \epsilon_t^2| r_t = k \sim \N(m_k, s_k^2)$ for $t=1, \dots, n$. Conditional on the mixture indicator $r_t$, 
\begin{align*}
\tilde{y}_t -  m_{r_t}&= x_t  + s_{r_t} \varepsilon_t, &&\varepsilon_t \sim \N(0,1),
\end{align*}
and the AR(1) plus noise model is recovered with $y_t$ replaced by $\tilde{y}_t -  m_{r_t}$ and $\sigma_\epsilon$ replaced by $s_{r_t}$.

\subsection{Stochastic conditional duration model}
The timing of transactions is central to the study of the microstructure of financial markets. As time intervals between transactions are highly irregular, \cite{Engle1998} formulated the autoregressive conditional duration model to analyze such data, treating arrival times as a point process and allowing the conditional intensity to be dependent on past durations. As an extension, \cite{Bauwens2004} propose the stochastic conditional duration model, which treats the log of the conditional mean of durations as a latent AR(1) process. The duration $y_t$ has positive support, and the observation equation is
\begin{align*}
y_t &= \e^{x_t} \epsilon_t, &&t \geq 1, 
\end{align*}
where $\epsilon_t$ follows the exponential distribution with mean 1, while $(x_t)$ is defined as in \eqref{AR(1)}. Other possible choices for the distribution of $\epsilon_t$ include Weibull and Gamma \citep{Strickland2006}. We similarly transform this model into a linear one by writing 
\begin{equation*}
\tilde{y}_t = \log y_t = x_t + \log \epsilon_t,
\end{equation*}
where $\log \epsilon_t$ has mean, $\psi(1) \approx -0.58$, and variance, $\pi^2/6 \approx 1.64$. The distribution of $\log \epsilon_t$ is also poorly approximated by $\N(-0.58, 1.64)$ and we approximate it using a mixture of normal densities. The 5-component mixture approximation for $\log \epsilon_t$ provided by \cite{Carter1997} is rather imprecise and we consider instead the 10-component mixture approximation for the standard Gumbel (type I extreme value) distribution found by \cite{Fruhwirth2007}, which is fitted by minimizing the Kullback-Leibler distance and shown to improve the accuracy of auxiliary mixture sampling when used in a Metropolis-Hastings algorithm. \cite{Fruhwirth2010} further illustrates the use of auxiliary mixture sampling for variable selection in non-Gaussian state space models. The distribution of $-\log \epsilon_t$ is a standard Gumbel and hence it suffices to add a negative sign to the means in Table 1 of \cite{Fruhwirth2007}. The parameters of the mixture approximation for the SV and stochastic conditional duration models are given in the Supplement \citep{Tan2022}. Figure \ref{fig_logchiexp} shows the true densities of $\log \chi_1^2$ and log exp(1), Gaussian approximation based on the mean and variance of the original variables and 10-component mixture approximations.

\begin{figure}[tb!]
\centering
\includegraphics[width=0.47\textwidth]{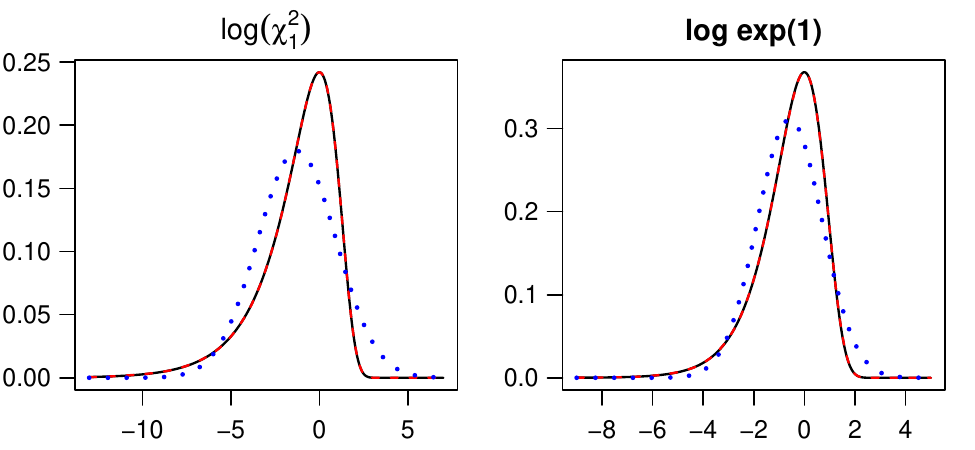}
\caption{Densities of $\log \chi_1^2$ and $\log \exp(1)$ (black solid lines), $\N(-1.27, 4.93)$ and $\N(-0.58, 1.64)$ approximations (blue dotted lines) and 10-component mixture approximations (red dashed lines).}
\label{fig_logchiexp}
\end{figure}

Using the mixture of normals approximation, we can construct MCMC algorithms for the SV and stochastic conditional duration models similarly, the only difference being $\tilde{y}_t = \log y_t^2$ for the SV model, while $\tilde{y}_t = \log y_t$ for the stochastic conditional duration model. Let $\tilde{\bfy} = (\tilde{y}_1, \dots, \tilde{y}_n)^T$, $\bfr = (r_1, \dots, r_n)^T$, $\bfm_\bfr = (m_{r_1}, \dots, m_{r_n})^T$ and $\bfs_\bfr^2 = (s_{r_1}^2, \dots, s_{r_n}^2)^T$. In matrix notation,
\begin{equation*}
\begin{gathered}
\tilde{\bfy}  - \bfm_\bfr | \bfr, \bfx \sim \N(\bfx,D_{\bfr}), \\
\bfx | \bftheta \sim \N(\mu\bone, \sigma_\eta^2\Lambda^{-1}), 
\end{gathered}
\end{equation*}
where $D_{\bfr} = \diag(\bfs_\bfr^2)$. As before, we introduce $\bfalpha = (\bfx- \mu \bfw)/\sigma_\eta^a$. Then the model can be written as 
\begin{equation*}
\begin{gathered}
\tilde{\bfy}  - \bfm_\bfr | \bfr, \bfalpha, \bftheta \sim \N(\sigma_\eta^a \bfalpha + \mu \bfw, D_{\bfr}), \\
\sigma_\eta^a\bfalpha | \bftheta \sim \N(\mu\bar{\bfw}, \sigma_\eta^2\Lambda^{-1}).
\end{gathered}
\end{equation*}
Following \cite{Kastner2014}, we specify independent priors for the components of $\bftheta$; $\mu \sim \N(b_\mu, B_\mu)$, $\sigma_\eta^2 \sim \text{Gamma}(0.5, 0.5B_\sigma)$ such that $\E(\sigma_\eta^2) = B_\sigma$, and $\text{Beta}(b_\phi, B_\phi)$ for $(\phi+1)/2$ \citep{Kim1998}. The joint posterior is 
\begin{equation*}
\begin{aligned}
p(\bfalpha, \bftheta, \bfr|\bfy ) \propto p(\bfy |\bfr, \bfalpha, \bftheta)p(\bfalpha|\bftheta) p(\bftheta) p(\bfr).
\end{aligned}
\end{equation*}
Choice of a gamma prior for $\sigma_\eta^2$ instead of the conventional inverse-gamma prior reduces sensitivity to prior hyperparameters \citep{Fruhwirth2010}.

\section{MCMC Algorithms}
\cite{Kastner2014} have designed MCMC samplers for the SV model based on the mixture of normals approximation approach described in Section \ref{sect_SV model}, using the CP or NCP. However, the performance of these algorithms vary across datasets with the CP or NCP preferred at times, depending on the level of volatility and persistence. On the other hand, the ancillarity-sufficiency interweaving strategy \citep[ASIS, ][]{Yu2011} was shown to improve the sampling efficiency of all parameters across different settings. In this section, we first introduce a baseline MCMC algorithm for the SV and stochastic conditional duration models that is similar to \cite{Kastner2014} in some aspects, but may be used with arbitrary values of $\{a, \bfw\}$. Then we describe the ASIS algorithm followed by the proposed block-specific reparametrization strategy (BSR) that makes use of the optimal data augmentation schemes derived in Section \ref{sec_EM}.

\begin{algorithm}
\renewcommand{\thealgorithm}{}
\floatname{algorithm}{}
\caption{MCMC Algorithm}
\begin{normalsize}
\begin{flushleft}
Initialize $\bftheta^{(0)}$ and $\bfr^{(0)}$. For $i=0, \dots, M-1$,
\begin{enumerate}
\item Sample $\bfalpha^{(i+1)}$ from $p(\bfalpha | \bfy, \bftheta^{(i)}, \bfr^{(i)})$,
\item Sample $\bftheta^{(i+1)}$ from $p(\bftheta|\bfy, \bfalpha^{(i+1)}, \bfr^{(i)})$,
\item Sample $\bfr^{(i+1)}$ from $p(\bfr|\bfy, \bfalpha^{(i+1)}, \bftheta^{(i+1)})$.
\end{enumerate}
Return the draws $\{\bftheta^{(i)}, \bfr^{(i)}| i=1, \dots, M\}$.
\end{flushleft}
\end{normalsize}
\end{algorithm}

The MCMC algorithm above may be used with arbitrary values of $\{a, \bfw\}$, and updates for the CP and NCP can be obtained simply by substituting appropriate values of $\{a, \bfw\}$. In steps 1 and 3, the sampling of $\bfalpha$ and $\bfr$ from their conditional posterior distributions are performed using Gibbs. To sample $\bftheta =(\mu, \sigma_\eta^2, \phi)$ in step 2, \cite{Kastner2014} investigated 1-block, 2-block and 3-block samplers, all requiring Metropolis-Hastings updates. Conditional on the parametrization, the simulation efficiency was similar across different blocking strategies, but the 3-block sampler has lower simulation efficiency for $\phi$ compared to the 1-block and 2-block samplers under the CP. Here we only consider the 3-block sampler for $\bftheta$. The updates in each step are derived below. 

For the latent states, $\bfalpha| \bfy, \bftheta, \bfr \sim \N(C_\bfalpha^{-1}\bfc_\bfalpha, C_\bfalpha^{-1})$, where 
\begin{equation*}
\begin{gathered}
C_\bfalpha = \sigma_\eta^{2a} ( D_{\bfr}^{-1} + \sigma_\eta^{-2}\Lambda ), \\
\bfc_{\bfalpha} = \sigma_\eta^a \{ D_{\bfr}^{-1} (\tilde{\bfy} - \bfm_\bfr - \mu \bfw) + \sigma_\eta^{-2} \mu \Lambda \bar{\bfw}\}.
\end{gathered}
\end{equation*}
As $C_\bfalpha$ is a symmetric tridiagonal matrix, we first compute the $LDL^T$ decomposition of $C_\bfalpha$, where $D$ is a diagonal matrix and $L$ is a unit lower triangular matrix with only the first lower diagonal nonzero. $C_\bfalpha^{-1}\bfc_\bfalpha$ is then computed using backward substitution. To generate $\bfalpha$, a random vector $\bfz$ is first simulated from $\N(0, I)$ so that $\bfalpha = C_\bfalpha^{-1}\bfc_\bfalpha + L^{-T} D^{-1/2} \bfz$. 

For the indicators $\bfr$, 
\[
p(\bfr | \bfy, \bfalpha, \bftheta) \propto \prod_{t=1}^n p(y_t | r_t, \alpha_t, \bftheta)p(r_t).
\]
Thus, the $\{r_t\}$ are conditionally independent a posteriori. For $t=1, \dots, n$, $k=1, \dots, K$, 
\begin{multline*}
\P(r_t = k | y_t, \alpha_t, \bftheta) \propto \exp\{ \log p_k - \log(s_k^2)/2 \\
 -(\tilde{y}_t - \mu w_t -\sigma_\eta^a \alpha_t - m_k)^2/(2s_k^2) \},
\end{multline*}
where $p_k = \P(r_t=k)$. Hence we can sample each $r_t$ according to the conditional posterior probabilities (normalized to sum to one) across $k=1, \dots, K$. 

We have $\mu | \bfy, \bfalpha, \bfr, \sigma_\eta^2, \phi \sim \N(C_\mu^{-1}c_\mu, C_\mu^{-1})$, where
\begin{equation*}
\begin{gathered}
C_\mu = B_\mu^{-1} + \bfw^T D_{\bfr}^{-1} \bfw +  \sigma_\eta^{-2} \bar{\bfw}^T \Lambda \bar{\bfw} , \\
c_\mu = B_\mu^{-1}b_\mu + \sigma_\eta^{a-2} \bfalpha^T \Lambda \bar{\bfw} + (\tilde{\bfy}  - \bfm_\bfr - \sigma_\eta^a \bfalpha)^T D_{\bfr}^{-1} \bfw.
\end{gathered}
\end{equation*}

For $\phi$, let $h_t = \sigma_\eta^a \alpha_t - \mu \bar{w}_t$ for $t=1, \dots, n$. Then
\begin{multline} \label{post_phi}
p(\phi|\bfy,  \bfalpha, \bfr, \mu, \sigma_\eta^2) \propto \\
\exp \bigg\{ g_\phi(\phi) -  \sum_{t=2}^n \frac{(h_t - \phi h_{t-1})^2}{2\sigma_\eta^2} \bigg\},
\end{multline}
where $g_\phi(\phi)$ is given by
\[
\Big(b_\phi-\frac{1}{2} \Big) \log(1+\phi) + \Big(B_\phi- \frac{1}{2} \Big) \log(1-\phi) + \frac{\phi^2 h_1}{2\sigma_\eta^2}.
\]
We perform a Metropolis-Hastings step \citep{Kim1998}, where a proposal $\phi^\new$ is generated from 
\[
\N \left( \frac{\sum_{t=1}^{n-1} h_t h_{t+1}}{\sum_{t=1}^{n-1} h_t^2}, \frac{ \sigma_\eta^2}{\sum_{t=1}^{n-1} h_t^2} \right),
 \]
based on a normal approximation to the last term in \eqref{post_phi}. If $|\phi^\new| < 1$, then $\phi^\new$ is accepted with probability equal to $\min \left( 1, \exp\{g_\phi(\phi^\new) - g_\phi(\phi^\old)\}  \right)$. 

For $\sigma_\eta^2$, we first discuss updates for the CP and NCP before proposing an update for arbitrary $\{a, \bfw\}$. For the CP, $p(\sigma_\eta^2 | \bfy, \bfalpha, \bfr, \mu, \phi) 
\propto \exp \{ f_\sigma(\sigma_\eta^2)\}$, where
\[
\begin{aligned}
f_\sigma(\sigma_\eta^2) = - C_\sigma/ \sigma_\eta^2 - (c_\sigma + 1) \log (\sigma_\eta^2)  - \sigma_\eta^2/(2B_\sigma),
\end{aligned}
\]
$C_\sigma = (\bfalpha - \mu \bone)^T \Lambda (\bfalpha - \mu \bone)/2$ and $c_\sigma = (n-1)/2$. \cite{Kastner2014} propose a Metropolis-Hastings step, where a proposal ${\sigma_\eta^2}^\new$ is generated from the inverse gamma distribution, $\text{IG}\left( c_\sigma, C_\sigma \right)$, and is accepted with a probability of
\[
\min(1, \exp[ ({\sigma_\eta^2}^\old - {\sigma_\eta^2}^\new)/(2B_\sigma) ]).
\]
For the NCP, 
\[
\begin{aligned}
p(\sigma_\eta|\bfy, \bfalpha, \bfr, \mu, \phi) &\propto \exp\{ c_\sigma' \sigma_\eta - C_\sigma'\sigma_\eta^2 /2 \},
\end{aligned}
\]
where $c_\sigma' = \bfalpha^T D_{\bfr}^{-1} (\tilde{\bfy} - \bfm_\bfr - \mu\bone)$ and $C_\sigma' = \bfalpha^T D_{\bfr}^{-1} \bfalpha + B_\sigma^{-1}$. Thus we can generate a proposal $\sigma_\eta^\new$ from $\N( C_\sigma'^{-1}c_\sigma', C_\sigma'^{-1} )$, which is accepted if $\sigma_\eta^\new > 0$. For arbitrary $\{a, \bfw\}$, we consider a transformation $\sigma_\eta^2 = \e^\nu$ so that $\nu$ is unconstrained. We have $p(\nu | \bfy, \bfalpha, \bfr, \mu, \phi) \propto \exp \{f_\nu(\nu)\}$, where $f_\nu(\nu)$ and its first and second order derivatives are given in the Supplement \citep{Tan2022}. Let $\hat{\nu} = \argmax_\nu f_\nu(\nu)$ denote the posterior mode found using the L-BFGS algorithm. Then $f_\nu(\nu) \approx f_\nu(\hat{\nu}) + f_\nu''(\hat{\nu})(\nu - \hat{\nu})^2/2$ since $f'(\hat{\nu}) = 0$, and a normal approximation to the conditional posterior is $\N(\hat{\nu}, -1/f_\nu''(\hat{\nu}) )$. A proposal $\nu^\new$ generated from this normal density is accepted with probability $\min \left(1, \exp\{g_\nu(\nu^\new) - g_\nu(\nu^\old)\} \right)$, where
\[
g_\nu(\nu) = f_\nu(\nu) -f_\nu''(\hat{\nu}) (\nu - \hat{\nu})^2/2.
\]

\subsection{Ancillarity-sufficiency interweaving strategy}
The mixing and convergence of the baseline MCMC sampler depends heavily on the parametrization chosen for the data. For any parametrization, the simulation efficiency may also differ across parameters. As the behavior of the CP and NCP are often complementary (one converges quickly if the other is slow), \cite{Yu2011} proposed an ancillarity-sufficiency interweaving strategy (ASIS), that interweaves a {\em sufficient} augmentation scheme where the missing data is a sufficient statistic for $\bftheta$ (usually the CP), and an {\em ancillary} augmentation scheme where the missing data is an ancillary statistic independent of $\bftheta$ (usually the NCP). It is proven that ASIS is at least better than the worse of the two schemes. ASIS can be implemented using the CP or NCP as baseline, but \cite{Kastner2014} report that choice of the baseline is immaterial. Hence, we only consider ASIS with CP as baseline, the steps of which are outlined below. 

\begin{algorithm}
\renewcommand{\thealgorithm}{}
\floatname{algorithm}{}
\caption{ASIS Algorithm (baseline: CP)}
\begin{normalsize}
\begin{flushleft}
Initialize $\bftheta^{(0)}$ and $\bfr^{(0)}$. For $i=0, \dots, M-1$,
\begin{enumerate}
\item Under the CP (set $a=0$, $\bfw=\bfzero$):
\begin{enumerate}
\item Sample $\bfalpha^{(i)}$ from $p(\bfalpha | \bfy, \bftheta^{(i)}, \bfr^{(i)})$.
\item Sample $\bftheta^{(i+0.5)}$ from $p(\bftheta|\bfy, \bfalpha^{(i)}, \bfr^{(i)})$.
\end{enumerate}
\item Move to the NCP (set $a=1$, $\bfw=\bone$):
\begin{enumerate}
\item Compute $\tilde{\bfalpha}^{(i+1)} = (\bfalpha^{(i)} - \mu^{(i+0.5)})/ \sigma_\eta^{(i+0.5)}$.
\item Sample $\bftheta^{(i+1)}$ from $p(\bftheta|\bfy, \tilde{\bfalpha}^{(i+1)} , \bfr^{(i)})$.
\end{enumerate}
\item Move back to the CP (set $a=0$, $\bfw=\bfzero$):
\begin{enumerate}
\item Compute $\bfalpha^{(i+1)} = \mu^{(i+1)} + \sigma_\eta^{(i+1)} \tilde{\bfalpha}^{(i+1)}$.
\item Sample $\bfr^{(i+1)}$ from $p(\bfr|\bfy, \bfalpha^{(i+1)}, \bftheta^{(i+1)})$.
\end{enumerate}
\end{enumerate}
Return the draws $\{\bftheta^{(i)}, \bfr^{(i)} | i=1, \dots, M\}$.
\end{flushleft}
\end{normalsize}
\end{algorithm}

\subsection{Block-specific reparametrization strategy}
To optimize the convergence rate of the EM algorithm for the AR(1) plus noise model, we employed Algorithm 3 so that the optimal parametrization for each parameter can be used while conditioning on the rest. As the convergence rate of the Gibbs sampler is equal to the corresponding EM algorithm when the target density is well approximated by a Gaussian \citep{Sahu1999}, tailoring the data augmentation scheme to each block in a Gibbs sampler may yield similar improvements in convergence rates as in the EM algorithm. This is important as the simulation efficiency of a parametrization may be good for certain parameters but poor for others. Hence we propose a {\em block-specific reparametrization} (BSR) strategy that allows one to adopt a different parametrization for sampling from each block in a multi-block MCMC sampler. This yields flexibility in choosing a parametrization optimal for sampling from each block. 

First, we outline the steps of an MCMC sampler adopting the BSR strategy and verify that the sampler will converge to the target density, before discussing the reparametrizations for each block. For simplicity, we limit the discussion to two blocks although BSR can be applied to multiple blocks. Suppose the parameters are split into two blocks, $\bftheta = (\bftheta_1^T, \bftheta_2^T)^T$, and consider two augmentation schemes with missing data $\bfalpha_1$ and $\bfalpha_2$ such that their joint distribution $p(\bfalpha_1, \bfalpha_2|\bftheta, \bfy)$ is well defined. This distribution may be degenerate, e.g. if $\bfalpha_2$ is a deterministic function of $\bfalpha_1$. Consider an MCMC algorithm that initializes $\bftheta$ and then performs the following steps at the $i$th iteration for $i=0, \dots, M-1$ (explicit conditioning on $\bfy$ has been omitted for simplicity):
\begin{enumerate}
\item Sample $\bfalpha_1^{(i+1)}$ from $p(\bfalpha_1|\bftheta^{(i)})$.
\item Sample $\bftheta_1^{(i+1)}$ from $p(\bftheta_1 |\bfalpha_1^{(i+1)}, \bftheta_2^{(i)})$.
\item Sample $\bfalpha_2^{(i+1)}$ from $p(\bfalpha_2|\bfalpha_1^{(i+1)}, \bftheta_1^{(i+1)}, \bftheta_2^{(i)})$. 
\item Sample $\bftheta_2^{(i+1)}$ from $p(\bftheta_2 | \bfalpha_2^{(i+1)}, \bftheta_1^{(i+1)})$.
\end{enumerate}
The draws $\{\bftheta^{(i)}|i=1, \dots, M\}$ are retained. The transition density of the above chain is 
\begin{multline*}
k(\bftheta'|\bftheta) = \int p(\bftheta_2' | \bfalpha_2, \bftheta_1') p(\bfalpha_2|  \bfalpha_1, \bftheta_1', \bftheta_2) p(\bftheta_1' |\bfalpha_1, \bftheta_2)\\
\times p(\bfalpha_1| \bftheta) d\bfalpha_1 d\bfalpha_2. 
\end{multline*}
Let $p(\bftheta)$ denote the stationary density. We have
\begin{align*}
& \int k(\bftheta'|\bftheta) p(\bftheta) d\bftheta \\
&= \int p(\bftheta_2' | \bfalpha_2, \bftheta_1') p(\bfalpha_2|  \bfalpha_1, \bftheta_1', \bftheta_2) p(\bftheta_1' |\bfalpha_1, \bftheta_2) \\
& \quad  \times \left( \int p(\bfalpha_1, \bftheta) d\bftheta_1 \right) d\bftheta_2 d\bfalpha_1 d\bfalpha_2 \\
&=\int \negmedspace p(\bftheta_2' | \bfalpha_2, \bftheta_1') p(\bfalpha_2|  \bfalpha_1, \bftheta_1', \bftheta_2) p(\bftheta_1', \bfalpha_1, \bftheta_2) d\bftheta_2 d\bfalpha_1 d\bfalpha_2 \\
&= \int p(\bftheta_2' | \bfalpha_2, \bftheta_1') \left( \int p(\bfalpha_2,  \bfalpha_1, \bftheta_1', \bftheta_2) d\bftheta_2 d\bfalpha_1 \right) d\bfalpha_2 \\
&= \int p(\bftheta_2' , \bfalpha_2, \bftheta_1')  d\bfalpha_2 
= p(\bftheta).
\end{align*}
Hence the stationary density is preserved by BSR. Moreover, if $\bfalpha_2$ is a deterministic function ($f$) of $\bfalpha_1$ given $\bftheta$, then we can compute $\bfalpha_2^{(i+1)} = f(\bfalpha_1^{(i+1)}| \bftheta_1^{(i+1)}, \bftheta_2^{(i)})$ without sampling $\bfalpha_2$ in step 3. This feature is important in models where $\bfalpha$ is high-dimensional (e.g. state space models) and sampling of $\bfalpha$ is time-consuming.

Next, we discuss how BSR can be applied. Conditional on $\bfr$, $\bfalpha$ and $\bftheta$, $\tilde{\bfy}  - \bfm_\bfr \sim \N(\sigma_\eta^a \bfalpha + \mu \bfw, D_{\bfr})$. Comparing this distribution with \eqref{mod_G}, the differences are, $\bfy$ is replaced by $\tilde{\bfy}  - \bfm_\bfr$, and $\sigma_\epsilon^2I$ is replaced by $D_\bfr^{-1}$. In Section \ref{sec_EMBayesian}, we have shown that when using the EM algorithm to find the posterior mode, the optimal parametrizations for $\mu$ and $\sigma_\eta^2$ remain unchanged and are independent of the priors. It can also be verified that for the AR(1) plus noise model, if $\sigma_\epsilon^2$ is not constant across time, then it suffices to replace $\sigma_\epsilon^{-2}I$ in the expressions of $a$ and $\bfw$ by a diagonal matrix filled with the variances at different time points (see Supplement \citep{Tan2022}). Applying these results to the SV and stochastic conditional duration models using a mixture of normals approximation, we divide the parameters into two blocks, $\mu$ and $\{\sigma_\eta^2, \phi, \bfr\}$, and use the optimal schemes for $\mu$ and $\sigma_\eta^2$ respectively in these two blocks. Based on the EM algorithm, $I_{\phi, \phi}$ is independent of $\{a, \bfw\}$. Hence it is not possible to choose $\{a, \bfw\}$ to optimize the convergence of $\phi$, and we just sample it along with $\sigma_\eta^2$ in the second scheme. For convenience, $\bfr$ is sampled in the second scheme as well. Conditional on $\bfr$, the values of $\{a, \bfw\}$ in the first scheme are 
\begin{equation} \label{scheme1}
a_1 = 0, \quad \bar{\bfw}_1 = V_0 D_\bfr^{-1} \bone,
\end{equation}
where $V_0 = (D_\bfr^{-1} + \sigma_\eta^{-2} \Lambda)^{-1}$ and we have chosen $a_1 = 0$ for simplicity. For the second scheme,
\begin{equation} \label{scheme2}
a_2 = 1 - \frac{ \tr(D_\bfr^{-1}V_0 )}{n}, \;\; \bar{\bfw}_2 = \frac{1}{\mu} \left( \frac{2V_0 \Lambda}{a_2 \sigma_\eta^2} - I \right)\bfm_{0\bone},
\end{equation}
where $\bfm_{0\bone} = V_0 D_{\bfr}^{-1} (\tilde{\bfy}  - \bfm_\bfr - \mu \bone)$. Since $\bfalpha_s = (\bfx - \mu \bfw_s)/\sigma_\eta^{a_s}$ for $s = 1, 2$, we can set $\bfalpha_2 = \{\bfalpha_1 + \mu (\bar{\bfw}_2 - \bar{\bfw}_1)\}/\sigma_\eta^{a_2}$ when switching from the first to second scheme. The steps in the MCMC sampler using BSR are outlined below. 

\begin{algorithm}
\renewcommand{\thealgorithm}{}
\floatname{algorithm}{}
\caption{BSR Algorithm}
\begin{normalsize}
\begin{flushleft}
Initialize $\bftheta^{(0)}$ and $\bfr^{(0)}$. For $i=0, \dots, M-1$,
\begin{enumerate}
\item First scheme (set $a=a_1$, $\bfw=\bfw_1$):
\begin{enumerate}
\item Sample $\bfalpha_1^{(i+1)} \sim p(\bfalpha | \bfy, \bftheta_1^{(i)}, \bftheta_2^{(i)})$.
\item Sample $\mu^{(i+1)} \sim p(\mu | \bfy, \bfalpha_1^{(i+1)},{\sigma_\eta^2}^{(i)}, \phi^{(i)}, \bfr^{(i)})$.
\end{enumerate}
\item Switch to second scheme (set $a=a_2$, $\bfw=\bfw_2$):
\begin{enumerate}
\item Compute $\bfalpha^{(i+1)}_2 = \{\bfalpha_1^{(i+1)} + \mu^{(i+1)} (\bar{\bfw}_2 - \bar{\bfw}_1)\}/ {\sigma_\eta^{(i)}}^{a_2}$.
\item Sample $\sigma_\eta^{2(i+1)} \sim p(\sigma_\eta^2|\bfy, \bfalpha_2^{(i+1)} ,\mu^{(i+1)},\phi^{(i)}, \bfr^{(i)})$.
\item Sample $\phi^{(i+1)} \sim p(\phi | \bfy, \bfalpha_2^{(i+1)} ,\mu^{(i+1)}, \sigma_\eta^{2(i+1)} , \bfr^{(i)})$.
\item Sample $\bfr^{(i+1)} \sim p(\bfr | \bfy, \bfalpha_2^{(i+1)}, \bftheta^{(i+1)})$.
\end{enumerate}
\end{enumerate}
Return the draws $\{\bftheta^{(i)}, \bfr^{(i)} | i=1, \dots, M\}$.
\end{flushleft}
\end{normalsize}
\end{algorithm}

The values of $\{a, \bfw\}$ in \eqref{scheme1} and \eqref{scheme2} depend on the true values of $\bftheta$ and $\bfr$, which are unknown. Hence we propose the following procedure to overcome this difficulty. To initialize $\bftheta^{(0)}$, $\bar{\bfw}_1$, $\bar{\bfw}_2$ and $a_2$, we run Algorithm 3 on transformed responses $\tilde{y}_t$, by approximating the distribution of $\log \epsilon_t^2$ in the SV model by a single Gaussian $\N(-1.27, 4.93)$, and $\log \epsilon_t$ in the stochastic conditional duration model by $\N(-0.58, 1.64)$. The maximum likelihood estimate $\hat{\bftheta}$ is used as initial estimate $\bftheta^{(0)}$, while $\bfr^{(0)}$ is simulated randomly from its prior distribution. This initialization is also used for the baseline MCMC algorithm and ASIS. We use $\hat{\bftheta}$ to compute $\{\bar{\bfw}_1$, $\bar{\bfw}_2$, $a_2\}$ by assuming $\sigma_\epsilon^2$ is constant across $t$ (equal to 4.93 and 1.64 for the SV and stochastic conditional duration models respectively). These values of $\{\bar{\bfw}_1$, $\bar{\bfw}_2$, $a_2\}$ are held fixed for the first two-thirds of the burn-in period. We estimate the posterior means of $\bftheta$ and $\bfr$ using draws after the first one-third till the first two-thirds of burn-in, and these are used to update the values of $\{\bar{\bfw}_1$, $\bar{\bfw}_2$, $a_2\}$ after the first two-thirds of burn-in. For the remainder of the sampling period, $\{\bar{\bfw}_1$, $\bar{\bfw}_2$, $a_2\}$ are held fixed at these values. Our experiments indicate that BSR works well even with these rather crude estimates of $\{\bar{\bfw}_1$, $\bar{\bfw}_2$, $a_2\}$.

For the SV model, it may be of interest to estimate the volatility at time $t$, $\exp(x_t/2)$. Since $p(\bfx | \bfy) = \int p(\bfx | \bftheta, \bfr, \bfy)p(\bftheta, \bfr | \bfy) \, d\bftheta \, d\bfr$, we can draw $\bfx$ by sampling from $p(\bfx | \bftheta, \bfr, \bfy)$, which is $\N(C_\bfalpha^{-1} \bfc_\bfalpha, C_\bfalpha^{-1} )$ under the CP, for each draw of $(\bftheta, \bfr)$ from $p(\bftheta, \bfr | \bfy)$. The MCMC samplers presented are based on a mixture of normals approximation of the $\log \chi_1^2$ and log exp(1) distributions, and hence the samples are not strictly drawn from the true posterior. This minor approximation error can be corrected by reweighting the draws \citep[see][]{Kim1998}. As the mixture approximation used is very accurate, \cite{Kim1998} and \cite{Omori2007} report that the effect of reweighting is small and estimates of posterior means are not significantly different statistically. There is however, some improvement in the Monte Carlo standard errors. We focus on comparison of different MCMC samplers, and do not perform reweighting in our experiments. \cite{Hosszejni2021} present two R packages, {\tt stochvol} and {\tt factorstochvol} for Bayesian estimation of univariate SV models (with linear means, leverage or heavy tails) and multivariate factor SV models. In {\tt stochvol}, sampling of the latent states also relies on a Gaussian mixture approximation and correction of the approximation error, while not implemented by default, can be enabled simply through an argument.

\section{Experimental results}
We compare BSR with the baseline MCMC algorithm using the CP or NCP and ASIS, using simulated data and several benchmark real datasets. Following \cite{Kastner2014}, we use the {\em inefficiency factor} to assess simulation efficiency, which is estimated as the ratio of the variance of the sample mean from the MCMC sampler, to the variance of the sample mean based on independent draws. It can be computed using the R package {\tt coda} \citep{Plummer2006} as number of draws/effective sample size. Smaller inefficiency factors indicate better mixing and lower autocorrelation among draws. Each MCMC sampler is run for 30,000 iterations, of which the first 10,000 are discarded as burn-in. All code is written in Julia 1.6.1 \citep{Bezanson2017} and R \citep{R2020}. The experiment on consumer price index is run on an Intel Xeon Gold 5222 CPU @3.80GHz while all other experiments are run on an Intel Core i9-9900K CPU @3.60 GHz.

\subsection{Simulations}
First, we present results on data simulated from the SV and stochastic conditional duration models. For each model, we consider $\mu = -10$, $\phi \in \{0.4, 0.7, 0.95\}$ and $\sigma_\eta^2 \in \{0.05, 0.5, 5\}$, which yields nine parameter settings. In each setting,  $n=3000$ observations are generated and each experiment is repeated ten times. For the priors, following \cite{Kastner2014}, we set $b_\mu = -10$, $B_\mu = 10$, $B_\sigma = {\sigma_\eta^2}_{\text{true}}$, $b_\phi = 40$ and $B_\phi = 80/(1+\phi_{\text{true}}) - 40$ (this represents a prior on $\phi$ with mean equal to $\phi_{\text{true}}$). 
\begin{table*}[tb!]
\caption{SV model: Inefficiency factors averaged over ten repetitions for simulated data.}
\label{tab_SVsim}
\begin{tabular}{@{}llrrrrrrrrrrrr@{}}
\hline
&& \multicolumn{4}{c}{$\mu$} & \multicolumn{4}{c}{$\sigma_\eta^2$} & \multicolumn{4}{c}{$\phi$}\\  
$\phi$ & $\sigma_\eta^2$ & NCP & CP & ASIS & BSR & NCP & CP & ASIS & BSR & NCP & CP & ASIS & BSR\\ \hline
\multirow{3}{*}{0.4} & 0.05 & 18 & 79 & 17 & 20 & 63 & 710 & 63 & 66 & 109 & 111 & 116 & 112 \\
& 0.5 & 17 & 16 & 9 & 9 & 42 & 68 & 32 & 31 & 51 & 54 & 52 & 48 \\
& 5 & 67 & 4 & 4 & 4 & 40 & 17 & 15 & 13 & 16 & 13 & 13 & 12 \\  \hline
\multirow{3}{*}{0.7} & 0.05 & 14 & 26 & 11 & 13 & 72 & 471 & 73 & 68 & 110 & 128 & 114 & 112 \\
& 0.5 & 32 & 7 & 5 & 5 & 54 & 64 & 37 & 31 & 43 & 43 & 34 & 28 \\
& 5 & 226 & 2 & 2 & 2 & 70 & 17 & 16 & 13 & 19 & 9 & 9 & 8 \\ \hline
\multirow{3}{*}{0.95} & 0.05 & 99 & 2 & 2 & 2 & 69 & 136 & 51 & 30 & 48 & 75 & 36 & 21 \\
& 0.5 & 823 & 1 & 1 & 1 & 114 & 28 & 24 & 13 & 33 & 9 & 8 & 5 \\
& 5 & 5994 & 1 & 1 & 1 & 330 & 12 & 12 & 9 & 13 & 3 & 2 & 2 \\  \hline
\end{tabular}
\end{table*}

The inefficiency factors averaged over ten repetitions for each parameter setting are shown in Tables \ref{tab_SVsim} and \ref{tab_SCDsim} for the SV and stochastic conditional duration models respectively. The BSR strategy worked very well for both models, and the inefficiency factors are (almost) always better than the worst of NCP and CP. It is even able to yield improvements over ASIS in many cases. We also observe that in a given setting, neither the CP or NCP may be superior to the other in terms of simulation efficiency across all parameters. Hence BSR can yield benefits beyond the CP and NCP by tailoring the parametrization to each block. We have also performed experiments for negative values of $\phi$ and the results (not shown) are similar to positive values of $\phi$ (much like a reflection about $\phi=0$).

\begin{table*}[tb!]
\caption{Stochastic conditional duration model: Inefficiency factors averaged over ten repetitions for simulated data.}
\label{tab_SCDsim}
\begin{tabular}{@{}llrrrrrrrrrrrr@{}}
\hline
&& \multicolumn{4}{c}{$\mu$} & \multicolumn{4}{c}{$\sigma_\eta^2$} & \multicolumn{4}{c}{$\phi$}\\  
$\phi$ & $\sigma_\eta^2$ & NCP & CP & ASIS & BSR & NCP & CP & ASIS & BSR & NCP & CP & ASIS & BSR\\ \hline
\multirow{3}{*}{0.4} & 0.05 & 20 & 35 & 19 & 18 & 94 & 402 & 87 & 87 & 112 & 110 & 113 & 112 \\
 & 0.5 & 34 & 10 & 8 & 8 & 51 & 49 & 34 & 33 & 42 & 36 & 36 & 33 \\
 & 5 & 197 & 2 & 2 & 2 & 70 & 9 & 8 & 7 & 8 & 6 & 6 & 6 \\  \hline
\multirow{3}{*}{0.7} & 0.05 & 21 & 16 & 11 & 11 & 111 & 302 & 92 & 96 & 110 & 127 & 106 & 105 \\
 & 0.5 & 92 & 4 & 4 & 3 & 76 & 43 & 34 & 29 & 44 & 27 & 25 & 21 \\
 & 5 & 700 & 1 & 1 & 1 & 130 & 10 & 10 & 8 & 11 & 5 & 5 & 4 \\  \hline
\multirow{3}{*}{0.95}  & 0.05 & 322 & 2 & 1 & 1 & 124 & 85 & 54 & 32 & 62 & 39 & 30 & 18 \\
 & 0.5 & 2427 & 1 & 1 & 1 & 239 & 23 & 21 & 16 & 36 & 6 & 6 & 5 \\
 & 5 & 9765 & 1 & 1 & 1 & 731 & 9 & 9 & 7 & 7 & 2 & 2 & 2 \\  \hline
\end{tabular}
\end{table*}

The average runtime of each MCMC sampler is given in Table \ref{tab_simtimes}. The amount of computation required by NCP and CP are about the same, while BSR requires slightly more computation (for updating the working parameters and switching to the second scheme in each iteration). ASIS requires the most computation due to switching of schemes and multiple simulation of certain parameters at each iteration, but the increase in runtime of BSR and ASIS compared to CP and NCP are not significant. 
\begin{table}[htb!]
\centering
\begin{tabular}{lrrrr}
\hline
& NCP & CP & ASIS & BSR \\   \hline
SV & 238 & 238 & 245 & 241 \\
Stochastic conditional duration & 202 & 202 & 210 & 206 \\    \hline
\end{tabular}
\caption{Runtimes averaged over all parameter settings and repetitions for simulated data.\label{tab_simtimes}}
\end{table}

\subsection{Real data}
We consider three datasets on exchange rates, US consumer price index and IBM transactions. We set the priors as $b_\phi = 20$, $B_\phi = 1.5$ \citep[following][]{Kim1998}, $B_\sigma = 0.5$ and $B_\mu = 100$. For the exchange rates, $b_\mu = -10$, while for the consumer price index and IBM transactions, $b_\mu = 0$. 

The exchange rate data is downloaded from the European Central Bank's Statistical Data Warehouse. \cite{Kastner2014} analyzed $n=3140$ observations of the daily Euro exchange rates of 23 currencies from 3 Jan 2000 to 4 Apr 2012 using the SV model. Of these, we select the Danish krone, New Zealand dollar and US dollar, which are representative of currencies with the lowest, moderate and highest persistence respectively. Let $r_t$ be the exchange rate at time $t$. The response $y_t = \log (r_t/r_{t-1}) - \sum_{t=2}^n \log (r_t/r_{t-1})/(n-1)$ is the mean-corrected differenced log return. Using MCMC samplers to fit the SV model to the exchange rates, we obtain the inefficiency factors in Table \ref{tab_IFex}. The simulation efficiency of CP is higher than NCP for $\mu$ but lower than NCP for $\{\sigma_\eta^2, \phi\}$. Hence, neither is a good option. ASIS achieves much better results but BSR is the clear winner, with the lowest inefficiency factors across all currencies and parameters. 
\begin{table*}
\caption{Inefficiency factors from SV model for exchange rates.}
\label{tab_IFex}
\begin{tabular}{@{}lrrrrrrrrrrrr@{}}
\hline & \multicolumn{4}{c}{$\mu$} & \multicolumn{4}{c}{$\sigma_\eta^2$} & \multicolumn{4}{c}{$\phi$}\\  
& NCP & CP & ASIS & BSR & NCP & CP & ASIS & BSR & NCP & CP & ASIS & BSR\\ \hline
Danish & 91 & 3 & 3 & 3 & 97 & 143 & 64 & 43 & 71 & 87 & 52 & 32 \\
NZ & 110 & 3 & 2 & 2 & 189 & 461 & 129 & 72 & 147 & 345 & 113 & 58 \\
US & 455 & 1 & 1 & 1 & 123 & 354 & 78 & 28 & 84 & 114 & 39 & 14 \\ \hline
\end{tabular}
\end{table*}

\begin{figure}[tb!]
\centering 
\includegraphics[width=0.48\textwidth]{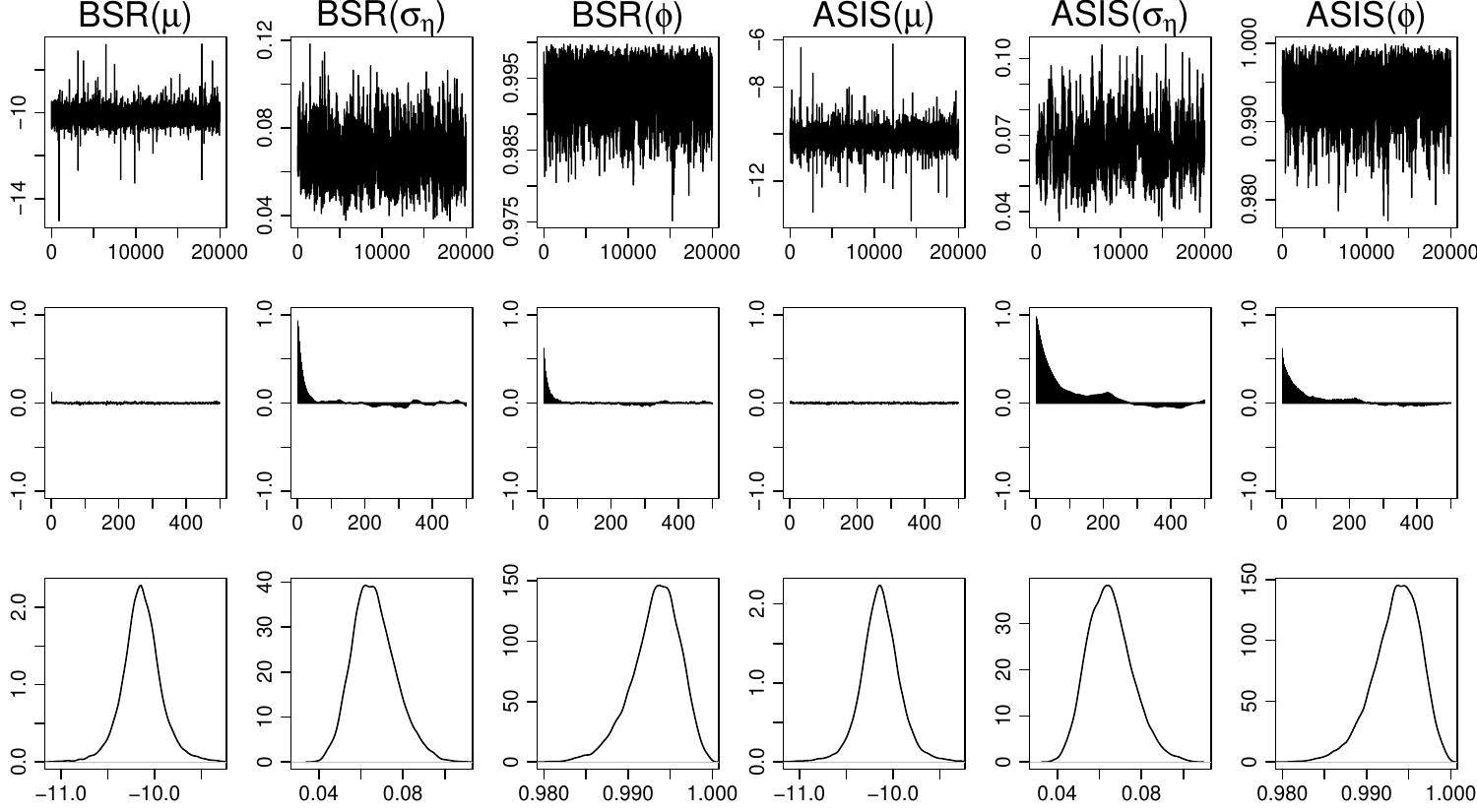}
\caption{US dollar exchange rates. First row: Trace plots. Second row: Sample autocorrelation functions. Third row: Estimated marginal posterior densities. \label{Fig_USplots}}
\end{figure}

The parameter estimates and runtime of each sampler are given in the Supplement \citep{Tan2022}. NCP is the fastest, followed by CP, BSR and ASIS, but the runtime of all samplers are practically the same. Figure \ref{Fig_USplots} shows the trace plots, sample autocorrelation function and marginal posterior densities of $\{\mu, \sigma_\eta, \phi\}$ for the US dollar. There is better mixing in BSR's $\sigma_\eta$ chain than ASIS's, and the sample autocorrelation for $\{\sigma_\eta, \phi\}$ also decay faster. These observations are consistent with the inefficiency factors. Figure \ref{Fig_USvol} shows the 0.05, 0.5 and 0.95 quantiles of volatilities estimated using 20000 samples from BSR, and a similar pattern with $|y_t|$ can be detected.

\begin{figure}[htb!]
\centering 
\includegraphics[width=0.48\textwidth]{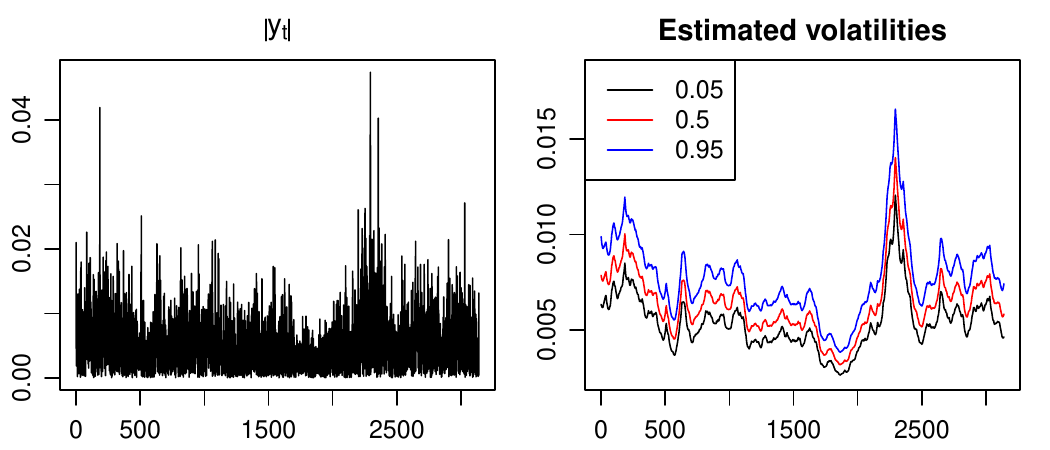}
\caption{US dollar exchange rates. Plot of $|y_t|$ and volatilities (0.05, 0.5 and 0.95 quantiles) estimated using BSR. \label{Fig_USvol}}
\end{figure}

The US consumer price index inflation rate from the 2nd quarter of 1947 to the 3rd quarter of 2011 \citep{Kroese2014} comprises of $n = 258$ observations and a SV model is fitted to the mean-corrected rates. Besides the baseline MCMC samplers, ASIS and BSR, we also provide comparisons with NUTS (Hamiltonian Monte Carlo No U-turn sampler in Stan) and PMMH (particle marginal Metropolis-Hastings algorithm in R package {\tt nimble}). No tuning is required for NUTS. Following \cite{Michaud2020}, PMMH is run using the auxiliary particle filter \citep{Pitt1999b} and a block random walk sampler with a normal proposal. It is rather sensitive to initialization, the number of particles $m$ and the scale of proposal, which must be set by the user. We set the starting values as the posterior means, and standard deviations of the proposal as $h$ times the posterior standard deviations computed using BSR. We consider $h \in \{0.6, 0.7, \dots, 1.1\}$ and $m \in \{50, 100, 150, 200\}$. For this dataset, PMMH tends to get stuck for many iterations and mixing is poor. We report results for $m=150$ and $h=0.7$, which has the best inefficiency factors. The parameter estimates, inefficiency factors and runtimes are given in Table \ref{tab_parcpi}. 

The simulation efficiency of CP is better than NCP, and ASIS improves upon CP, while BSR improves upon ASIS. NUTS also provided improvements compared to the respective baseline samplers. NUTS and PMMH are generally more computationally intensive than standard MCMC samplers. NUTS uses automatic differentiation to compute gradients of the Hamiltonian at each iteration, which enables it to explore the posterior distribution of highly correlated parameters more thoroughly. On the other hand, a particle filter has to be run at each iteration for PMMH and computation costs scale linearly with the number of particles. While PMMH is valid even for small number of particles, it has a risk of getting stuck due to the variance in the likelihood approximation. The consumer price index inflation exhibits high persistence, and the volatilities estimated using 20000 samples from BSR in Figure \ref{Fig_CPIvol} captures the high inflation in the late 1970s to early 1980s. After that, inflation remains low and stable until the 2008 financial crisis.
\begin{table*}[tb!]
\caption{Parameter estimates, inefficiency factors and runtime (seconds) from SV model for consumer price index.}
\label{tab_parcpi}
\begin{tabular}{@{}llrrrrrrr@{}}
\hline
&& NCP & CP & ASIS & BSR & NUTS (NCP) & NUTS (CP) & PMMH \\   \hline
\multirow{3}{*}{\parbox[t]{1.5cm}{Parameter estimates}} & $\mu$ & 1.56 $\pm$ 0.48 & 1.61 $\pm$ 0.48 & 1.62 $\pm$ 0.47 & 1.61 $\pm$ 0.47 & 1.61 $\pm$ 0.45 & 1.61 $\pm$ 0.55 & 1.66 $\pm$ 0.42 \\
& $\sigma_\eta$ &  0.60 $\pm$ 0.10 & 0.60 $\pm$ 0.11 & 0.59 $\pm$ 0.10 & 0.60 $\pm$ 0.10 & 0.61 $\pm$ 0.11 & 0.60 $\pm$ 0.11 & 0.56 $\pm$ 0.09 \\
& $\phi$ &  0.90 $\pm$ 0.04 & 0.90 $\pm$ 0.04 & 0.90 $\pm$ 0.04 & 0.90 $\pm$ 0.04 &  0.90 $\pm$ 0.04 &  0.90 $\pm$ 0.04  & 0.91 $\pm$ 0.03 \\ \hline 
\multirow{3}{*}{\parbox[t]{1.5cm}{Inefficiency factor}} & $\mu$ & 260 & 2 & 1 & 1 & 24 & 1 &72\\
& $\sigma_\eta$ &  50 & 37 & 26 & 19 & 8 & 24 &118 \\
& $\phi$ &  37 & 15 & 13 & 8 & 7 & 10 & 92 \\   \hline 
Runtime && 14 & 14 & 15 & 16 & 113 & 138 & 749 \\ \hline
\end{tabular}
\end{table*}

\begin{figure}[tb!]
\centering 
\includegraphics[width=0.48\textwidth]{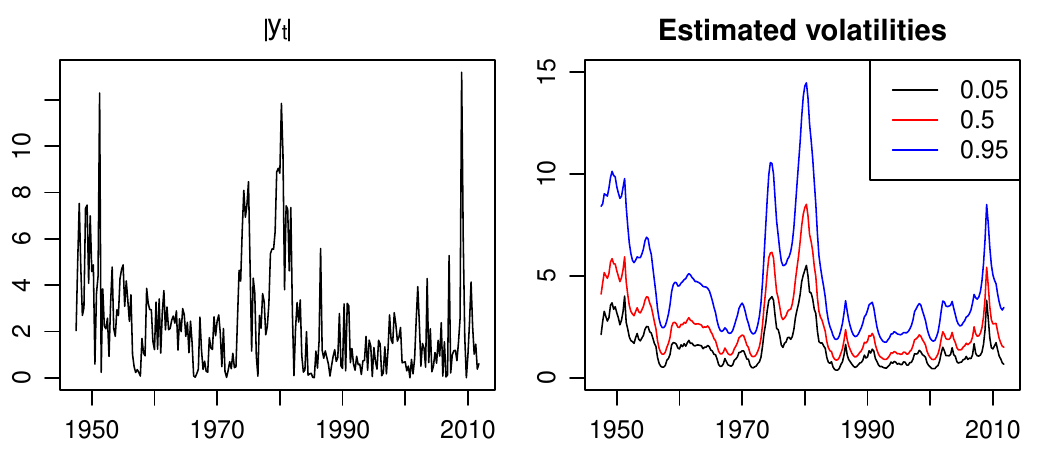}
\caption{Consumer price index. Plot of $|y_t|$ and volatilities (0.05, 0.5 and 0.95 quantiles) estimated using BSR. \label{Fig_CPIvol}}
\end{figure}

\begin{figure}
\centering 
\includegraphics[width=0.48\textwidth]{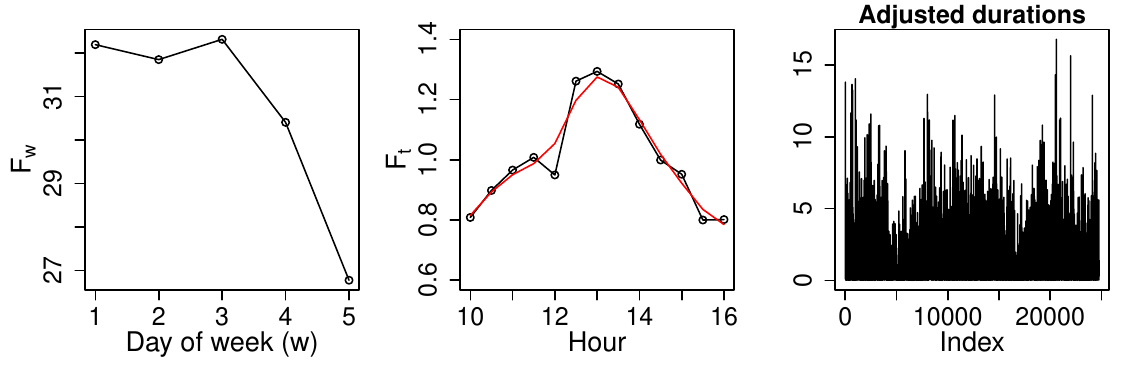}
\caption{IBM data. Left: Sample average duration $F_w$ for day $w$ of the week. Middle: Sample average duration at each knot (black) and smoothed value (red). Right: Adjusted durations. \label{Fig_seasonality}}
\end{figure}
The benchmark IBM transactions data is extracted from the TORQ data built by Joel Hasbrouck and NYSE \citep{Tsay2010}. It has the time of occurrence (measured in seconds from midnight) of each transaction from 1 Nov 1990 to 31 Jan 1991, and we analyze the duration between consecutive transactions, $d_i$, using the stochastic conditional duration model. Only transactions within normal trading hours (9.30 a.m to 4 p.m.) from 1 Nov to 21 Dec are considered to avoid holiday effects, and 23 Nov is excluded due to a halt in trading. Due to the unusually high concentration of short durations immediately after opening, \cite{Engle1998} recommend removing transactions before 10 a.m. and taking the first duration as the average duration of transactions from 9.50 a.m. to 10 a.m. After this procedure and the removal of zero durations, we have $n=24765$ observations. Following \cite{Feng2004}, we remove seasonality due to day-of-week effect (trading increases in frequency towards end of week) and time-of-day effect (trading is more frequent at the start and end, than in middle of the day) before analysis. To remove day-of-week effect, we compute the average duration $F_w$ for each weekday $w$ and set $\tilde{d}_i = d_i/F_w$ if $i$ falls in the weekday $w$. For time-of-week effect, we construct 13 knots at 30 minutes interval from 10 a.m. to 4 p.m., $k_1 =36000, \dots, k_{13} = 57600$. The value at knot $k_s$ is the average of durations whose transaction time falls in $(k_s-900, k_s +900]$. The smoothed value $F_s$ is computed using piecewise cubic splines via the {\tt smooth.spline} function in R, and the adjusted duration is $\tilde{d_i}/F_s$ if transaction time of $d_i$ falls in $(k_s-900, k_s +900]$. See Figure \ref{Fig_seasonality} for illustration.

Fitting the stochastic conditional duration model to the adjusted durations, the parameter estimates and inefficiency factors are shown in Table \ref{tab_paribm}. The simulation efficiency of the CP is better than NCP for $\mu$ and $\phi$, but worse for $\sigma_\eta$. ASIS improves the simulation efficiency for $\sigma_\eta$ and $\phi$ significantly but BSR produces even better results. The trace plots in Figure \ref{Fig_IBMplots} also indicate that there is better mixing in the MCMC chains of  $\sigma_\eta$ and $\phi$ for BSR than ASIS (sample autocorrelations also decay faster).
\begin{table*}[tb!]
\caption{Parameter estimates and inefficiency factors from stochastic conditional duration model for IBM data.}
\label{tab_paribm}
\begin{tabular}{@{}llrrrr@{}}
\hline
&& NCP & CP & ASIS & BSR \\   \hline
\multirow{3}{*}{\parbox[t]{1.4cm}{Parameter estimates}} & $\mu$ & $-0.112$ $\pm$ 0.022 & $-0.119$ $\pm$ 0.023 & $-0.119$ $\pm$ 0.023 & $-0.119$ $\pm$ 0.023 \\
& $\sigma_\eta$ &    0.149 $\pm$ 0.008 & 0.150 $\pm$ 0.008 & 0.150 $\pm$ 0.008 & 0.150 $\pm$ 0.008 \\
& $\phi$ &  0.957 $\pm$ 0.004 & 0.956 $\pm$ 0.004 & 0.957 $\pm$ 0.004 & 0.956 $\pm$ 0.004 \\ \hline 
\multirow{3}{*}{\parbox[t]{1.4cm}{Inefficiency factor}} & $\mu$ &   132 & 3 & 2 & 2 \\
& $\sigma_\eta$ &    259 & 311 & 135 & 88 \\
& $\phi$ &    207 & 186 & 107 & 68 \\    \hline
\end{tabular}
\end{table*}

\begin{figure}[tb!]
\centering 
\includegraphics[width=0.48\textwidth]{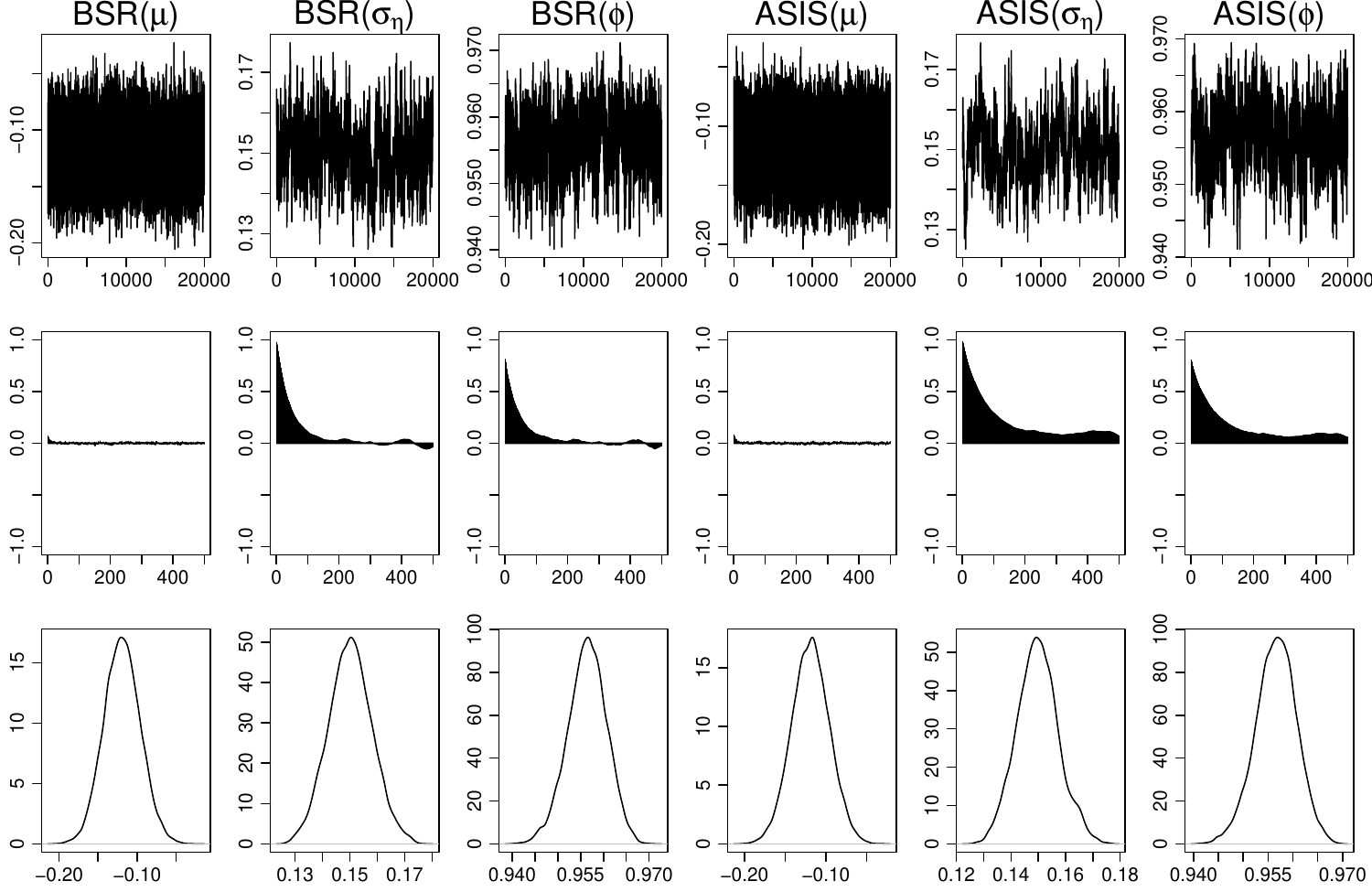}
\caption{Left: IBM data. First row: Trace plots. Second row: Sample autocorrelation functions up to lag 500. Third row: Estimated marginal posterior densities. \label{Fig_IBMplots}}
\end{figure}

\section{Conclusion}
This article investigates efficient data augmentation techniques for some classes of state space models by introducing working parameters $\{a, \bfw\}$ for rescaling and recentering of the latent states. First, we focus on maximum likelihood estimation of the AR(1) plus noise model via the EM algorithm and derived optimal values of $\{a, \bfw\}$ by minimizing the fraction of missing information. An alternating expectation-conditional maximization algorithm that allows the optimal parametrization for each parameter to be used while conditioning on others is then designed. Experiments indicate that this algorithm is often able to outperform the CP and NCP in rate of convergence and maximized likelihood. The proposed data augmentation scheme is then utilized to design efficient MCMC algorithms for Bayesian inference of the SV and stochastic conditional duration models through a mixture of normals approximation. A BSR strategy which allows the data augmentation scheme to be tailored to each block is proposed for a MCMC sampler. Experimental results indicate that BSR works well, surpassing ASIS in all real applications.

These encouraging results motivate the application of BSR to other classes of state space models where the parametrization greatly affects the convergence of the MCMC sampler. Consider the multivariate factor SV model of \cite{Pitt1999c}, where each of the factor terms and idiosyncratic error terms follow independent SV processes. As the standard full conditional sampler converges slowly, \cite{Kastner2017} employ ASIS to improve convergence. We can also apply BSR to sample the parameters of each SV process via the mixture of normals approximation, by allowing each SV process to have its own set of working parameters. Optimal values of the working parameters for each SV process can likely be deduced from this article by conditioning on other variables. For parsimony, one may also restrict $\bfw_t$ to be constant across time points which is feasible when the number of observations is large (Theorem \ref{thm_alg1_largen}). \cite{Li2020} further apply ASIS to factor SV models with leverage, asymmetry and heavy tails, and the BSR strategy is likely to be useful here as well although it is challenging to derive optimal values of the working parameters analytically. It is worth exploring if plugging in values based on Gaussian approximations of such models will yield improvements compared to the NCP and CP.

Another application is to the time-varying parameter models investigated by \cite{Bitto2019}, which resemble regression models with time-varying coefficients following a random walk. To prevent overfitting, the double gamma shrinkage prior was placed on process variances to reduce time-varying coefficients to statics ones. The parametrization is important here as the CP is preferred for time-varying coefficients while the NCP is preferred for static ones, and \cite{Bitto2019} use ASIS to ``combine the best of both worlds". In this case, it is possible to partially rescale and recenter the time-varying coefficients using fixed mean and standard deviation parameters. If each coefficient is allowed to have its own set of working parameters, then BSR offers the flexibility of using the optimal parametrization for each coefficient and hence it can potentially perform as well as or even better than ASIS. However, the random walk is a non-stationary process and the optimal working parameters will have to be re-derived using say the EM algorithm as a starting point (as in this article), or some other criteria such as minimization of the posterior correlation between blocks. 

\cite{Kreuzer2020} consider nonlinear state space models where the latent states follow an AR(1) process and apply ASIS to improve the efficiency of the proposed elliptical slice sampler. They consider an NCP of the form $\tilde{x}_{t+1} = (x_{t+1} - \mu - \phi(x_{t} - \mu))/\sigma_\eta$ instead of $\tilde{x}_{t+1} = (x_{t+1} - \mu)/\sigma_\eta$, and it is worthwhile to contemplate the construction of an associated PNCP which directly impacts the convergence rate of $\phi$. Finally, it is also interesting to explore how the proposed data augmentation scheme can be applied to construct efficient variational approximation algorithms for state space models whose rate of convergence is also strongly dependent on the choice of parametrization.

%%%%%%%%%%%%%%%%%%%%%%%%%%%%%%%%%%%%%%%%%%%%%%
%% Single Appendix:                         %%
%%%%%%%%%%%%%%%%%%%%%%%%%%%%%%%%%%%%%%%%%%%%%%
%\begin{appendix}
%\section*{???}%% if no title is needed, leave empty \section*{}.
%\end{appendix}
%%%%%%%%%%%%%%%%%%%%%%%%%%%%%%%%%%%%%%%%%%%%%%
%% Multiple Appendixes:                     %%
%%%%%%%%%%%%%%%%%%%%%%%%%%%%%%%%%%%%%%%%%%%%%%
%\begin{appendix}
%\section{???}
%
%\section{???}
%
%\end{appendix}

%%%%%%%%%%%%%%%%%%%%%%%%%%%%%%%%%%%%%%%%%%%%%%
%% Support information, if any,             %%
%% should be provided in the                %%
%% Acknowledgements section.                %%
%%%%%%%%%%%%%%%%%%%%%%%%%%%%%%%%%%%%%%%%%%%%%%
\begin{acks}[Acknowledgments]
The author would like to thank the Editor, Associate Editor and three referees for their comments and helpful suggestions which have improved this manuscript greatly. The author also wish to thank Dion Kwan and Robert Kohn for comments and discussion on this work.
\end{acks}
%%%%%%%%%%%%%%%%%%%%%%%%%%%%%%%%%%%%%%%%%%%%%%
%% Funding information, if any,             %%
%% should be provided in the                %%
%% funding section.                         %%
%%%%%%%%%%%%%%%%%%%%%%%%%%%%%%%%%%%%%%%%%%%%%%
\begin{funding}
The author was supported by the start-up grant (R-155-000-190-133).
\end{funding}

%%%%%%%%%%%%%%%%%%%%%%%%%%%%%%%%%%%%%%%%%%%%%%
%% Supplementary Material, including data   %%
%% sets and code, should be provided in     %%
%% {supplement} environment with title      %%
%% and short description. It cannot be      %%
%% available exclusively as external link.  %%
%% All Supplementary Material must be       %%
%% available to the reader on Project       %%
%% Euclid with the published article.       %%
%%%%%%%%%%%%%%%%%%%%%%%%%%%%%%%%%%%%%%%%%%%%%%
\begin{supplement}
\stitle{Zip file}
\sdescription{Julia code and a pdf file containing derivations, proofs and additional details.} 
\end{supplement}

%%  MR numbers will be added by VTeX.                      %%
%%                                                         %%
%%  Use \cite{...} to cite references in text.             %%
%%                                                         %%
%%%%%%%%%%%%%%%%%%%%%%%%%%%%%%%%%%%%%%%%%%%%%%%%%%%%%%%%%%%%%

%% if your bibliography is in bibtex format, uncomment commands:
\bibliographystyle{imsart-nameyear} % Style BST file (imsart-number.bst or imsart-nameyear.bst)
\bibliography{ref}       % Bibliography file (usually '*.bib')

%% or include bibliography directly:
% \begin{thebibliography}{}
% \bibitem{b1}
% \end{thebibliography}

\newpage 
\begin{frontmatter}
\title{Supplement to ``Efficient data augmentation techniques for some classes of state space models"}
%\title{A sample article title with some additional note\thanksref{T1}}
\runtitle{Supplement}
%\thankstext{T1}{A sample of additional note to the title.}

\end{frontmatter}

\setcounter{section}{0} \renewcommand{\thesection}{S\arabic{section}}
\setcounter{figure}{0} \renewcommand{\thefigure}{S\arabic{figure}}
\setcounter{table}{0} \renewcommand{\thetable}{S\arabic{table}}
\setcounter{equation}{0} \renewcommand{\theequation}{S\arabic{equation}}
\setcounter{corollary}{0} \renewcommand{\thecorollary}{S\arabic{corollary}}
\setcounter{lemma}{0} \renewcommand{\thelemma}{S\arabic{lemma}}
\setcounter{theorem}{0} \renewcommand{\thetheorem}{S\arabic{theorem}}
\setcounter{property}{0} \renewcommand{\theproperty}{S\arabic{property}}
\setcounter{algorithm}{0} \renewcommand{\thealgorithm}{S\arabic{algorithm}}

\section{Observed information matrix} \label{sec_obs_info_mat}
For the partially noncentered state space model, 
\begin{equation*} 
\begin{gathered}
\bfy | \bfalpha, \bftheta  \sim \N(\sigma_\eta^a \bfalpha + \mu \bfw, \sigma_\epsilon^2 I), \\
\sigma_\eta^{a} \bfalpha | \bftheta \sim \N( \mu \bar{\bfw}, \sigma_\eta^{2}\Lambda^{-1} ),
\end{gathered}
\end{equation*}
the marginal distribution of $\bfy$ is $\N(\mu \bone, S)$ as
\begin{equation*}
\begin{aligned}
\E(\bfy) 
&= \E[\E(\bfy|\bfalpha)] = \E(\sigma_\eta^a \bfalpha + \mu \bfw) \\
&= \mu \bar{\bfw} + \mu \bfw = \mu \bone, \\
\var(\bfy) 
&= \var[\E(\bfy|\bfalpha)] + \E[\var(\bfy|\bfalpha)] \\
&= \sigma_\eta^2 \Lambda^{-1} + \sigma_\epsilon^2 I =  S.
\end{aligned}
\end{equation*}
Hence the log-likelihood of $\bfy$ is
\begin{multline*}
L(\bftheta) = \log p(\bfy|\bftheta) \\
= -\tfrac{n}{2}\log (2\pi) -\tfrac{1}{2}\log|S| - \tfrac{1}{2} (\bfy-\mu \bone)^T S^{-1} (\bfy-\mu \bone). 
\end{multline*}
We compute $L(\bftheta^{(i)})$ at the end of each iteration in the EM algorithm and the algorithm is terminated when the relative increment in $L(\bftheta)$ is negligible. We can compute $L(\bftheta^{(i)})$ efficiently by finding the $LDL^T$ decomposition of $(V_0/\sigma_\epsilon^2)^{-1}$ where $D$ is a diagonal matrix and $L$ is a unit lower triangular matrix, and noting that $S^{-1} =  \Lambda V_0/(\sigma_\eta^2 \sigma_\epsilon^2) = (I - V_0/\sigma_\epsilon^2)/\sigma_\epsilon^2$. Further useful idenities are given in Appendix \ref{Appendix A}. Hence 
\[
\log |S^{-1}| = \log(1-\phi^2) - n\log(\sigma_\eta^2) - \log|D|,
\]
and $S^{-1} (\bfy - \mu\bone)$ can be evaluated efficiently using back-substitution. 
Note that $S^{-1} (\bfy - \mu \bone) = \Lambda \bfm_{0\bf1}/\sigma_\eta^2$. We have 
\begin{align*}
\nabla_\mu L &= (\bfy-\mu \bone)^T S^{-1} \bone, \qquad
\nabla_\mu^2 L = - \bone^T S^{-1} \bone, \\
\nabla_{\sigma_\eta^2} L &= \tfrac{1}{2}\{ (\bfy-\mu \bone)^T S^{-1} \Lambda^{-1} S^{-1} (\bfy-\mu \bone) \\
& \quad - \tr(S^{-1}\Lambda^{-1})\} , \\
\nabla_{\sigma_\eta^2}^2 L &= \tr(V_0^2)/(2\sigma_\eta^4 \sigma_\epsilon^4)  - \bfm_{0\bf1}^T S^{-1}\bfm_{0\bf1}/\sigma_\eta^4,  \\
\nabla_{\mu, \sigma_\eta^2}^2 L &= - \bfm_{0\bf1}^T S^{-1} \bone/\sigma_\eta^2.
\end{align*}
Hence $I_{\mu, \mu}^\obs = \bone^T S^{-1} \bone$, $I_{\mu, \sigma_\eta^2}^\obs = { \bfm_{0\bf1}^T S^{-1} \bone}/{\sigma_\eta^2}$ and $I_{\sigma_\eta^2, \sigma_\eta^2}^\obs = [2{\bfm_{0\bf1}^T S^{-1}\bfm_{0\bf1}} - {\tr(V_0^2)} /\sigma_\epsilon^4]/{(2\sigma_\eta^4)}$.

\section{Updates in EM algorithm}
Let $\Lambda_1= \partial \Lambda/\partial \phi$, which is a symmetric tridiagonal matrix with diagonal $(0, 2\phi, \dots, 2\phi, 0)^T$ and off-diagonal elements equal to $-1$. Let
\[
U = \sigma_\eta^{2a} V_a^{(i)} \\+ (\sigma_\eta^a \bfm_{a\bfw}^{(i)} -\mu \bar{\bfw})(\sigma_\eta^a \bfm_{a\bfw}^{(i)} -\mu \bar{\bfw})^T.
\] 
The first derivative of $Q(\bftheta|\bftheta^{(i)})$ with respect to each parameter is given below. 
\begin{equation*}
\begin{aligned}
\nabla_\mu Q(\bftheta | \bftheta^{(i)}) 
&= \frac{\bfw^Tz^{(i)} }{\sigma_\epsilon^2} + \frac{(\sigma_\eta^a \bfm_{a\bfw}^{(i)} - \mu \bar{\bfw})^T\Lambda \bar{\bfw}}{\sigma_\eta^2}. \\
\nabla_\phi Q(\bftheta | \bftheta^{(i)}) 
&= - \frac{\tr(\Lambda_1 U)}{2\sigma_\eta^2} - \frac{\phi}{1-\phi^2} \\
&= - \frac{1}{\sigma_\eta^2} \bigg[ \phi \sum_{t=2}^{n-1} U_{tt} - \sum_{t=1}^{n-1} U_{t, t+1} \bigg] - \frac{\phi}{1-\phi^2}. \\
\nabla_{\sigma_\epsilon^2} Q(\bftheta | \bftheta^{(i)}) 
&= \frac{\sigma_\eta^{2a} \tr(V_a^{(i)}) + z^{(i)T} z^{(i)}}{2\sigma_\epsilon^4} - \frac{n}{2\sigma_\epsilon^2}. \\
\end{aligned}
\end{equation*}
\begin{multline*}
\nabla_{\sigma_\eta^2} Q(\bftheta | \bftheta^{(i)}) = \frac{1}{2\sigma_\eta^4} \Big\{   a\sigma_\eta^{a+2}\sigma_\epsilon^{-2} (\bfy- \mu \bfw)^T  \bfm_{a\bfw}^{(i)} \\
+ (1-a)\sigma_\eta^{2a} [ \tr(\Lambda V_a^{(i)}) + \bfm_{a\bfw}^{(i)T} \Lambda \bfm_{a\bfw}^{(i)} ]  + \mu^2 \bar{\bfw}^T \Lambda \bar{\bfw} \\
+ n(a-1)\sigma_\eta^2+ (a-2) \mu \sigma_\eta^a \bfm_{a\bfw}^{(i)T} \Lambda \bar{\bfw} \\
- a \sigma_\eta^{2a+2}\sigma_\epsilon^{-2}[ \tr(V_a^{(i)}) + \bfm_{a\bfw}^{(i)T} \bfm_{a\bfw}^{(i)} ] \Big\}. 
\end{multline*}
Setting $\nabla_{\theta_s} Q(\bftheta | \bftheta^{(i)}) = 0$ then leads to the update for $\theta_s$. When $\bftheta^{(i)} = \bftheta$ at convergence, we can use the identities in Appendix \ref{Appendix A} to obtain
\begin{align*}
\nabla_\mu Q(\bftheta|\bftheta) &= (\bfy - \mu \bone)^T S^{-1} \bone, \\
\nabla_{\sigma_\eta^2} Q(\bftheta | \bftheta) &= \tfrac{1}{2\sigma_\eta^2} \big\{  n(a-1) - a \tr(\tfrac{V_0}{\sigma_\epsilon^2}) +(1-a)\tr(V_0 \tfrac{ \Lambda }{\sigma_\eta^2}) \\
& \quad + \bfm_{0\bone}^T \tfrac{ \Lambda }{\sigma_\eta^2} \bfm_{0\bone}  \big\}.  \\
\nabla_\phi Q(\bftheta | \bftheta) 
&= - \frac{\tr[\Lambda_1 (V_0 + \bfm_{0\bf1} \bfm_{0\bf1}^T)]}{2\sigma_\eta^2} - \frac{\phi}{1-\phi^2} \\
\nabla_{\sigma_\epsilon^2} Q(\bftheta | \bftheta) 
&= \frac{\tr(V_0) + (\bfy - \bfm_{0\bfzero})^T(\bfy - \bfm_{0\bfzero})}{2\sigma_\epsilon^4} - \frac{n}{2\sigma_\epsilon^2}. \\
\end{align*}
Setting $\nabla_\mu Q(\bftheta|\bftheta) = 0$ yields the MLE,
\begin{equation} \label{mueq}
\begin{aligned}
\mu = \frac{\bfy^T S^{-1}\bone}{\bone^T S^{-1} \bone}.
\end{aligned}
\end{equation}
Setting $\nabla_{\sigma_\eta^2} Q(\bftheta|\bftheta) = 0$, we obtain 
\begin{equation*}
\sigma_\eta^2 = \frac{(1-a) \tr(V_0  \Lambda)  + \bfm_{0\bone}^T  \Lambda \bfm_{0\bone} }{n(1-a) + a \tr(V_0)/\sigma_\epsilon^2 },
\end{equation*}
which leads to 
\begin{align} \label{sigetaeq}
n \sigma_\eta^2 &=\tr(V_0  \Lambda)  + \bfm_{0\bone}^T  \Lambda \bfm_{0\bone}  &&(a=0),\\
\gamma \tr(V_0) &= \bfm_{0\bone}^T  \Lambda \bfm_{0\bone}  &&(a=1). \nonumber 
\end{align}
These two formulas are equivalent since 
\[
\gamma \tr(V_0)  = \sigma_\eta^2 \tr[ (V_0^{-1} - \frac{ \Lambda }{\sigma_\eta^2}) V_0] = n \sigma_\eta^2 - \tr(V_0  \Lambda).
\]
Setting $\nabla_\phi Q(\bftheta|\bftheta) = 0$, 
\begin{equation} \label{phieq}
2 \phi \sigma_\eta^2 = - (1-\phi^2) [\tr(\Lambda_1V_0) + \bfm_{0\bone}^T\Lambda_1\bfm_{0\bone}].
\end{equation}
Setting $\nabla_{\sigma_\epsilon^2} Q(\bftheta|\bftheta) = 0$,
\begin{equation} \label{sigepseq}
n\sigma_\epsilon^2 =  (\bfy - \bfm_{0\bfzero})^T(\bfy - \bfm_{0\bfzero}) +  \tr(V_0).
\end{equation}

We use $I_\augm(\bftheta^*)$ to optimize the rate of convergence of the EM algorithm with respect to $a$ and $\bfw$, where
\begin{equation*}
\begin{aligned}
I_{\augm}(\bftheta^*) &= - \int p(\bfalpha|\bfy,\bftheta)  \nabla^2_\bftheta \log p(\bfy, \bfalpha|\bftheta) \, d\bfalpha  \big|_{\bftheta= \bftheta^*} \\
&= - \nabla^2_\bftheta \int p(\bfalpha|\bfy,\bftheta^{(i)}) \log p(\bfy, \bfalpha|\bftheta) \, d\bfalpha \big|_{\bftheta, \bftheta^{(i)}= \bftheta^*} \\
&= -\nabla^2_\bftheta Q(\bftheta|\bftheta^{(i)}) \big|_{\bftheta, \,\bftheta^{(i)}= \bftheta^*}.
\end{aligned}
\end{equation*}
Note that $\nabla^2_\bftheta$ applies only to $\log p(\bfy, \bfalpha|\bftheta)$. The above expression shows that $I_\augm(\bftheta^*)$ measures the curvature of $Q(\bftheta|\bftheta^{(i)})$ at the MLE $\bftheta^*$.

\section{Unknown location parameter}
In this section, we consider the location parameter $\mu$ as the only unknown parameter. 
Since
\[
\nabla_\mu^2 Q(\bftheta | \bftheta^{(i)}) = - (\sigma_\epsilon^{-2} \bfw^T \bfw+ \sigma_\eta^{-2} \bar{\bfw}^T\Lambda \bar{\bfw}), 
\]
we have
\[
I_{ \mu, \mu} = \sigma_\epsilon^{-2} \bfw^T \bfw+ \sigma_\eta^{-2} \bar{\bfw}^T\Lambda \bar{\bfw} = \tau(\bfw).
\]

\begin{proof}[Proof of Theorem 1]
From Section S1, $I_\obs = \bone ^T S^{-1} \bone$. Since $I_{\augm} = \tau(\bfw)$, the rate of convergence is 
\[
1 - I_\obs I_{\augm}^{-1} = 1 -  \bone ^T S^{-1} \bone/\tau(\bfw) = \rho(\bfw)^T V_0 \rho(\bfw)/\tau(\bfw).
\]
Since $\rho(\bfw) = 0$ only at $\bfw = \sigma_\eta^{-2} V_0 \Lambda \bone$, the convergence rate is minimized to zero at $\bfw^\opt$. 
\end{proof}

From (10), $\bar{\bfw}^\opt = \sigma_\epsilon^{-2} V_0 \bone$ is simply the row sums of $V_0$ divided by $\sigma_\epsilon^2$. Recall that $V_0^{-1} = \sigma_\epsilon^{-2} I + \sigma_\eta^{-2}\Lambda$ and $\sigma_\eta^2 V_0^{-1} = \gamma I + \Lambda$. If $\phi \neq 0$, let 
\[
Q = \sigma_\eta^{2} V_0^{-1}/ |\phi|,
\]
so that $Q$ is a symmetric tridiagonal matrix with diagonal $(d, c, \dots, c, d)^T$  and off-diagonal elements equal to $-b$, where $b = \phi/|\phi|$, $d = (\gamma + 1)/|\phi|$ and $c = d + |\phi|$. Note that $b=\pm 1$, $d > 1$ and $c > 2$. From Corollary 2, Property 2 and Theorems 2 and 5 of \cite{Tan2019}, we have the following properties of $Q^{-1}$ respectively.

\begin{property} \label{prop1}
If $\phi \neq 0$, $Q^{-1}_{ij} = u_i v_j$ for $i \leq j$, where 
\begin{equation*}
\begin{gathered}
v_i = b^{i-1}\kappa_{n-i}/{\kappa}, \quad u_i = b^{i-1} \kappa_{i-1}/\kappa_0, \quad (i=1, \dots, n) ,\\
\kappa_i = \vp \rpp^i - \vm \rmm^i, \quad (i=0, \dots, n-1), \\
\kappa = \vp^2\rpp^{n-1}  -  \vm^2\rmm^{n-1}, \\
\varphi_{\pm} = r_{\pm} - |\phi|, \quad  r_{\pm} = (c \pm \sqrt{c^2-4})/2.
\end{gathered}
\end{equation*}
\end{property}
Note that $\rpp\rmm = 1$, $\rpp + \rmm = c$ and $\kappa> 0$.

\begin{property} \label{prop2}
If $0 < \phi < 1$, $u_i$ and $v_i$ are positive for $i=1, \dots, n$, and all elements of $Q^{-1}$ are positive.
\end{property}

\begin{property}\label{prop3}
The sum of the $i$th row of $Q^{-1}$ is 
\[
s_i = \frac{1 - b(1-\phi) (v_i + v_{n-i+1}) }{c - 2b}.
\]
If $0 < \phi <1$, 
\[
\frac{1}{d-\phi} < s_i < \frac{1}{c-2}.
\]
If $-1 < \phi < 0$, 
\[
\frac{2}{c+2} - \frac{1}{d+\phi} < s_i < \frac{1}{d + \phi}.
\]
\end{property}

\begin{theorem} \label{thm_wopt}
Let $w_t^{\opt}$ be the $t^\text{th}$ element of $\bfw^\opt$. If $\phi=0$,  $w_t^{\opt} = (1+\gamma)^{-1}$ and if $ \phi \neq 0$, 
\begin{equation} \label{eq_wopt}
w_t^{\opt} =  \frac{(1-\phi)^2 + b \gamma (1 - \phi)(v_t + v_{n-t+1})}{ (1-\phi)^2 + \gamma}.
\end{equation}
\end{theorem}
\begin{proof}
If $\phi=0$, $\Lambda = I$ and $V_0 = \sigma_\eta^{2}/(1+\gamma)I$. Hence $\bfw^\opt = 1/(1+\gamma) \bone$. 
If $\phi \neq 0$, $\bar{\bfw}^\opt = \gamma/|\phi| Q^{-1} \bone$ and $w_t^{\opt} = 1 -  \gamma s_t/|\phi|$. Substituting the expression for $s_t$ from Property \ref{prop3} and noting that $|\phi|(c-2b) = (1-\phi)^2 + \gamma $ yields the result in \eqref{eq_wopt}. 
\end{proof}

As $u_1 = 1$ and $Q_{1j}^{-1} = v_j$ from Property \ref{prop1}, Theorem \ref{thm_wopt} implies that $\bfw^\opt$ can be computed using only elements from the first row of $Q^{-1}$. In addition, $\bfw^\opt$ is symmetric since $w_t^\opt = w_{n-t+1}^\opt$ for each $t$.

\begin{proof}[Proof of Corollary 1]
From Theorem \ref{thm_wopt}, $w_t^\opt = (1+\gamma)^{-1}$ if $\phi=0$ which is the value of the lower and upper bounds. If $0 < \phi < 1$, $b=1$ and $v_t > 0$ for $t=1, \dots, n$ from Property \ref{prop2}. Hence $w_t^\opt > 1 - B_1$ from \eqref{eq_wopt}. From Property \ref{prop3}, $s_t > (d-\phi)^{-1}$ implies that $w_t^\opt = 1 -  \gamma s_t/|\phi|< 1- {\gamma}/\{\phi (d-\phi)\}= 1-B_2$. If $-1 < \phi < 0$, $b=-1$ and $\phi=-|\phi|$. From Property \ref{prop3}, $2/(c+2) - 1/(d+\phi) < s_t < 1/(d + \phi)$ implies $1 - B_2 \leq w_t^\opt \leq 1 + B_2 - 2B_1 $. When $\phi \in [0,1)$, $ 0 < B_2 \leq B_1 < 1$ and when $\phi \in (0,1)$, $0 < B_1 < B_2 < 1$.
\end{proof}

The proof of Corollary 2 depends on Lemmas \ref{lem_S1} and \ref{lem_S2}. 
\begin{lemma} \label{lem_S1}
If $0 < \phi <1$, then 
\[
cv_{i} = v_{i-1}+ v_{i+1}, \quad i = 2, \dots, n-1.
\]
\end{lemma}
\begin{proof}
First, 
\begin{equation*}
\begin{aligned}
\kappa_{i-1}+ \kappa_{i+1} &= \vp \rpp^{i-1} - \vm \rmm^{i-1} + \vp \rpp^{i+1} - \vm \rmm^{i+1} \\
&= \vp \rpp^{i} (1/\rpp + \rpp) - \vm \rmm^{i} (1/\rmm + \rmm) \\
&= \vp \rpp^{i} (\rmm + \rpp) - \vm \rmm^{i} (\rpp + \rmm) \\
&=  c\kappa_i.
\end{aligned}
\end{equation*}
As $b=1$, $c v_{i} = c \kappa_{n-i}/ \kappa = (\kappa_{n-i-1} + \kappa_{n-i+1})/\kappa = v_{i-1} + v_{i+1}$.
\end{proof}

\begin{lemma} \label{lem_S2}
If $0 < \phi <1$, then 
\[
2\phi s_t - s_{t-1} - s_{t+1} < 0.
\]
\end{lemma}
\begin{proof}
As $b=1$, from Property \ref{prop3},
\[
(c - 2)s_t = 1 - (1-\phi) (v_t + v_{n-t+1}).
\]
From Lemma \ref{lem_S1},
\begin{align*}
(c-2)& (2\phi s_t - s_{t-1} - s_{t+1}) \\
&= 2\phi - 2 - (1-\phi)\{ 2\phi (v_t + v_{n-t+1}) \\\
& \quad - (v_{t-1} + v_{n-t+2}) - (v_{t+1} + v_{n-t})\} \\
&= 2\phi - 2 - (1-\phi)(2\phi - c )(v_t + v_{n-t+1}) \\
&= 2\phi - 2 - (2\phi - c )\{ 1 - (c-2) s_t \} \\
&= (c-2)\{ 1 + (2\phi-c) s_t \}.
\end{align*}
From Property \ref{prop3}, $s_t >1/(d-\phi)$ implies that 
\[
1+ (2\phi-c) s_t < 1 + \frac{2\phi-c}{d-\phi}  = 0.
\]
Hence $2\phi s_t - s_{t-1} - s_{t+1} < 0$.
\end{proof}

\begin{proof}[Proof of Corollary 2]
Writing
\[
\bfw^\opt = (\gamma \Lambda^{-1} + I) ^{-1}\bone,
\]
\begin{align*}
\nabla_\gamma \bfw^\opt &= - Q^{-1} \bfw^\opt/|\phi|, \\ 
\nabla_\phi \bfw^\opt &= Q^{-1} \Lambda_1 \bar{\bfw}^\opt/|\phi|.
\end{align*}
From Property \ref{prop2} and Corollary 1, all elements of $Q^{-1}$ and $\bfw^\opt$ are positive. Hence each element of $\nabla_\gamma \bfw^\opt$ is negative and $w_t^\opt$ decreases strictly with $\gamma$ for all $t$. To show that each element of $\nabla_\phi \bfw^\opt $ is negative, it suffices to show that each element of $\Lambda_1 \bar{\bfw}_\opt$ (a symmetric vector) is negative. The first and last elements are equal to $-{\bar{w}^\opt}_2$, which is negative from Corollary 1. For $t=2, \dots, n-1$, the $t^\text{th}$ element of $\Lambda_1\bar{\bfw}_\opt$ is 
\begin{equation*}
\begin{aligned}
2\phi \bar{w}_t^\opt - \bar{w}_{t-1}^\opt - \bar{w}_{t+1}^\opt &=\frac{\gamma}{\phi} [2\phi s_t - s_{t-1} -s_{t+1}] < 0        
\end{aligned}                                  
\end{equation*}
from Lemma \ref{lem_S2}. It is clear from \eqref{eq_wopt} that as $\phi$ approaches 1, $\bfw^\opt$ approaches $\bf0$.
\end{proof}

The proof of Theorem 2 depends on Lemma \ref{lem_S3}.
\begin{lemma}\label{lem_S3}
Let $u^* = \gamma/(1-\phi)^2$. If $\phi \neq 0$, 
\[
\lim_{n \rightarrow \infty}  \frac{1}{n}\sum_{t=1}^n w_t^{\opt} = \frac{1}{1+ u^*},
\]
\end{lemma}
\begin{proof}
From the expression of $w_t^{\opt}$ in Theorem \ref{thm_wopt}, 
\begin{equation*}
\begin{aligned}
\frac{1}{n}\sum_{t=1}^n w_t^{\opt} =  \frac{1}{ 1+ u^*} + \frac{ 2b \gamma (1 - \phi) }{ n[(1-\phi)^2 + \gamma]} \sum_{t=1}^n v_t,
\end{aligned}
\end{equation*}
where
\begin{align*}
\kappa \sum_{t=1}^n v_t &= \sum_{t=1}^n  b^{t-1} (\vp \rpp^{n-t} - \vm \rmm^{n-t}) \\
&= \frac{\vp(b^n - \rpp^n)}{b - \rpp} - \frac{\vm(b^n - \rmm^n)}{b - \rmm} \\
&= [\vp(b^{n+1} - b^n \rmm -b \rpp^n + \rpp^{n-1}) \\
& \quad - \vm(b^{n+1} - b^n \rpp - b \rmm^n + \rmm^{n-1})]/(2-bc).
\end{align*}
Thus 
\begin{equation*}
\begin{aligned}
\frac{1}{n} \sum_{t=1}^n v_t 
&= \frac{\vp (R_1^{n - 1} - R_1^{n-2} R_2 -b \rpp + 1)}{n(2-bc)\left(\vp^2 - \vm^2 R_2^{n-1} \right) } \\
& \quad - \frac{\vm (R_1^{n-1} - R_1^{n-2} - b \rmm  R_2^{n-1} +R_2^{n-1})}{n(2-bc)\left(\vp^2 - \vm^2 R_2^{n-1} \right) } \\
\end{aligned}
\end{equation*}
where $R_1 = b/\rpp$ and $R_2= \rmm/\rpp$. Since $|R_1| <1$ and $|R_2| <1$, $\lim_{n \rightarrow \infty} \sum_{t=1}^n v_t /n = 0$ and $\lim_{n \rightarrow \infty} \sum_{t=1}^n w_t^{\opt}/n  = 1/ (1+ u^*)$. 
\end{proof}

\begin{proof}[Proof of Theorem 2]
If $\bfw = w\bone$, 
\begin{align*}
\nabla_w \tau(w\bone) &= 2[(\bone^T V_0^{-1} \bone) w - \bone^T \Lambda \bone/\sigma_\eta^2] \\
\nabla_w^2 \tau(w\bone) &= 2(\bone^T V_0^{-1} \bone) > 0.
\end{align*}
Hence $\tau(w\bone)$ is minimized at $w^\opt  = 1/(1+u)$ where $u = n\gamma /(\bone^T \Lambda \bone)$, and $\tau(w^\opt \bone) = n/[\sigma_\epsilon^2 (1+u)]$. 
Since $\bone^T \Lambda \bone = n(1-\phi)^2 + 2\phi(1-\phi)$, $u \rightarrow u^*$ and hence $w^\opt \rightarrow 1/(1+u^*)$ as $n \rightarrow \infty$. The rate of convergence of Algorithm 1 at $w^\opt$ is 
\[
1 - \sigma_\epsilon^2 (1+u) \bone^T S^{-1} \bone/n = 1 - (1+u) (\bone^T \bfw^\opt)/n,
\]
since $\sigma_\epsilon^2 \bone^T S^{-1} \bone = \bone^T (I+\gamma\Lambda^{-1})^{-1} \bone = \bone^T \bfw^\opt$. From Lemma \ref{lem_S3}, $\lim_{n \rightarrow \infty}  (\bone^T \bfw^\opt)/n  = 1/ (1+ u^*)$. Hence the rate of convergence of Algorithm 1 goes to zero.
\end{proof}

Figure \ref{fig_wconstrate} shows how the rate of convergence of Algorithm 1 converges to zero as $n$ increases if $\bfw =w^\opt \bone$ for $\phi \in \{-0.9, 0.9\}$ and $\gamma \in \{0.1, 1, 10\}$.
\begin{figure}[htb!]
\includegraphics[width=0.48\textwidth]{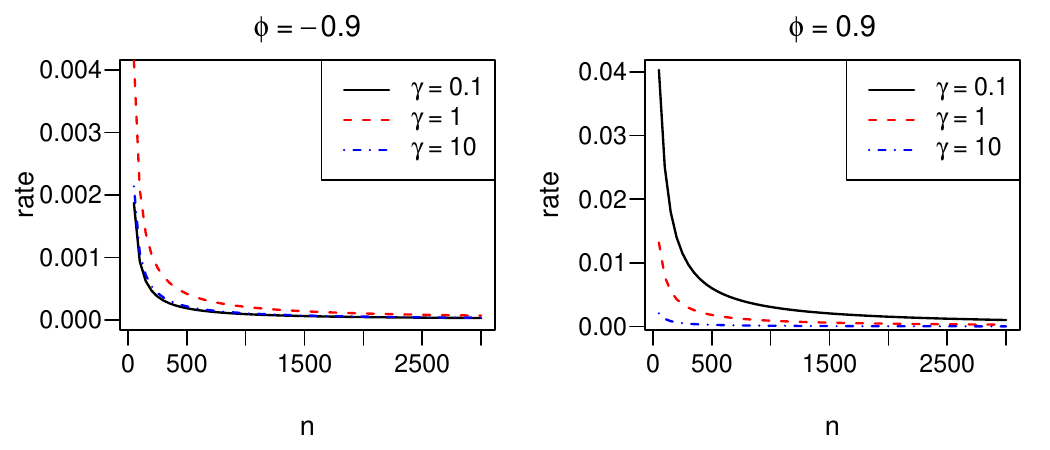}
\caption{Convergence rate of Algorithm 1 when $\bfw = w^\opt\bone$.}
\label{fig_wconstrate}
\end{figure}

\section{Unknown scale parameter} \label{sec unknown scale parameter}
Suppose $\sigma_\eta^2$ is the only unknown parameter.
\begin{multline*}
\nabla_{\sigma_\eta^2}^2 Q(\bftheta | \bftheta^{(i)}) = \tfrac{1}{2\sigma_\eta^4} \Big\{a (\tfrac{a}{2} - 1) \sigma_\epsilon^{-2}  \sigma_\eta^{a}(\bfy- \mu \bfw)^T \bfm_{a\bfw}^{(i)} \\
+n(1-a) - a(a-1) \sigma_\epsilon^{-2}\sigma_\eta^{2a}  [ \tr(V_a^{(i)})  +{\bfm_{a\bfw}^{(i)}}^T\bfm_{a\bfw}^{(i)} ] 
\\
- 2 \mu^2 \bar{\bfw}^T \tfrac{\Lambda}{\sigma_\eta^2}\bar{\bfw} + (a-2)(\tfrac{a}{2} - 2)  \mu \sigma_\eta^a {\bfm_{a\bfw}^{(i)}}^T \tfrac{\Lambda}{\sigma_\eta^2} \bar{\bfw} \\
- (a-1)(a-2) \sigma_\eta^{2a}  [ \tr(\tfrac{\Lambda}{\sigma_\eta^2} V_a^{(i)}) + {\bfm_{a\bfw}^{(i)}}^T \tfrac{\Lambda}{\sigma_\eta^2} \bfm_{a\bfw}^{(i)} ]  \Big\}.
\end{multline*}
When $\bftheta = \bftheta^{(i)}$ (at convergence), we can use the identities in Appendix \ref{Appendix A} to obtain
\begin{equation} \label{Qsigmaeta}
\begin{aligned}
&\nabla_{\sigma_\eta^2}^2 Q(\bftheta | \bftheta) = - \tfrac{1}{2\sigma_\eta^4} \Big\{ a(a-1)\sigma_\epsilon^{-2} [\tr(V_0)  + \bfm_{0\bfw}^T\bfm_{0\bfw}] \\
& \quad + (a-1)(a-2) [ \tr(\tfrac{\Lambda}{\sigma_\eta^2} V_0) + \bfm_{0\bfw}^T \tfrac{\Lambda}{\sigma_\eta^2} \bfm_{0\bfw} ] \\
& \quad + 2 \mu^2 \bar{\bfw}^T \tfrac{\Lambda}{\sigma_\eta^2} \bar{\bfw} - a (\tfrac{a}{2} - 1)  \sigma_\epsilon^{-2}   \bfm_{0\bfw}^T(\bfy- \mu \bfw) \\
& \quad + n(a-1) - (a-2) (\tfrac{a}{2} - 2) \mu \bfm_{0\bfw}^T \tfrac{\Lambda}{\sigma_\eta^2} \bar{\bfw} \Big\}, \\
&= - \tfrac{1}{2\sigma_\eta^4} \Big\{ \tfrac{a^2}{2} \bfm_{0\bfw}^TV_0^{-1} \bfm_{0\bfw}  -2a \mu \bfm_{0\bone}^T \tfrac{\Lambda}{\sigma_\eta^2} \bar{\bfw}  \\ 
& \quad  + n(1-a)^2 + 2(a-1) [\tr(V_0)/\sigma_\epsilon^2  - {\bfm_{0\bone}}^T \tfrac{\Lambda}{\sigma_\eta^2}  \bfm_{0\bone} ] \Big\}. 
\end{aligned}
\end{equation}
Substituting \eqref{sigetaeq} in \eqref{Qsigmaeta}, we obtain 
\begin{equation*}
\begin{aligned}
I_{\sigma_\eta^2, \sigma_\eta^2} &= \tfrac{1}{2\sigma_\eta^4} \left\{ \tfrac{a^2}{2}\bfm_{0\bfw}^T V_0^{-1} \bfm_{0\bfw} - 2a \mu \bfm_{0\bone}^T \tfrac{\Lambda}{\sigma_\eta^2} \bar{\bfw} \right. \\
&\left. \quad +  n(a-1)^2  \right\}. 
\end{aligned}
\end{equation*}

\begin{proof}[Proof of Theorem 3]
The rate of convergence of Algorithm 2 can be optimized by minimizing $I_{\sigma_\eta^2, \sigma_\eta^2}$ with respect to $a$ and $\bfw$. 
\begin{equation*}
\begin{aligned}
\nabla_a I_{\sigma_\eta^2, \sigma_\eta^2} 
&= \tfrac{1}{2\sigma_\eta^4} \left\{ a\bfm_{0\bfw}^T V_0^{-1} \bfm_{0\bfw}  - 2\mu \bar{\bfw}^T \tfrac{\Lambda}{\sigma_\eta^2}  \bfm_{0\bone} \right. \\
&\left. + 2n(a-1) \right\} 
\end{aligned}
\end{equation*}
Setting $\nabla_a I_{\sigma_\eta^2, \sigma_\eta^2} = 0$, we obtain 
\begin{align} \label{aopt}
a &= \frac{2n + 2\mu \bar{\bfw}^T \frac{\Lambda}{\sigma_\eta^2}  \bfm_{0\bone}}{2n + \bfm_{0\bfw}^T V_0^{-1} \bfm_{0\bfw}} \\
&= 1 - \frac{ \bfm_{0\bfw}^T V_0^{-1} \bfm_{0\bfw} - 2\mu \bar{\bfw}^T \frac{\Lambda}{\sigma_\eta^2}  \bfm_{0\bone} }{2n + \bfm_{0\bfw}^T V_0^{-1} \bfm_{0\bfw}}. \nonumber
\end{align}
Note that $a \leq 1$ for any $\bfw$ as  
\begin{multline*}
\bfm_{0\bfw}^T V_0^{-1} \bfm_{0\bfw} - 2\mu \bar{\bfw}^T \tfrac{\Lambda}{\sigma_\eta^2}  \bfm_{0\bone} \\
= \bfm_{0\bfw}^T\bfm_{0\bfw}/\sigma_\epsilon^2 + \bfm_{0\bone} ^T \tfrac{\Lambda}{\sigma_\eta^2} \bfm_{0\bone} + \mu^2 \bar{\bfw}^T \tfrac{\Lambda}{\sigma_\eta^2} \bar{\bfw} \geq 0.
\end{multline*}
\begin{equation*} 
\begin{aligned}
\nabla_{\bar{\bfw}} I_{\sigma_\eta^2, \sigma_\eta^2}
&= \frac{a^2\mu}{2\sigma_\eta^4} V_0^{-1}  \left\{ \bfm_{0\bone} + \mu \bar{\bfw} - \frac{2V_0 \Lambda}{a\sigma_\eta^2}  \bfm_{0\bone} \right\}.  
\end{aligned}
\end{equation*}
Setting $\nabla_{\bar{\bfw}} I_{\sigma_\eta^2, \sigma_\eta^2} = 0$ leads to 
\begin{equation}\label{wbaropt}
\bar{\bfw} = \frac{1}{\mu} \left(\frac{2V_0 \Lambda}{a\sigma_\eta^2} - I \right) \bfm_{0\bone}.
\end{equation}
Note that \eqref{wbaropt} implies $\bfm_{0\bfw} = \frac{2V_0 \Lambda}{a\sigma_\eta^2} \bfm_{0\bone}$. Solving \eqref{aopt} and \eqref{wbaropt} simultaneously, 
\begin{multline*} 
2na + \tfrac{4}{a} \bfm_{0\bone}^T  \tfrac{\Lambda}{\sigma_\eta^2}V_0 \tfrac{\Lambda}{\sigma_\eta^2} \bfm_{0\bone} \\
= 2n + \tfrac{4}{a} \bfm_{0\bone}^T \tfrac{\Lambda}{\sigma_\eta^2}V_0 \tfrac{\Lambda}{\sigma_\eta^2}  \bfm_{0\bone} - 2 \bfm_{0\bone}^T \tfrac{\Lambda}{\sigma_\eta^2}  \bfm_{0\bone}.
\end{multline*}
Thus 
\[
 a= 1 - \frac{\bfm_{0\bone}^T \Lambda \bfm_{0\bone}}{n\sigma_\eta^2}
= \frac{\tr(V_0\Lambda) }{n \sigma_\eta^2} = 1 - \frac{\tr(V_0)}{n\sigma_\epsilon^2},
\]
where the last two equalities follows from \eqref{sigetaeq}. The second order derivatives are 
 \begin{equation*}
\begin{aligned}
\nabla_{a}^2 I_{\sigma_\eta^2, \sigma_\eta^2} &= ( 2n + \bfm_{0\bfw}^T V_0^{-1} \bfm_{0\bfw} )/(2\sigma_\eta^4),\\
\nabla_{\bar{\bfw}}^2 I_{\sigma_\eta^2, \sigma_\eta^2} &= a^2 \mu^2 V_0^{-1}/(2\sigma_\eta^4), \\
\nabla_{\bar{\bfw}, a}^2 I_{\sigma_\eta^2, \sigma_\eta^2}&= \mu (aV_0^{-1} \bfm_{0\bfw} - \tfrac{\Lambda}{\sigma_\eta^2} \bfm_{0\bf1})/\sigma_\eta^4.  
\end{aligned}
\end{equation*}
Let $g = 2 \frac{\Lambda}{\sigma_\eta^2} \bfm_{0\bf1}$. At $(a^\opt, \bar{\bfw}^\opt)$, the Hessian of $I_{\sigma_\eta^2, \sigma_\eta^2}$ is 
\begin{equation*}
\mathcal{H} = \frac{1}{2\sigma_\eta^4} \begin{bmatrix}
 2n + (a^\opt)^{-2} g^T V_0 g & \mu g^T  \\
\mu g & \mu^2  (a^\opt)^{2} V_0^{-1}
\end{bmatrix},
\end{equation*}
$\mathcal{H}$ is symmetric positive definite because 
\begin{equation*}
\begin{aligned}
x^T \mathcal{H} x = \frac{1}{2\sigma_\eta^4} \left\{ 2nx_1^2 + \Big( \frac{x_1 g}{a^\opt} + \mu a^\opt V_0^{-1}   x_2 \Big)^T V_0 \right.\\
\left. \times  \left( \frac{x_1 g}{a^\opt}  + \mu a^\opt V_0^{-1} x_2 \right) \right\} > 0
\end{aligned}
\end{equation*}
for any $x = (x_1, x_2)^T \neq 0$, where the dimensions of $x_1, x_2$ matches that of $a$ and $\bfw$. Hence $I_{\sigma_\eta^2, \sigma_\eta^2} $  is minimized at $(a^\opt, \bar{\bfw}^\opt)$. 
\end{proof}

\begin{proof}[Proof of Corollary 3]
$a^\opt =\tr(C^{-1})/n$ where $C = \gamma \Lambda^{-1} + I$. Hence $a^\opt$ is a function of $\phi$ and $\gamma$ only. The eigenvalues $\{\lambda_i\}$ of $\gamma \Lambda^{-1} $ are positive since $\Lambda^{-1}$ is a symmetric positive definite matrix. Thus the eigenvalues of $C$ are $\{\lambda_i + 1\}$ and $\tr(C^{-1}) = \sum_{i=1}^n 1/(\lambda_i +1) \in (0, n)$. It follows that $a^\opt \in (0,1)$. $\partial a^\opt/\partial \gamma = - \tr(C^{-1}\Lambda^{-1} C^{-1})/n < 0$ since $C^{-1}\Lambda^{-1} C^{-1}$ is positive definite.
\end{proof}

\begin{proof}[Proof of Corollary 4]
As $\bar{\bfw}^\opt = AS^{-1} (y/\mu - \bone)$ and $y \sim \N(\mu \bone, S)$, $\E(\bar{\bfw}^\opt) = \bfzero$ and $\cov(\bar{\bfw}^\opt) = AS^{-1}SS^{-1}A/\mu^2 = AS^{-1}A/\mu^2$.
\end{proof}

The proof of Theorem 4 requires the trace of $Q^{-1}$ which is stated in Property \ref{prop4} \citep[Theorem 3,][]{Tan2019}. The trace of $Q^{-2}$, which is required for deriving the rate of convergence of Algorithm 2 in Theorem 5 is also given in Property \ref{prop4}. 

\begin{property} \label{prop4}
\[
\tr(Q^{-1}) = \frac{n\kappa_0  (\vp^2\rpp^{n-1} + \vm^2 \rmm^{n-1}) + 2\gamma(\rpp^n - \rmm^n)}{\kappa_0^2 \kappa}
\]
and $\tr(Q^{-2}) = (\kappa_0^2 \kappa^2)^{-1} \mathcal{S}$, where
\begin{equation*}
\begin{gathered}
\mathcal{S}= 4n^2\gamma^2  + 8n\gamma  (\phi^2-1) + \frac{nc(\vp^4 \rpp^{2n-2} - \vm^4\rmm^{2n-2})} {\kappa_0} \\
- 4\gamma(1 + \phi^2) + 4\gamma \kappa_0^{-2} \{4 \gamma +  c( \vp^2\rpp^{2n-1} +\vm^2\rmm^{2n-1})\} \\
+ 2(\phi^2 -1)^2+ 2\kappa_0^{-1}(\phi^2-1)(\vp^2 \rpp^{2n-1} - \vm^2 \rmm^{2n-1}).
\end{gathered}
\end{equation*}
\end{property}

\begin{proof}[Proof of Theorem 4]
We have 
\[
a^\opt = 1 - \tr(V_0)/(n\sigma_\epsilon^2) = 1 - \gamma \tr(Q^{-1}) / (n|\phi|).
\]
From Property \ref{prop4}, 
\begin{multline*}
\lim_{n\rightarrow \infty} \frac{\tr(Q^{-1})}{n}  =\frac{1}{\kappa_0}\lim_{n\rightarrow \infty} \frac{\vp^2+ \vm^2 (\rmm/\rpp)^{n-1}  }{\vp^2 - \vm^2 (\rmm/\rpp)^{n-1} } \\
+ \frac{2\gamma}{\kappa_0^2} \lim_{n\rightarrow \infty} \frac{ 1 - (\rmm/\rpp)^n }{n\{\vp^2/\rpp -( \vm^2 /\rpp) (\rmm/\rpp)^{n-1}\}},
\end{multline*}
which reduces to $\kappa_0^{-1}$ since $0 <\rmm/\rpp < 1$. Moreover 
\[
\begin{aligned}
\kappa_0 &= \rpp-\rmm = \sqrt{c^2-4} \\
&=  \sqrt{\{(1-\phi)^2 + \gamma\}\{(1+\phi)^2 + \gamma\}}/|\phi|.
\end{aligned}
\]
Thus we obtain the large-sample limit of $a^\opt$.
\end{proof}

The large-sample estimate of $a^\opt$ is $\hat{a}^\opt = 1 - \gamma/\sqrt{f}$, where 
\begin{equation*}
\begin{aligned}
f &= \{(1-\phi)^2 + \gamma\}\{(1+\phi)^2 + \gamma\} \\
&= \gamma^2 + 2\gamma(1+\phi^2) + (1-\phi^2)^2. \\
\nabla_\gamma \hat{a}^\opt &= - [\gamma (1+\phi^2) + (1-\phi^2)^2] < 0. \\
\nabla_{|\phi|} \hat{a}^\opt &= 2f^{-3/2} \gamma |\phi| (\gamma+ \phi^2 - 1).
\end{aligned}
\end{equation*}

\begin{proof}[Proof of Theorem 5]
The optimal rate of convergence of Algorithm 2 is achieved when $a= a^\opt$ and $\bar{\bfw} = \bar{\bfw}^\opt$, where $\xi =  \bfm_{0\bone}^T \Lambda \bfm_{0\bone} / \sigma_\eta^2 = \tr(V_0)/\sigma_\epsilon^2$ from \eqref{sigetaeq}. Plugging $a= a^\opt$ and $\bar{\bfw} = \bar{\bfw}^\opt$ in $I_{\sigma_\eta^2, \sigma_\eta^2} $,
\begin{align*}
I_{\sigma_\eta^2, \sigma_\eta^2} 
&= \tfrac{1}{2\sigma_\eta^4} \big\{ 2a \bfm_{0\bone}^T \tfrac{\Lambda}{\sigma_\eta^2} \bfm_{0\bone} - 2 \bfm_{0\bone}^T \tfrac{\Lambda}{\sigma_\eta^2} V_0 \tfrac{\Lambda}{\sigma_\eta^2} \bfm_{0\bone} \\
& \quad + n(1-a)^2 \big\} \\
&= \tfrac{1}{2\sigma_\eta^4} \big\{ 2a \xi  - 2 \xi +  2 \bfm_{0\bone}^T S^{-1} \bfm_{0\bone} + \tfrac{\xi^2}{n} \big\} \\
&=  \tfrac{1}{2\sigma_\eta^4} \big\{   2 \bfm_{0\bone}^T S^{-1} \bfm_{0\bone}   -  \tfrac{\xi^2}{n} \big\}.
\end{align*}
The rate of convergence of Algorithm 2 at $(a^\opt, \bar{\bfw}^\opt)$ is 
\begin{equation*}
\begin{aligned}
1 - \frac{I_{\sigma_\eta^2, \sigma_\eta^2}^\obs}{I_{\sigma_\eta^2, \sigma_\eta^2}}
&= 1 - \frac{ 2{\bfm_{0\bf1}^T S^{-1}\bfm_{0\bf1}} - {\tr(V_0^2)} /\sigma_\epsilon^4}{2 \bfm_{0\bone}^T S^{-1} \bfm_{0\bone} -  \tr(V_0)^2 /(n\sigma_\epsilon^4)} \\
&= \frac{ \tr(V_0^2)/(n\sigma_\epsilon^4) -  \tr(V_0)^2 /(n^2\sigma_\epsilon^4)}{2 \bfm_{0\bone}^T S^{-1} \bfm_{0\bone}/n -  \tr(V_0)^2 /(n^2\sigma_\epsilon^4)}.
\end{aligned}
\end{equation*}
This rate is dependent on the observed data and varies even across data sets simulated from the same distribution. We approximate the convergence rate by replacing $\bfm_{0\bone}^T S^{-1} \bfm_{0\bone}/n$ with its mean, $\tr(V_0^2)/(n\sigma_\epsilon^4)$. Note that $\bfm_{0\bone}^T S^{-1} \bfm_{0\bone}/n$ converges in mean square to $\tr(V_0^2)/(n\sigma_\epsilon^4)$ since $\var(\bfm_{0\bf1}^T S^{-1}\bfm_{0\bf1}/n) = 2\tr(V_0^4)/(n^2\sigma_\epsilon^8) = \mathcal{O}(1/n)$. From Property \ref{prop4},
\begin{align*}
\lim_{n \rightarrow \infty} \frac{\tr(V_0)^2}{n^2\sigma_\epsilon^4} 
= \frac{\gamma^2}{\phi^2} \left[ \lim_{n \rightarrow \infty} \frac{\tr(Q^{-1})}{n} \right]^2 
= \frac{\gamma^2}{\phi^2 \kappa_0^2},
\end{align*}
since we have shown that $\lim_{n \rightarrow \infty} \tr(Q^{-1})/n = 1/\kappa_0$ when proving Theorem 4. In addition,
\begin{align*}
\lim_{n \rightarrow \infty} \frac{\tr(V_0^2)}{n\sigma_\epsilon^4} 
= \frac{\gamma^2}{\phi^2} \lim_{n \rightarrow \infty} \frac{\tr(Q^{-2})}{n}
= \frac{\gamma^2c}{\phi^2 \kappa_0^3}.
\end{align*}
This is because, from Property \ref{prop4}, all the terms in $\tr(Q^{-2})/n$ go to zero as $n \rightarrow \infty$ except for
\begin{multline*}
\frac{nc(\vp^4 \rpp^{2n-2} - \vm^4\rmm^{2n-2})} {\kappa_0^3 \kappa^2} =  \frac{c}{\kappa_0^3}\frac{(\vp^4 - \vm^4 (\tfrac{\rmm}{\rpp})^{2n-2})} {(\vp^2 - \vm^2 (\tfrac{\rmm}{\rpp})^{n-1})^2},
\end{multline*}
which goes to $c/\kappa_0^3$.
Thus the approximate convergence rate of Algorithm 2 when $n$ is large is 
\begin{equation*}
\begin{aligned}
J_{\text{Alg 2}}^\text{asymp} &= \lim_{n \rightarrow \infty} \frac{ \tr(V_0^2)/(n\sigma_\epsilon^4) -  \tr(V_0)^2 /(n^2\sigma_\epsilon^4)}{2 \tr(V_0^2)/(n\sigma_\epsilon^4)  -  \tr(V_0)^2 /(n^2\sigma_\epsilon^4)} \\
&= \frac{c - \kappa_0}{2c - \kappa_0} < \frac{1}{2}.
\end{aligned}
\end{equation*}
Recall that $\kappa_0 = \sqrt{c^2-4} < c$. Writing $J_{\text{Alg 2}}^\text{asymp}=4(c^2 + c\sqrt{c^2-4} + 4)^{-1}$,
\begin{equation*}
\begin{aligned}
\nabla_\gamma J_{\text{Alg 2}}^\text{asymp} &= -\frac{4(2c + \sqrt{c^2-4} + c^2/\sqrt{c^2-4})}{|\phi|(c^2 + c\sqrt{c^2-4} + 4)^2} < 0, \\
\nabla_{|\phi|} J_{\text{Alg 2}}^\text{asymp} &= -\nabla_\gamma J_{\text{Alg 2}}^\text{asymp} \frac{1 +\gamma -\phi^2}{ |\phi|} > 0.
\end{aligned}
\end{equation*}
Thus the approximate rate of convergence decreases as $\gamma$ increases and $|\phi|$ decreases.
\end{proof}

Next, we consider the rate of convergence of Algorithm 2 at some other combinations of $(a, \bfw)$, which are necessarily greater than the optimal rate above. For the CP ($a=0, \bfw=\bfzero$), 
\[
I_{\sigma_\eta^2, \sigma_\eta^2} = n/(2\sigma_\eta^4).
\]
Let $\delta = \bfm_{0\bf1}^T V_0^{-1} \bfm_{0\bf1}/(2n) > 0$ unless $\bfy = \mu \bone$. For the NCP ($a=0, \bfw=\bfzero$), 
\[
I_{\sigma_\eta^2, \sigma_\eta^2} = n\delta/(2\sigma_\eta^4).
\]
We can also consider fixing $\bfw = \bone$ and optimizing $I_{\sigma_\eta^2, \sigma_\eta^2}$ with respect to $a$ only. From \eqref{aopt}, this sub-optimal parametrization corresponds to
\[
a = 1/(1+\delta) \quad \text{and} \quad I_{\sigma_\eta^2, \sigma_\eta^2} = n\delta/[2\sigma_\eta^4(1+\delta)].
\]
The convergence rates follow a hierarchy,
\begin{equation*}
J_{\text{Alg 2}}^{\opt} \leq J_{\text{Alg 2}}^{\text{sub-opt}} < \min\{ J_{\text{Alg 2}}^{\text{CP}}, J_{\text{Alg 2}}^{\text{NCP}} \}.
\end{equation*}

\section{Augmented information matrix}
Let $\Lambda_2= \nabla_\phi^2 \Lambda = \diag([0,2 , \dots, 2, 0])$. The other elements in $\nabla_\bftheta^2 Q(\bftheta | \bftheta^{(i)})$ are given by 
\begin{align*}
\nabla_\phi^2 Q(\bftheta | \bftheta^{(i)}) &= - \frac{\tr(\Lambda_2 U)}{2\sigma_\eta^2} - \frac{1+\phi^2}{(1-\phi^2)^2}, \\
\nabla_{\sigma_\epsilon^2}^2 Q(\bftheta | \bftheta^{(i)}) &=  \frac{n}{2\sigma_\epsilon^4} - \frac{\sigma_\eta^{2a} \tr(V_a^{(i)}) +{z^{(i)}}^T z^{(i)}}{\sigma_\epsilon^6}, \\
\nabla_{\mu,\sigma_\eta^2}^2 Q(\bftheta | \bftheta^{(i)}) 
&= \tfrac{1}{2\sigma_\eta^2} \Big[2\mu\bar{\bfw}^T \tfrac{\Lambda}{\sigma_\eta^2}\bar{\bfw}  - \tfrac{a}{\sigma_\epsilon^2}  \sigma_\eta^a {\bfm_{a\bfw}^{(i)}}^T \bfw \\
\quad & + (a-2) \sigma_\eta^a {\bfm_{a\bfw}^{(i)}}^T \tfrac{\Lambda}{\sigma_\eta^2}\bar{\bfw} \Big],  \\
\nabla_{\mu, \phi}^2 Q(\bftheta | \bftheta^{(i)})
&= \bar{\bfw}^T \Lambda_1 (\sigma_\eta^a \bfm_{a\bfw}^{(i)} - \mu \bar{\bfw})/\sigma_\eta^2, \\
\nabla_{\mu, \sigma_\epsilon^2}^2 Q(\bftheta | \bftheta^{(i)}) 
&= - {z^{(i)}}^T\bfw/\sigma_\epsilon^4, \\
\nabla_{\sigma_\eta^2, \sigma_\epsilon^2}^2 Q(\bftheta | \bftheta^{(i)}) &= \frac{a [\sigma_\eta^{2a} \tr(V_a^{(i)}) - \sigma_\eta^a{z^{(i)}}^T \bfm_{a\bfw}^{(i)} ] }{2\sigma_\epsilon^4 \sigma_\eta^2}, \\
\nabla_{\phi, \sigma_\epsilon^2}^2 Q(\bftheta | \bftheta^{(i)}) &= 0,
\end{align*}
\begin{multline*}
\nabla_{\sigma_\eta^2, \phi}^2 Q(\bftheta | \bftheta^{(i)}) = \tfrac{1}{2\sigma_\eta^4}\{  (a-2) \mu \sigma_\eta^a {\bfm_{a\bfw}^{(i)}}^T \Lambda_1 \bar{\bfw} \\
+ \mu^2 \bar{\bfw}^T \Lambda_1 \bar{\bfw} + (1-a)\sigma_\eta^{2a} [\tr(\Lambda_1 V_a) + { \bfm_{a\bfw}^{(i)}}^T \Lambda_1  \bfm_{a\bfw}^{(i)} ]\}.
\end{multline*}

When $\bftheta = \bftheta^{(i)}$ at convergence, we can use the identities in Appendix \ref{Appendix A} to obtain
\begin{align*}
\nabla_\phi^2 Q(\bftheta | \bftheta) &= - \frac{\tr[ \Lambda_2 (V_0 + \bfm_{0\bf1} \bfm_{0\bf1}^T)]}{2\sigma_\eta^2} - \frac{1+\phi^2}{(1-\phi^2)^2}, \\
\nabla_{\sigma_\epsilon^2}^2 Q(\bftheta | \bftheta)  &= \frac{n}{2\sigma_\epsilon^4}- \frac{ \tr(V_0) + (\bfy - \bfm_{0\bfzero})^T (\bfy - \bfm_{0\bfzero})}{\sigma_\epsilon^6}, \\
\nabla_{\mu, \sigma_\eta^2}^2 Q(\bftheta | \bftheta) 
&= \tfrac{a}{2\sigma_\eta^2}\bfm_{0\bfw}^T \big( \tfrac{\Lambda}{\sigma_\eta^2} \bone - V_0^{-1} \bfw \big) - \bfm_{0\bone}^T \tfrac{\Lambda}{\sigma_\eta^4} \bar{\bfw}, \\
\nabla_{\mu, \phi}^2 Q(\bftheta | \bftheta) &= \bfm_{0\bone}^T \Lambda_1 \bar{\bfw}/ \sigma_\eta^2, \\
\nabla_{\mu, \sigma_\epsilon^2}^2 Q(\bftheta | \bftheta) &= - \bfm_{0\bone}^T\Lambda \bfw/ (\sigma_\eta^2 \sigma_\epsilon^2). \\
\nabla_{ \sigma_\eta^2,\phi}^2 Q(\bftheta | \bftheta) 
&=  \tfrac{1}{2\sigma_\eta^4} \{(1-a)[\tr(\Lambda_1 V_0) + \bfm_{0\bone}^T \Lambda_1 \bfm_{0\bone} ]\\
& \quad  - a \mu \bfm_{0\bone}^T \Lambda_1  \bar{\bfw}\}, \\
\nabla_{ \sigma_\eta^2, \sigma_\epsilon^2}^2 Q(\bftheta | \bftheta) &= \frac{a[\gamma \tr(V_0) - \bfm_{0\bone}^T \Lambda \bfm_{0\bone}- \mu \bfm_{0\bone}^T \Lambda \bar{\bfw} ]}{2\sigma_\epsilon^2 \sigma_\eta^4} ,  
\end{align*}

Next, we need to evaluate $\nabla^2_\bftheta Q(\bftheta|\bftheta)$ at $\bftheta= \bftheta^*$ (MLE). To simplify notation, we omit the asterisk below, but the expressions in the remainder of this section are all evaluated at $\bftheta = \bftheta^*$. Using the identities in \eqref{mueq}, \eqref{sigetaeq}, \eqref{phieq} and \eqref{sigepseq}, we can simplify $I_\augm(\bftheta^*)$ so that it is clearer whether there is any dependence on $a$ and $\bfw$. The elements in $I_\augm(\bftheta^*)$ are given by 
\begin{align*}
I_{\mu, \sigma_\eta^2} &= \frac{1}{\sigma_\eta^2} \left[ \frac{a \bfm_{0 \bfw}^T}{2} \left( V_0^{-1} \bfw - \frac{\Lambda \bone}{\sigma_\eta^2} \right) + \frac{\bfm_{0\bf1}^T \Lambda \bar{\bfw}}{\sigma_\eta^2} \right], \\
I_{\sigma_\eta^2, \phi} &= \frac{(1-a) \phi}{\sigma_\eta^2 (1-\phi^2)} + \frac{a \mu \bfm_{0\bone}^T \Lambda_1\bar{\bfw}}{2 \sigma_\eta^4}, \\
I_{\phi, \phi} &= \frac{1}{\sigma_\eta^2} \sum_{t=2}^{n-1} (V_0 + \bfm_{0\bone} \bfm_{0\bone}^T)_{tt} + \frac{1+\phi^2}{(1-\phi^2)^2}, \\
I_{\mu, \phi} &=- \frac{\bfm_{0\bf1}^T\Lambda_1 \bar{\bfw}}{\sigma_\eta^2},
\qquad
I_{ \mu, \sigma_\epsilon^2} =  -\frac{\bfm_{0\bone}^T\Lambda \bar{\bfw}}{\sigma_\epsilon^2\sigma_\eta^2},  \\
I_{\sigma_\eta^2, \sigma_\epsilon^2}  &= \frac{a\mu \bfm_{0\bf1}^T  \Lambda \bar{\bfw}}{2\sigma_\eta^4 \sigma_\epsilon^2}, 
\quad 
I_{\phi, \sigma_\epsilon^2} = 0, \quad I_{\sigma_\epsilon^2,\sigma_\epsilon^2} = \frac{n}{2\sigma_\epsilon^4}. 
\end{align*}

Note that $\bfm_{0\bf1}^T \frac{\Lambda}{\sigma_\eta^2} \bfw = - \bfm_{0\bf1}^T \frac{\Lambda}{\sigma_\eta^2} \bar{\bfw}$ as $\bfm_{0\bf1}^T \frac{\Lambda}{\sigma_\eta^2} \bone = 0$ from \eqref{mueq}. Next, we analyze the elements in $I_\augm(\bftheta^*)$ to identify those that will vary substantially with $a$ and $\bfw$ when $n$ is large. This is done by considering $\lim_{n \rightarrow \infty} I_{\theta_i, \theta_j}/n$. We will use the fact that since $\bfy \sim \N(\mu \bone, S)$, $\bfm_{0\bf1} \sim \N(\bfzero, V_0 S V_0/\sigma_\epsilon^4)$ and Theorem \ref{lem_quad} below.

\begin{lemma} \label{lem_normal_as}
Let $X_n\sim N(0, \sigma_n^2)$, where $\sigma_n^2 \leq M/n$ with $M>0$ being a constant. Then $X_n \overset{a.s.}{\longrightarrow} 0$.
\end{lemma}
\begin{proof}
If $x > 0$,
\[
\begin{aligned}
1 - \Phi(x) = \int_x^\infty \frac{ \e^{-z^2/2}}{\sqrt{2\pi}} dz 
\leq \int_x^\infty \frac{z\e^{-z^2/2}}{x\sqrt{2\pi}} dz 
= \frac{\e^{-x^2/2}}{x\sqrt{2\pi}}.
\end{aligned}
\]
\begin{multline*}
P(|X_n| > \epsilon) = 2[1-\Phi(\epsilon/\sigma_n)]  \\
\leq  2[1-\Phi(\epsilon\sqrt{n/M})]  \leq \sqrt{\frac{2M}{n\pi}} \frac{\e^{-n\epsilon^2/(2M)}}{\epsilon}.
\end{multline*}
Thus
\[
\begin{aligned}
\sum_{n=1}^\infty P(|X_n| > \epsilon) &\leq \frac{\sqrt{2M/\pi}}{\epsilon} \sum_{n=1}^\infty \frac{\e^{-n\epsilon^2/(2M)}}{\sqrt{n}}\\
& \leq \frac{\sqrt{2M/\pi}}{\epsilon} \sum_{n=1}^\infty \e^{-n\epsilon^2/(2M)} \\
&= \frac{\sqrt{2M/\pi}}{\epsilon(\e^{\epsilon^2/(2M)}-1)} < \infty.
\end{aligned}
\]
Hence $X_n \overset{a.s.}{\longrightarrow} 0$.
\end{proof}

\begin{theorem} \label{lem_quad}
Suppose $\bfu = (u_1, \dots, u_n)^T$ is a bounded constant vector, that is, there exists $C>0$ such that $|u_t| \leq C$ for all $t=1, \dots, n$ where $n > 1$. Then $\bfm_{0\bf1}^T \Lambda \bfu/ n$, $\bfm_{0\bf1}^T \Lambda_1 \bfu/ n$ and $\bfm_{0\bf1}^T S^{-1} \bfu/ n$ converge to zero almost surely.
\end{theorem}
\begin{proof}
As $\bfm_{0\bf1}^T \Lambda \bfu/ n$, $\bfm_{0\bf1}^T \Lambda_1 \bfu/ n$ and $\bfm_{0\bf1}^T S^{-1} \bfu/ n$ are all normally distributed with zero means, from Lemma \ref{lem_normal_as}, it suffices to show that $\var(\bfm_{0\bf1}^T \Lambda \bfu/ n)$, $\var(\bfm_{0\bf1}^T \Lambda_1 \bfu/ n)$ and $\var(\bfm_{0\bf1}^T S^{-1} \bfu/ n)$ are bounded by $M/n$ for some $M > 0$.
\begin{multline*}
\var\Big( \frac{\bfm_{0\bf1}^T\Lambda\bfu}{n} \Big) 
= \frac{\bfu^T \Lambda V_0 S V_0 \Lambda \bfu}{n^2\sigma_\epsilon^4} \\
= \frac{\sigma_\eta^4}{n^2} \bfu^T \bigg(\frac{I}{\sigma_\epsilon^2} -\frac{V_0}{\sigma_\epsilon^4} \bigg) \bfu
\leq \frac{\sigma_\eta^4}{n^2\sigma_\epsilon^2} \bfu^T \bfu
\leq  \frac{C^2 \sigma_\eta^4 }{n\sigma_\epsilon^2}.
\end{multline*}
\begin{multline*}
\var\Big( \frac{\bfm_{0\bf1}^T\Lambda_1 \bfu}{n} \Big) 
= \frac{ \bfu^T \Lambda_1 V_0 S V_0 \Lambda_1 \bfu }{n^2 \sigma_\epsilon^4} \\
= \tfrac{1}{n^2}\bfu^T \Lambda_1 (\sigma_\eta^2 \Lambda^{-1} - V_0) \Lambda_1 \bfu 
\leq \tfrac{\sigma_\eta^2}{n^2} \bfu^T \Lambda_1 \Lambda^{-1} \Lambda_1 \bfu \\
\leq \tfrac{\sigma_\eta^2}{n^2} \lambda_{\max}(\Lambda^{-1}) \| \Lambda_1\bfu \|^2 
\leq 16  \lambda_{\max}(\Lambda^{-1}) \sigma_\eta^2C^2/n,
\end{multline*}
where $\lambda_{\max}(\cdot)$ denotes the maximum eigenvalue. Note that
\begin{align*}
\| \Lambda_1 \bfu \|^2 
&= u_2^2 + u_{n-1}^2 + \sum_{t=2}^{n-1} (2\phi u_t - u_{t-1} - u_{t+1})^2 \\
&\leq 2C^2 + (n-2)(4C)^2 = (16n-30)C^2.
\end{align*}

Finally,
\begin{multline*}
\var\Big( \frac{\bfm_{0\bf1}^T S^{-1} \bfu}{n} \Big) 
= \frac{ \bfu^T  S^{-1} V_0 S V_0  S^{-1} \bfu }{n^2 \sigma_\epsilon^4} \\
= \frac{ \bfu^T  V_0^2  S^{-1} \bfu }{n^2 \sigma_\epsilon^4} 
= \frac{ \bfu^T  V_0^2}{n^2 \sigma_\epsilon^4} \Big( \frac{I}{\sigma_\epsilon^2} - \frac{V_0}{\sigma_\epsilon^4} \Big) \bfu \\
\leq \frac{ \bfu^T  V_0^2 \bfu }{n^2 \sigma_\epsilon^6} 
\leq \frac{\lambda_{\max}(V_0^2) \| \bfu \|^2}{n^2 \sigma_\epsilon^6} \leq \frac{ \lambda_{\max}(V_0^2) C^2}{n \sigma_\epsilon^6}.
\end{multline*}
\end{proof}
Assuming $\bfw$ (and hence $\bar{\bfw}$) is bounded, $\bfm_{0\bf1}^T \Lambda \bar{\bfw}/n$, $\bfm_{0\bf1}^T \Lambda_1 \bar{\bfw}/n$ and $\bfm_{0\bf1}^T S^{-1} \bar{\bfw}/n$ converge to zero almost surely. As $\Lambda = I$ if $\phi=0$, $\bfm_{0\bf1}^T \bar{\bfw}/n$ and $\bfm_{0\bf1}^T\bfw/ n$ also converge to zero almost surely. Thus we have 
\begin{align*}
\lim_{n \rightarrow \infty} \frac{I_{\mu, \phi}}{n} = 0, \quad
\lim_{n \rightarrow \infty} \frac{I_{ \mu, \sigma_\epsilon^2}}{n} = 0, \\
\lim_{n \rightarrow \infty} \frac{I_{\sigma_\eta^2, \phi}}{n} = 0, \quad
\lim_{n \rightarrow \infty} \frac{I_{\sigma_\eta^2, \sigma_\epsilon^2}}{n} = 0.
\end{align*}

\section{Kalman Filter and Smoother}
Consider the AR(1) plus noise model specified (with $a=0, \bfw=\bone$) as 
\begin{align}\label{SSM_01}
y_t &= \alpha_t + \mu + \sigma_\epsilon \epsilon_t, && t \geq 1,\\
\alpha_{t+1} &= \phi \alpha_{t} + \sigma_\eta \eta_t, \nonumber \\
\alpha_1 & \sim \N \left(0, \sigma_\eta^2/(1-\phi^2) \right). \nonumber
\end{align}
Define $\hat{\alpha}_{t|j} = \E(\alpha_t|y_{1:j})$ and $P_{t|j} = \var(\alpha_t|y_{1:j})$ for $t, j =1, \dots, n$, with $\hat{\alpha}_{1|0} = 0$ and $P_{1|0} = \sigma_\eta^2/(1-\phi^2)$ being the initial prior values. From \eqref{SSM_01}, the one-step ahead prediction step is given by
\begin{equation*}
\begin{aligned}
\hat{\alpha}_{t+1|t} = \phi \hat{\alpha}_{t|t}, \qquad  P_{t+1|t} = \phi^2 P_{t|t} + \sigma_\eta^2.
\end{aligned}
\end{equation*}
The joint distribution $p(\alpha_t, y_t|y_{1:t-1})$ is 
\begin{equation*}
\N \left( \begin{bmatrix} \hat{\alpha}_{t|t-1} \\  \hat{\alpha}_{t|t-1} + \mu \end{bmatrix}, 
\begin{bmatrix} P_{t|t-1} & P_{t|t-1} \\ P_{t|t-1} & F_t \end{bmatrix}  \right).
\end{equation*}
where $F_t = P_{t|t-1} + \sigma_\epsilon^2$. Hence, the updating step is 
\begin{equation} \label{update_step}
\begin{aligned}
\hat{\alpha}_{t|t} &=  \E(\alpha_t|y_{1:t-1}, y_t) \\
&= \hat{\alpha}_{t|t-1} + P_{t|t-1} F_t^{-1}v_t = \hat{\alpha}_{t|t-1} + K_tv_t, \\
P_{t|t} &=  \var(\alpha_t|y_{1:t-1}, y_t) \\
&= P_{t|t-1} - P_{t|t-1}^2 F_t^{-1} = P_{t|t-1} (1 - K_t),
\end{aligned}
\end{equation}
where $v_t = y_t - \mu- \hat{\alpha}_{t|t-1}$ and the Kalman gain $K_t = P_{t|t-1} F_t^{-1}$. The joint distribution $p(\alpha_t, \alpha_{t+1}|y_{1:t})$ is 
\begin{equation*}
\N \left( \begin{bmatrix} \hat{\alpha}_{t|t} \\ \hat{\alpha}_{t+1|t} \end{bmatrix}, 
\begin{bmatrix} P_{t|t} & \phi P_{t|t}\\
 \phi P_{t|t} &  P_{t+1|t}
\end{bmatrix}  \right),
\end{equation*}
and $\alpha_t$ is independent of $y_{t+1:n}$ given $\alpha_{t+1}$. Hence
\begin{equation*}
\begin{aligned}
\E(\alpha_t|\alpha_{t+1}, y_{1:n}) &= \E(\alpha_t|\alpha_{t+1}, y_{1:t}) \\
&= \hat{\alpha}_{t|t} + \phi P_{t|t} P_{t+1|t}^{-1} (\alpha_{t+1} - \hat{\alpha}_{t+1|t} ). \\
\var(\alpha_t|\alpha_{t+1}, y_{1:n}) &= \var(\alpha_t|\alpha_{t+1}, y_{1:t}) \\
&= P_{t|t} - \phi^2 P_{t|t}^2 P_{t+1|t}^{-1}. 
\end{aligned}
\end{equation*}
Applying the Tower property and law of total variance, 
\begin{align}\label{smoothing_eq}
\hat{\alpha}_{t|n} &= \E(\alpha_t|y_{1:n}) = \E[ \E(\alpha_t|\alpha_{t+1}, y_{1:n})|y_{1:n} ] \\
&= \hat{\alpha}_{t|t} + \phi P_{t|t} P_{t+1|t}^{-1} (\hat{\alpha}_{t+1|n} - \hat{\alpha}_{t+1|t} ), \nonumber \\
P_{t|n} &= \var(\alpha_t|y_{1:n}) =\E [\var(\alpha_t|\alpha_{t+1}, y_{1:n})|y_{1:n}] \nonumber\\
& \quad + \var[\E(\alpha_t|\alpha_{t+1}, y_{1:n}) |y_{1:n}] \nonumber\\
&= P_{t|t} - \phi^2 P_{t|t}^2 P_{t+1|t}^{-2} (P_{t+1|t} - P_{t+1|n}). \nonumber
\end{align}
From \eqref{update_step} and \eqref{smoothing_eq}, $\hat{\alpha}_{t|n} = \hat{\alpha}_{t|t-1} + P_{t|t-1} r_{t-1}$, where
\begin{equation*}
\begin{aligned}
r_{t-1} &=  F_t^{-1}v_t + \phi (1 - K_t) P_{t+1|t}^{-1} (\hat{\alpha}_{t+1|n} - \hat{\alpha}_{t+1|t} )\\
&=  F_t^{-1}v_t +L_t r_t, 
\end{aligned}
\end{equation*}
where $L_t = \phi (1 - K_t)$. Similarly, $P_{t|n} =  P_{t|t-1} - P_{t|t-1}^2N_{t-1}$, where
\begin{equation*}
\begin{aligned}
N_{t-1} &= F_t^{-1}  + L_t^2 P_{t+1|t}^{-2} (P_{t+1|t} - P_{t+1|n}) = F_t^{-1}  + L_t^2 N_t. 
\end{aligned}
\end{equation*}
Hence we obtain the backwards recurrence relations for $\{r_t\}$ and $\{N_t\}$ with initial values $r_n = 0$ and $N_n = 0$. 
In addition, $P_{t,t+1|n} = \E(\alpha_t \alpha_{t+1}|y_{1:n}) - \hat{\alpha}_{t|n}\hat{\alpha}_{t+1|n}$, where
\begin{multline*}
\E(\alpha_t \alpha_{t+1}|y_{1:n}) =\E [\E(\alpha_t \alpha_{t+1}|\alpha_{t+1}, y_{1:n})|y_{1:n} ] \\
=\E [\{\hat{\alpha}_{t|t} + \phi P_{t|t} P_{t+1|t}^{-1} (\alpha_{t+1} - \hat{\alpha}_{t+1|t} )\}\alpha_{t+1}|y_{1:n} ] \\
= \hat{\alpha}_{t|t} \hat{\alpha}_{t+1|n} + L_t P_{t|t-1} P_{t+1|t}^{-1} \{ P_{t+1|n} \\
+ (\hat{\alpha}_{t+1|n} - \hat{\alpha}_{t+1|t})\hat{\alpha}_{t+1|n} \}.
\end{multline*}
Thus $P_{t,t+1|n}$ is given by 
\begin{equation*}
\begin{aligned}
& (\hat{\alpha}_{t|t} - \hat{\alpha}_{t|n}) \hat{\alpha}_{t+1|n} + L_t P_{t|t-1} (P_{t+1|t}^{-1} P_{t+1|n} + r_t\hat{\alpha}_{t+1|n})  \\
& = L_t P_{t|t-1} (P_{t+1|t}^{-1} P_{t+1|n} +  r_t \hat{\alpha}_{t+1|n}) - P_{t|t-1} L_t r_t \hat{\alpha}_{t+1|n}   \\
& = P_{t|t-1} L_t (1 - P_{t+1|t} N_t).
\end{aligned}
\end{equation*}
Algorithm \ref{KF_Alg} summarizes the steps for implementing the Kalman filter and smoother. Note that $(P_{1|n}, \dots, P_{n|n})^T$ and $(P_{1, 2|n}, \dots, P_{n-1, n|n})^T$ are respectively the diagonal and first diagonal of $V_0$ and $(\hat{\alpha}_{1|n}, \dots, \hat{\alpha}_{n|n})^T = \bfm_{0\bf1} = V_0(\bfy - \mu \bone)/\sigma_\epsilon^2$. We can also use the Kalman filter to compute $\bar{\bfw}^\opt = V_0 \bone/\sigma_\epsilon^2$ and $a^\opt = 1 - \tr(V_0)/\sigma_\epsilon^2$ by setting $\bfy= \bfzero$, $\mu=-1$, $\sigma_\epsilon^2=1$ and $\sigma_\eta^2 = \gamma$ in the above algorithm. If $\sigma_\epsilon$ is not constant across $t$, we only need to replace $\sigma_\epsilon$ by $\sigma_{\epsilon_t}$ in the Kalman filter.

\begin{algorithm}
\caption{Kalman filter and smoother} \label{KF_Alg}
\begin{normalsize}
\begin{flushleft}
Set $\hat{\alpha}_{1|0} = 0$ and $P_{1|0} = \sigma_\eta^2/(1-\phi^2)$. For $t=1, \dots, n,$
\begin{enumerate}
\item Compute $v_t = y_t - \mu- \hat{\alpha}_{t|t-1}$, $F_t = P_{t|t-1} + \sigma_\epsilon^2$ and $K_t = P_{t|t-1}F_t^{-1}$.
\item Update: $\hat{\alpha}_{t|t} = \hat{\alpha}_{t|t-1} + K_t v_t$, $P_{t|t} = P_{t|t-1} ( 1 - K_t)$.
\item Predict: $\hat{\alpha}_{t+1|t} = \phi \hat{\alpha}_{t|t}$, $P_{t+1|t} = \phi^2 P_{t|t} + \sigma_\eta^2$. 
\end{enumerate}
Store $\{v_t\}$, $\{F_t\}$, $\{K_t\}$, $\{\hat{\alpha}_{t+1|t}\}$ and $\{P_{t+1|t}\}$.\\
Initialize $r_n = 0$ and $N_n = 0$. For $t = n, \dots, 1$,
\begin{enumerate}
\item Compute $L_t = \phi (1 - K_t)$.
\item If $t < n$, compute $P_{t,t+1|n} = P_{t|t-1} L_t (1 - P_{t+1|t} N_t)$. 
\item Recurrence relation: $r_{t-1} =  F_t^{-1}v_t +L_t r_t$ and $N_{t-1} = F_t^{-1}  + L_t^2 N_t$. 
\item Compute $\hat{\alpha}_{t|n} = \hat{\alpha}_{t|t-1} + P_{t|t-1} r_{t-1}$ and $P_{t|n} =  P_{t|t-1} - P_{t|t-1}^2N_{t-1}$.
\end{enumerate}
\end{flushleft}
\end{normalsize}
\end{algorithm}

\section{EM algorithm parameter estimates for simulated data}
For the simulations in Section 3.1 where all parameters are unknown, Table \ref{Alg3avgest} shows the averaged estimates obtained at convergence. Any differences have been highlighted in bold. The greatest differences arise when $|\phi|=0.3$ and $\gamma=0.1$. For these settings, NCP is preferred to CP as it converges faster and achieves a higher log-likelihood on average. PNCP leans towards NCP automatically in these case and its estimates are closer to that of NCP than CP. However, PNCP is able to achieve a higher average log-likelihood than NCP, suggesting that its estimates are closer on average to the true maximum likelihood estimates.
\begin{table*}[htb!]
\caption{(All parameters unknown) Averaged estimates of $\{\mu, \sigma_\eta^2, \phi, \sigma_\epsilon^2\}$.}
\label{Alg3avgest}
\begin{tabular}{@{}ccccc@{}}
\hline $\phi$ & $\gamma$ & NCP & CP & PNCP \\   \hline
\multirow{3}{*}{$-$0.9} & 0.1 &($-$1.000, 0.010, $-$0.897, 0.100) & ($-$1.000, 0.010, $-$0.897, 0.100) & ($-$1.000, 0.010, $-$0.897, 0.100) \\
& 1 & ($-$1.000, 0.101, $-$0.899, 0.100) & ($-$1.000, 0.101, $-$0.899, 0.100) & ($-$1.000, 0.101, $-$0.899, 0.100) \\
& 10 & ($-$0.999, 1.005, $-$0.900, 0.099) & ($-$0.999, 1.005, $-$0.900, 0.099) & ($-$0.999, 1.005, $-$0.900, 0.099) \\ \hline
\multirow{3}{*}{$-$0.6} & 0.1 &  ($-$1.000, 0.011, $-$0.581, 0.099) & ($-$1.000, 0.011, $-$0.581, 0.099) & ($-$1.000, 0.011, $-$0.581, 0.099) \\
& 1 & ($-$1.000, 0.102, $-$0.599, 0.099) & ($-$1.000, 0.102, $-$0.599, 0.099) & ($-$1.000, 0.102, $-$0.599, 0.099) \\
& 10 & ($-$0.999, 1.002, $-$0.601, 0.100) & ($-$0.999, 1.002, $-$0.601, 0.100) & ($-$0.999, 1.002, $-$0.601, 0.100) \\ \hline
\multirow{3}{*}{$-$0.3} & 0.1 &  ($-$1.000, {\bf 0.025, $-$0.285, 0.085}) & ($-$1.000, {\bf 0.023, $-$0.291, 0.087}) & ($-$1.000, {\bf 0.027, $-$0.283, 0.083}) \\
& 1 & ($-$1.000, 0.109, $-$0.299, 0.092) & ($-$1.000, 0.109, $-$0.299, 0.092) & ($-$1.000, 0.109, $-$0.299, 0.092) \\
& 10 & ($-$0.999, 0.969, $-$0.311, 0.129) & ($-$0.999, 0.969, $-$0.311, 0.129) & ($-$0.999, 0.969, $-$0.311, 0.129) \\ \hline
\multirow{3}{*}{0.3} & 0.1 &  ($-$1.000, {\bf 0.026, 0.276, 0.085}) & ($-$1.000, {\bf 0.024, 0.282, 0.087}) & ($-$1.000, {\bf 0.027, 0.274, 0.083}) \\
& 1 & ($-$0.999, 0.109, {\bf 0.296}, 0.091) & ($-$0.999, 0.109, {\bf 0.296}, 0.091) & ($-$0.999, 0.109, {\bf 0.297}, 0.091) \\
& 10 & ($-$0.997, 0.956, 0.313, {\bf 0.142}) & ($-$0.997, 0.956, 0.313, {\bf 0.141}) & ($-$0.997, 0.956, 0.313, {\bf 0.141}) \\ \hline
\multirow{3}{*}{0.6} & 0.1 &  ($-$1.000, 0.011, {\bf 0.578}, 0.100) & ($-$1.000, 0.011, {\bf 0.578}, 0.100) & ($-$1.000, 0.011, {\bf 0.579}, 0.100) \\
& 1 & ($-$0.999, 0.100, 0.601, 0.101) & ($-$0.999, 0.100, 0.601, 0.101) & ($-$0.999, 0.100, 0.601, 0.101) \\
& 10 & ($-$0.995, 0.992, 0.603, 0.107) & ($-$0.995, 0.992, 0.603, 0.107) & ($-$0.995, 0.992, 0.603, 0.107) \\ \hline
\multirow{3}{*}{0.9} & 0.1 &  ($-$0.998, 0.010, 0.901, 0.100) & ($-$0.998, 0.010, 0.901, 0.100) & ($-$0.998, 0.010, 0.901, 0.100) \\
& 1 & ($-$0.993, 0.100, 0.901, 0.101) & ($-$0.993, 0.100, 0.901, 0.101) & ($-$0.993, 0.100, 0.901, 0.101) \\
& 10 & {(\bf $-$0.977}, 0.996, 0.901, 0.103) & ({\bf $-$0.976}, 0.996, 0.901, 0.103) & ({\bf$-$0.976}, 0.996, 0.901, 0.103) \\    \hline
\end{tabular}
\end{table*}

\section{EM Algorithm for finding posterior mode (unknown scale parameter)}
Suppose a prior $p(\sigma_\eta^2)$ has been specified on $\sigma_\eta^2$. Then at convergence, 
\begin{multline*}
\nabla_{\sigma_\eta^2} Q(\bftheta | \bftheta) = \tfrac{1}{2\sigma_\eta^2} \big\{  n(a-1) - a \tr(\tfrac{V_0}{\sigma_\epsilon^2}) +(1-a)\tr(V_0 \tfrac{ \Lambda }{\sigma_\eta^2}) \\
 \quad + \bfm_{0\bone}^T \tfrac{ \Lambda }{\sigma_\eta^2} \bfm_{0\bone}  \big\}+ \nabla_{\sigma_\eta^2} \log p(\sigma_\eta^2)= 0. 
\end{multline*}
If we set $a=1$ and $\bfw=\bone$, we obtain 
\begin{equation} \label{posterior_mode}
\begin{aligned}
 \bfm_{0\bone}  \tfrac{\Lambda}{\sigma_\eta^2} \bfm_{0\bone} - \tr(\tfrac{V_0}{\sigma_\epsilon^2}) + 2\sigma_\eta^2\nabla_{\sigma_\eta^2} \log p(\sigma_\eta^2)= 0. \\
\end{aligned}
\end{equation}
Thus, the posterior mode $\sigma_\eta^2$ satisfies \eqref{posterior_mode}. From \eqref{Qsigmaeta}, 
\begin{multline*}
\nabla_{\sigma_\eta^2}^2 Q(\bftheta | \bftheta) 
= - \tfrac{1}{2\sigma_\eta^4} \big\{ n(1-a)^2 + \tfrac{a^2}{2} \bfm_{0\bfw}^TV_0^{-1} \bfm_{0\bfw}   \\ 
+ 2(a-1) [\tr(V_0)/\sigma_\epsilon^2  - {\bfm_{0\bone}}^T \tfrac{\Lambda}{\sigma_\eta^2}  \bfm_{0\bone} ] \\
-2a \mu \bfm_{0\bone}^T \tfrac{\Lambda}{\sigma_\eta^2} \bar{\bfw}    \big\} + \nabla^2_{\sigma_\eta^2} \log p(\sigma_\eta^2).
\end{multline*}
Evaluating $\nabla_{\sigma_\eta^2}^2 Q(\bftheta | \bftheta) $ at the posterior mode,
\begin{multline*}
I_{\sigma_\eta^2, \sigma_\eta^2}^B 
= \tfrac{1}{2\sigma_\eta^4} \big\{\tfrac{a^2}{2} \bfm_{0\bfw}^TV_0^{-1} \bfm_{0\bfw} -2a \mu \bfm_{0\bone}^T \tfrac{\Lambda}{\sigma_\eta^2} \bar{\bfw}  \\ 
+ n(1-a)^2 + 4(a-1) \sigma_\eta^2\nabla_{\sigma_\eta^2} \log p(\sigma_\eta^2) \big\} -  \nabla^2_{\sigma_\eta^2} \log p(\sigma_\eta^2)  .
\end{multline*}
Differentiating with respect to $a$ and $\bar{\bfw}$,
\begin{equation*}
\begin{aligned}
\nabla_a I_{\sigma_\eta^2, \sigma_\eta^2}^B &=\tfrac{1}{2\sigma_\eta^4} \{  a \bfm_{0\bfw}^TV_0^{-1} \bfm_{0\bfw}  -2\mu \bfm_{0\bone}^T \tfrac{\Lambda}{\sigma_\eta^2} \bar{\bfw} \\
& \quad + 2n(a-1) + 4\sigma_\eta^2\nabla_{\sigma_\eta^2} \log p(\sigma_\eta^2)  \}, \\
\nabla_{\bar{\bfw}} I_{\sigma_\eta^2, \sigma_\eta^2}^B &= \tfrac{a\mu}{2\sigma_\eta^4}(aV_0^{-1} \bfm_{0\bfw}  -2 \tfrac{\Lambda}{\sigma_\eta^2}\bfm_{0\bone}). \\
\end{aligned}
\end{equation*}
Solving $\nabla_a I_{\sigma_\eta^2, \sigma_\eta^2}^B$ and $\nabla_{\bar{\bfw}} I_{\sigma_\eta^2, \sigma_\eta^2}^B = \bf0$ simultaneously, 
\begin{align*}
a^\opt &  = 1 - \tr(V_0)/ (n\sigma_\epsilon^2), \\
\bar{\bfw}^\opt &= \frac{1}{\mu} \left(\frac{2V_0 \Lambda}{a^\opt \sigma_\eta^2} - I  \right) \bfm_{0\bone}.
\end{align*}
Thus we obtain the same expressions of $(a^\opt, \bar{\bfw}^\opt)$ as in Section 2.2.

\section{Gibbs sampler and variational Bayes}
Suppose $\mu$ is only unknown parameter and $p(\mu) \propto 1$.
\begin{multline*}
p(\mu, \bfalpha | \bfy) \propto p(\bfy | \bfalpha, \mu) p(\bfalpha | \mu) \\
\propto \exp \{ - (\bfy - \sigma_\eta^a \bfalpha - \mu \bfw)^T(\bfy - \sigma_\eta^a \bfalpha - \mu \bfw)/(2\sigma_\epsilon^2) \\
- (\sigma_\eta^a \bfalpha - \mu \bar{\bfw})^T \Lambda(\sigma_\eta^a \bfalpha - \mu \bar{\bfw}) / (2\sigma_\eta^2)\ \} \\
\propto \exp \big\{ -\tfrac{1}{2} [\sigma_\eta^{2a} \bfalpha^TV_0^{-1}\bfalpha -2 \sigma_\eta^a \bfalpha^T \bfy/\sigma_\epsilon^2 + \mu^2 \tau(\bfw)\\
- 2 \mu \bfw^T \bfy/\sigma_\epsilon^2+ 2\mu \sigma_\eta^{a} \bfalpha^T \rho(\bfw)] \big\}
\end{multline*}
where $\tau(\bfw) = \bfw^T \bfw / \sigma_\epsilon^2 + \bar{\bfw}^T \Lambda \bar{\bfw}/ \sigma_\eta^2$ and $\rho(\bfw) = V_0^{-1} \bfw - \sigma_\eta^{-2} \Lambda \bone$. Hence $p(\mu, \bfalpha | \bfy)$ is Gaussian with precision matrix given by 
\begin{equation*}
\begin{aligned}
\begin{bmatrix} \tau(\bfw)  & \sigma_\eta^a \rho(\bfw)^T\\  \sigma_\eta^a \rho(\bfw) & \sigma_\eta^{2a}V_0^{-1} \end{bmatrix}.
\end{aligned}
\end{equation*}
Since
\begin{equation*}
p(\mu | \bfalpha, \bfy) \propto \exp ( -\tfrac{1}{2} \{\mu^2 \tau(\bfw) - 2 \mu [ \tfrac{\bfw^T \bfy}{\sigma_\epsilon^2}- \sigma_\eta^{a} \bfalpha^T \rho(\bfw)] \}),
\end{equation*}
\begin{equation*}
\mu | \bfalpha, \bfy \sim \N( \tau(\bfw)^{-1} \{\tfrac{\bfw^T \bfy}{\sigma_\epsilon^2} - \sigma_\eta^a \bfalpha^T \rho(\bfw) \},  \tau(\bfw)^{-1} ).
\end{equation*}
Since 
\begin{equation*}
\begin{aligned}
p(\bfalpha | \mu, \bfy) &\propto \exp (-\tfrac{1}{2}\{ \sigma_\eta^{2a} \bfalpha^TV_0^{-1}\bfalpha -2\bfalpha^T \sigma_\eta^a [  \bfy/\sigma_\epsilon^2 \\
& \quad -\mu  \rho(\bfw)]\}), 
\end{aligned}
\end{equation*}
\begin{equation*}
\bfalpha | \mu, \bfy \sim \N( \sigma_\eta^{-a} V_0 \{  \bfy/\sigma_\epsilon^2-\mu  \rho(\bfw)\} ,  \sigma_\eta^{-2a} V_0 ).
\end{equation*}
Let $CC^T = \sigma_\eta^{-2a} V_0$ and  $Z_s \overset{iid}{\sim} \N(0,1)$ for $s=1, 2,3$. We can write 
\begin{align*}
\bfalpha^{(i)} &= \sigma_\eta^{-a} V_0 \{  \bfy/\sigma_\epsilon^2-\mu^{(i-1)}  \rho(\bfw)\} + CZ_1, \\
\mu^{(i)} &= \frac{\sigma_\epsilon^{-2} \bfy^T \bfw - \sigma_\eta^a {\bfalpha^{(i)}}^T \rho(\bfw) }{\tau(\bfw)} +\frac{Z_2}{\sqrt{ \tau(\bfw)}}  \\
& = \tau(\bfw)^{-1} \{-  \rho(\bfw)^T V_0 \{  \bfy/\sigma_\epsilon^2 - \mu^{(i-1)}  \rho(\bfw)\}   \\
& \quad +  \bfy^T \bfw/\sigma_\epsilon^2 + \sqrt{ \tau(\bfw)}Z_2 +  \sigma_\eta^a \rho(\bfw) ^T CZ_1 \}\\
& =  \tau(\bfw)^{-1} \{ \bfy^T S^{-1} \bone + \rho(\bfw)^T V_0 \rho(\bfw) \mu^{(i-1)}\\
& \quad + \sqrt{\tau(\bfw) + \rho(\bfw)^T V_0 \rho(\bfw) }Z_3\}.
\end{align*}

For variational Bayes,
\begin{align*}
&q(\mu) \propto \exp [\E_{q(\bfalpha)}\{ \log p(\mu, \bfalpha, \bfy)\}] \\
&\propto  \exp \left[  -\tfrac{1}{2} \{\mu^2 \tau(\bfw) - 2 \mu \{ \bfw^T \bfy/\sigma_\epsilon^2 - \sigma_\eta^{a}  \rho(\bfw)^T \bfm_\bfalpha^q\}\}  \right]. \\
&q(\bfalpha) \propto \exp [\E_{q(\mu)}\{ \log p(\mu, \bfalpha, \bfy)\}] \\
&\propto \exp \left[ -\tfrac{1}{2} \{ \sigma_\eta^{2a} \bfalpha^TV_0^{-1}\bfalpha -2 \sigma_\eta^a \bfalpha^T \{\bfy/\sigma_\epsilon^2 - m_\mu^q \rho(\bfw)\} \}  \right].
\end{align*}
Hence
\begin{align*}
q(\mu) &= \N( \tau(\bfw)^{-1}\{ \bfw^T \bfy/\sigma_\epsilon^2 - \sigma_\eta^{a}  \rho(\bfw)^T \bfm_\bfalpha^q\}, \tau(\bfw)^{-1}), \\
q(\bfalpha) &= \N( \sigma_\eta^{-a}V_0  \{\bfy/\sigma_\epsilon^2 - m_\mu^q \rho(\bfw)\}, \sigma_\eta^{-2a}V_0).
\end{align*}
In the updating, 
\[
{\bfm_\bfalpha^q}^{(i)} = \sigma_\eta^{-a}V_0  \{\bfy/\sigma_\epsilon^2 - {m_\mu^q}^{(i-1)} \rho(\bfw)\} 
\]
\begin{equation*}
\begin{aligned}
{m_\mu^q}^{(i)} &= \{ \bfw^T \bfy/\sigma_\epsilon^2 - \sigma_\eta^{a}  \rho(\bfw)^T {\bfm_\bfalpha^q}^{(i)} \}/\tau(\bfw) \\
&= \{ \bfy^T S^{-1} \bone + {m_\mu^q}^{(i-1)} \rho(\bfw)^T V_0\rho(\bfw) \}/\tau(\bfw).
\end{aligned}
\end{equation*}

\section{Parameters of mixture of normals}
Table \ref{Mixture} shows the parameters of the mixture approximation for the SV and SCD models in Section 5.
\begin{table}[htb!]
\caption{Parameters of mixture of normals for SV and SCD model.}
\label{Mixture}
\begin{tabular}{@{}cccrc@{}} \hline
& $k$ & $p_k$ & $m_k$ & $s_k^2$ \\ \hline
\multirow{10}{*}{SV} & 1 & 0.00609 & 1.92677 & 0.11265 \\
& 2 & 0.04775 & 1.34744 & 0.17788 \\
& 3 & 0.13057 & 0.73504 & 0.26768 \\
& 4 & 0.20674 & 0.02266 & 0.40601 \\
& 5 & 0.22715 & $-0.85173$ & 0.62699 \\
& 6 & 0.18842 & $-1.97278$ & 0.98583 \\
& 7 & 0.12047 & $-3.46788$ & 1.57469 \\
& 8 & 0.05591 & $-5.55246$ & 2.54498 \\
& 9 & 0.01575 & $-8.68384$ & 4.16591 \\
& 10 & 0.00115 & $-14.65000$ & 7.33342 \\
\hline
\multirow{10}{*}{SCD} & 1 & 0.00397 & $-5.09000$ & 4.50000 \\
&  2 & 0.03960 & $-3.29000$ & 2.02000 \\
&  3 & 0.16800 & $-1.82000$ & 1.10000 \\
&  4 & 0.14700 & $-1.24000$ & 0.42200 \\
&  5 & 0.12500 & $-0.76400$ & 0.19800 \\
&  6 & 0.10100 & $-0.39100$ & 0.10700 \\
&  7 & 0.10400 & $-0.04310$ & 0.07780 \\
&  8 & 0.11600 & $0.30600$ & 0.07660 \\
&  9 & 0.10700 & $0.67300$ & 0.09470 \\
&  10 & 0.08800 & $1.06000$ & 0.14600 \\
\hline
\end{tabular}
\end{table}

\section{MCMC Algorithms}
The conditional posterior distribution for the latent states $\bfalpha$ is
\begin{align*}
p(\bfalpha|\bfy, \bftheta, \bfr) &\propto \exp \big\{ -\tfrac{1}{2}[ (\tilde{\bfy}  - \bfm_\bfr  - \mu \bfw - \sigma_\eta^a \bfalpha)^T D_{\bfr}^{-1} \\
& \quad \times (\tilde{\bfy}  - \bfm_\bfr - \mu \bfw  - \sigma_\eta^a \bfalpha)  \\
& \quad + \sigma_\eta^{-2}(\sigma_\eta^a\bfalpha - \mu \bar{\bfw})^T \Lambda  (\sigma_\eta^a \bfalpha -  \mu \bar{\bfw}) ] \big\} \\
&\propto \exp \big\{ -\tfrac{1}{2}(\bfalpha^T C_\bfalpha \bfalpha -2 \bfalpha^T \bfc_{\bfalpha} ) \big\}.
\end{align*}
The conditional posterior distribution of the parameter $\mu$ is
\begin{multline*}
p(\mu|\bfy, \bfalpha,  \bfr,  \sigma_\eta^2, \phi) \propto \exp \{ -\tfrac{1}{2} [ (\tilde{\bfy}  - \bfm_\bfr  - \mu \bfw - \sigma_\eta^a \bfalpha)^T  \\
\times D_{\bfr}^{-1} (\tilde{\bfy}  - \bfm_\bfr - \mu \bfw  - \sigma_\eta^a \bfalpha)   + (\mu - b_\mu)^2/B_\mu \\
+ \sigma_\eta^{-2}(\sigma_\eta^a\bfalpha - \mu \bar{\bfw})^T \Lambda (\sigma_\eta^a \bfalpha -  \mu \bar{\bfw}) ]  \} \\
\propto \exp \big\{ -\tfrac{1}{2}( C_\mu \mu^2 -2 c_\mu \mu ) \big\}.
\end{multline*}
The conditional posterior distribution of $\phi$ is
\begin{multline*}
p(\phi|\bfy,  \bfalpha, \bfr, \mu, \sigma_\eta^2)  \propto \exp[ \{\log(1-\phi^2) \\
-  \sigma_\eta^{-2}(\sigma_\eta^a\bfalpha - \mu \bar{\bfw})^T \Lambda (\sigma_\eta^a \bfalpha - \mu \bar{\bfw})\}/2 \\
\quad  \qquad  + (b_\phi-1) \log(1+\phi) + (B_\phi-1) \log(1-\phi) ].
\end{multline*}
Lastly, the conditional posterior distribution of $\sigma_\eta^2$ is 
\begin{multline*}
p(\sigma_\eta^2 | \bfy, \bfalpha, \bfr, \mu, \phi) 
\propto \exp \{ -\sigma_\eta^{2a} \bfalpha^T ( D_{\bfr}^{-1} +  \sigma_\eta^{-2}\Lambda) \bfalpha/2 \\
+ \sigma_\eta^a \bfalpha^T D_{\bfr}^{-1} (\tilde{\bfy} - \bfm_\bfr - \mu \bfw) - \mu^2 \bar{\bfw}^T \Lambda \bar{\bfw}/(2\sigma_\eta^2)\\
+ \mu \Lambda \bar{\bfw} /\sigma_\eta^2 - [n(1-a)+1] (\log \sigma_\eta^2)/2 - \sigma_\eta^2/(2B_\sigma) \}.
\end{multline*}
If $\nu = \log \sigma_\eta^2$, then $p(\nu | \bfy, \bfalpha, \bfr, \mu, \phi) \propto \exp \{f_\nu(\nu)\}$, where
\[
\begin{aligned}
f_\nu(\nu) &= A_1 \e^{a\nu} + A_2 \e^{(a-1)\nu} + A_3 \e^{a\nu/2} + A_4 \e^{(a/2-1)\nu} \\
& \quad + A_5 \e^{-\nu} + A_6 \e^\nu + A_7 \nu,
\end{aligned}
\]
$A_1 = - \bfalpha^TD_{\bfr}^{-1}\bfalpha/2$, $A_2 = - \bfalpha^T \Lambda\bfalpha/2$, $A_3 =  \bfalpha^T D_{\bfr}^{-1} (\tilde{\bfy} - \bfm_\bfr - \mu \bfw)$, $A_4 =  \mu \bfalpha^T \Lambda \bar{\bfw}$, $A_5 = -  \mu^2 \bar{\bfw}^T \Lambda \bar{\bfw}/2$, $A_6 = -1/(2B_\sigma)$ and $A_7 = -[n(1-a)-1]/2$.
We have 
\[
\begin{aligned}
f_\nu'(\nu) &= aA_1 \e^{a\nu} + (a-1)A_2 \e^{(a-1)\nu} + \tfrac{1}{2} aA_3 \e^\frac{a\nu}{2}  \\
& \quad + \tfrac{1}{2} (a-2)A_4 \e^{\frac{(a-2)\nu}{2}} - A_5 \e^{-\nu} + A_6 \e^\nu + A_7.  \\
f_\nu''(\nu) &= a^2A_1 \e^{a\nu} + (a-1)^2A_2 \e^{(a-1)\nu} + \tfrac{1}{4} a^2A_3 \e^\frac{a\nu}{2} \\
& \quad + \tfrac{1}{4}(a-2)^2A_4 \e^{\frac{(a-2)\nu}{2}}  + A_5 \e^{-\nu} + A_6 \e^\nu. 
\end{aligned}
\]

\section{Extension to multivariate observations}
Suppose the observation equation in the AR(1) plus noise model is replaced by the more general equation,
\begin{align*}
\bfy_t &= Z_t x_t + \bfeps_t,  && \bfeps_t \sim \N(\bfzero, H_t),
\end{align*}
where $\bfy_t \in \mathbb{R}^p$ is the observation vector, $\bfeps_t \in \mathbb{R}^p$ is the error vector, $H_t$ is a $p \times p$ covariance matrix and $Z_t$ is a $p \times 1$ vector assumed to be known. Let $\bfy = (\bfy_1^T, \dots, \bfy_n^T)^T$, $Z = \diag(Z_1, \dots, Z_n)$, $H= \diag(H_1, \dots, H_n)$ and 
\[
\bftheta = (\mu, \sigma_\eta^2, \phi, \vech(H_1)^T, \dots, \vech(H_n)^T)^T.
\]
In matrix notation, 
\[
\bfy|\bfalpha, \bftheta \sim \N(\sigma_\eta^a Z \bfalpha + \mu Z\bfw, H).
\]
The marginal distribution of $\bfy|\bftheta$ is $\N(\mu Z \bone, S)$, where $S= H + \sigma_\eta^2 Z \Lambda^{-1} Z^T$ and $\bfalpha| \bfy,\bftheta \sim \N(\bfm_{a\bfw}, V_a)$, where 
\begin{align*}
V_a &=  \sigma_\eta^{-2a}(Z^T H^{-1} Z + \sigma_\eta^{-2}\Lambda)^{-1}, \\
\bfm_{a\bfw} &= \sigma_\eta^{a} V_a [Z^TH^{-1} (\bfy - \mu Z\bfw)+ \mu \sigma_\eta^{-2}\Lambda \bar{\bfw}].
\end{align*}
When $a=0$ and $\bfw = \bone$, $V_0 =(Z^T H^{-1} Z + \sigma_\eta^{-2}\Lambda)^{-1}$ and $\bfm_{0\bone} = V_0 Z^T H^{-1} (\bfy - \mu Z \bone)$. Further useful identities are given in Appendix \ref{Appendix B}. It can be derived that 
\begin{multline*}
Q(\bftheta|\bftheta^{(i)}) = \tfrac{1}{2} \big[ \log(1-\phi^2) - {z^{(i)T}} H^{-1} z^{(i)} - \log|H| \\
- n(1-a) \log(\sigma_\eta^2) - \sigma_\eta^{2(a-1)} \tr(\Lambda V_a^{(i)}) \\
- \sigma_\eta^{-2} (\sigma_\eta^a \bfm_{a\bfw}^{(i)} -\mu \bar{\bfw})^T \Lambda (\sigma_\eta^a \bfm_{a\bfw}^{(i)} -\mu \bar{\bfw}) \\
- \sigma_\eta^{2a}\tr(Z^T H^{-1} Z V_{a}^{(i)}) - n(p+1)\log(2\pi) \big],
\end{multline*}
where $z^{(i)} = \bfy - \sigma_\eta^a Z \bfm_{a\bfw}^{(i)} - \mu Z \bfw$.
\begin{equation*}
\begin{aligned}
\nabla_\mu Q(\bftheta | \bftheta^{(i)}) &= {z^{(i)T}} H^{-1} Z \bfw+ (\sigma_\eta^a \bfm_{a\bfw}^{(i)} - \mu \bar{\bfw})^T \tfrac{\Lambda}{\sigma_\eta^2} \bar{\bfw}. \\
\nabla_\mu^2 Q(\bftheta | \bftheta^{(i)})
&= - (\bfw^T Z^T H^{-1} Z \bfw+ \sigma_\eta^{-2} \bar{\bfw}^T\Lambda \bar{\bfw}).
\end{aligned}
\end{equation*}
\begin{multline*}
\nabla_{\sigma_\eta^2} Q(\bftheta | \bftheta^{(i)}) 
= \tfrac{1}{2\sigma_\eta^4} \{  a\sigma_\eta^{a+2}(\bfy- \mu Z\bfw)^T H^{-1} Z \bfm_{a\bfw}^{(i)} \\
+ (1-a)\sigma_\eta^{2a} [ \tr(\Lambda V_a^{(i)}) + {\bfm_{a\bfw}^{(i)T}} \Lambda \bfm_{a\bfw}^{(i)} ] + \mu^2 \bar{\bfw}^T \Lambda 
\bar{\bfw}  \\
+ n(a-1)\sigma_\eta^2 + (a-2) \mu \sigma_\eta^a {\bfm_{a\bfw}^{(i)T}} \Lambda \bar{\bfw} - a \sigma_\eta^{2a+2}\\
\times [ \tr(Z^TH^{-1} Z V_a^{(i)})  +{\bfm_{a\bfw}^{(i)T}}Z^T H^{-1} Z \bfm_{a\bfw}^{(i)} ] \}. 
\end{multline*}
\begin{multline*}
\negthickspace\negthickspace\negthickspace \nabla^2_{\sigma_\eta^2} Q(\bftheta | \bftheta^{(i)}) = \tfrac{1}{2\sigma_\eta^4} \{ a (\tfrac{a}{2}-1)\sigma_\eta^a  (\bfy- \mu Z\bfw)^T H^{-1} Z \bfm_{a\bfw}^{(i)} \\
 - a(a-1) \sigma_\eta^{2a}[ \tr(Z^TH^{-1} Z V_a^{(i)})  + {\bfm_{a\bfw}^{(i)T}} Z^T H^{-1} Z \bfm_{a\bfw}^{(i)} ]\\
+n(1-a)- 2 \mu^2 \bar{\bfw}^T  \tfrac{\Lambda}{\sigma_\eta^2} \bar{\bfw} + (a-2) (\tfrac{a}{2}-2) \mu \sigma_\eta^a {\bfm_{a\bfw}^{(i)T}} \tfrac{\Lambda}{\sigma_\eta^2} \bar{\bfw}\\
- (a-1)(a-2) \sigma_\eta^{2a}  [ \tr(\tfrac{\Lambda}{\sigma_\eta^2} V_a^{(i)}) + {\bfm_{a\bfw}^{(i)T}} \tfrac{\Lambda}{\sigma_\eta^2} \bfm_{a\bfw}^{(i)} ] \}.
\end{multline*}

Setting $\nabla_\mu Q(\bftheta|\bftheta^{(i)}) = 0$ at $\bftheta,\bftheta^{(i)}= \bftheta^*$ yields the MLE, 
\begin{equation*}
\begin{aligned}
\mu = \frac{\bfy^T S^{-1} Z \bone}{\bone^T Z^T S^{-1} Z\bone}.
\end{aligned} 
\end{equation*}

Setting $\nabla_{\sigma_\eta^2} Q(\bftheta|\bftheta^{(i)}) = 0$ at $\bftheta,\bftheta^{(i)}= \bftheta^*$, we obtain
\begin{equation*}
\sigma_\eta^2 = \frac{(1-a) \tr(V_0  \Lambda)  + \bfm_{0\bone}^T  \Lambda \bfm_{0\bone} }{n(1-a) + a \tr(Z^T H^{-1} Z V_0) }.
\end{equation*}
Setting $a=0$ or $a=1$, we obtain 
\begin{equation*}
 \bfm_{0\bone}^T  \Lambda \bfm_{0\bone} = \sigma_\eta^2 \tr(Z^T H^{-1} Z V_0) = n \sigma_\eta^2 - \tr(V_0  \Lambda).
\end{equation*}
The elements in the augmented information matrix are  
\begin{equation*}
\begin{aligned}
I_{ \mu, \mu}  &= \bfw^T Z^T H^{-1} Z \bfw+  \sigma_\eta^{-2} \bar{\bfw}^T \Lambda \bar{\bfw}, \\
I_{\sigma_\eta^2, \sigma_\eta^2} &= \tfrac{1}{2\sigma_\eta^4} \big\{ \tfrac{a^2}{2}\bfm_{0\bfw}^T V_0^{-1} \bfm_{0\bfw}- 2a \mu \bfm_{0\bone}^T \tfrac{\Lambda}{\sigma_\eta^2} \bar{\bfw}  \\
& \quad + n(a-1)^2 \big\}.
\end{aligned}
\end{equation*}

Since $\nabla_{\bfw} I_{ \mu, \mu}  = 2[Z^T H^{-1} Z \bfw - \frac{\Lambda}{\sigma_\eta^2} \bar{\bfw} ]$ and $\nabla_{\bfw}^2 I_{ \mu, \mu}  = 2V_0^{-1}$ is positive definite, the rate of convergence of Algorithm 1 is minimized at
\begin{equation*}
\bfw^\opt = \sigma_\eta^{-2} V_0 \Lambda \bone \quad {\text or } \quad \bar{\bfw}^\opt = V_0 Z^T H^{-1} Z \bone.
\end{equation*}

For Algorithm 2, setting $\nabla_a I_{\sigma_\eta^2, \sigma_\eta^2} = 0$ and $\nabla_{\bar{\bfw}} I_{\sigma_\eta^2, \sigma_\eta^2} = 0$ yields
\begin{equation*}
a = \frac{2n + 2\mu \bar{\bfw}^T \frac{\Lambda}{\sigma_\eta^2}  \bfm_{0\bone}}{2n + \bfm_{0\bfw}^T V_0^{-1} \bfm_{0\bfw}}, \quad
\bar{\bfw} = \frac{1}{\mu}  \bigg( \frac{2V_0\Lambda}{a\sigma_\eta^2} - I \bigg)\bfm_{0\bone}.
\end{equation*}
Solving these two equations simultaneously, 
\begin{equation*} 
\begin{aligned}
 a &= 1 - \frac{ \bfm_{0\bone}^T\Lambda \bfm_{0\bone}}{n\sigma_\eta^2} = \frac{\tr(V_0\Lambda) }{n \sigma_\eta^2} = 1 - \frac{\tr(V_0 Z^T H^{-1} Z ) }{n}.
\end{aligned}
\end{equation*}
It can be verified as in Section \ref{sec unknown scale parameter} that the Hessian of $I_{\sigma_\eta^2, \sigma_\eta^2}$ is symmetric positive definite.

In the special case where $p=1$, $Z_t=1$ and $H_t =s_{r_t}^2$, as discussed in Section 6.2, then $Z = I_n$ and $H = D_\bfr$. The rate of convergence of Algorithm 1 is minimized at
\begin{equation*}
\bfw_1 = \sigma_\eta^{-2} V_0 \Lambda \bone \quad {\text or } \quad \bar{\bfw}_1 = V_0 D_\bfr^{-1} \bone,
\end{equation*}
while the rate of convergence of Algorithm 2 is minimized at 
\begin{equation*}
 a_2 = 1 - \frac{\tr(V_0 D_\bfr^{-1}) }{n}, \quad
\bar{\bfw}_2 = \frac{1}{\mu}  \bigg( \frac{2V_0\Lambda}{a_2\sigma_\eta^2} - I \bigg)\bfm_{0\bone}.
\end{equation*}

\begin{table*}[htb!]
\caption{Parameter estimates and runtimes from SV model for exchange rate data.}
 \label{tab_parex}
\begin{tabular}{@{}llcccc@{}}
\hline  & & $\mu$ & $\sigma_\eta$ & $\phi$ & runtime\\    \hline
\multirow{4}{*}{Danish} & NCP & $-$18.05 $\pm$ 0.09 & 0.377 $\pm$ 0.037 & 0.916 $\pm$ 0.015 & 260 \\
& CP    & $-$18.04 $\pm$ 0.09 & 0.372 $\pm$ 0.038 & 0.918 $\pm$ 0.016 & 265 \\
& ASIS & $-$18.04 $\pm$ 0.09 & 0.374 $\pm$ 0.037 & 0.917 $\pm$ 0.015 & 272 \\
& BSR  & $-$18.04 $\pm$ 0.09 & 0.378 $\pm$ 0.038 & 0.916 $\pm$ 0.016 & 269 \\ \hline
\multirow{4}{*}{NZ} &   NCP & $-$10.03 $\pm$ 0.10 & 0.172 $\pm$ 0.031 & 0.964 $\pm$ 0.012 & 291 \\
& CP    & $-$10.02 $\pm$ 0.10 & 0.171 $\pm$ 0.035 & 0.964 $\pm$ 0.014 & 296 \\
& ASIS & $-$10.02 $\pm$ 0.10 & 0.172 $\pm$ 0.031 & 0.964 $\pm$ 0.012 & 302 \\
& BSR  & $-$10.02 $\pm$ 0.10 & 0.175 $\pm$ 0.031 & 0.963 $\pm$ 0.012 & 296 \\ \hline
\multirow{4}{*}{US} & NCP & $-$10.18 $\pm$ 0.20 & 0.065 $\pm$ 0.011 & 0.994 $\pm$ 0.003 & 260 \\
& CP    & $-$10.14 $\pm$ 0.25 & 0.064 $\pm$ 0.011 & 0.994 $\pm$ 0.003 & 264 \\
& ASIS & $-$10.13 $\pm$ 0.24 & 0.065 $\pm$ 0.011 & 0.994 $\pm$ 0.003 & 272 \\
& BSR  & $-$10.14 $\pm$ 0.24 & 0.066 $\pm$ 0.010 & 0.993 $\pm$ 0.003 & 268 \\ \hline
\end{tabular}
\end{table*}

\section{MCMC results for exchange rate data}

Table \ref{tab_parex} shows the parameter estimates and runtimes of each MCMC sampler for the exchange rate data in Section 7.2. The Danish krone has the lowest modal instantaneous volatility ($\mu = -18$) and the highest volatility of the log-volatility ($\sigma_\eta =0.38$), while the US dollar has the highest persistence in volatility ($\phi = 0.99$). In terms of runtime, a similar pattern as in the simulated data is observed. NCP is the fastest, followed by CP, BSR and ASIS, but the runtime of all the MCMC samplers are essentially about the same.

%%%%%%%%%%%%%%%%%%%%%%%%%%%%%%%%%%%%%%%%%%%%%%
%% Single Appendix:                         %%
%%%%%%%%%%%%%%%%%%%%%%%%%%%%%%%%%%%%%%%%%%%%%%
%\begin{appendix}
%\section*{???}%% if no title is needed, leave empty \section*{}.
%\end{appendix}
%%%%%%%%%%%%%%%%%%%%%%%%%%%%%%%%%%%%%%%%%%%%%%
%% Multiple Appendixes:                     %%
%%%%%%%%%%%%%%%%%%%%%%%%%%%%%%%%%%%%%%%%%%%%%%
\begin{appendix} 
\section{Identities for AR(1) plus noise model}\label{Appendix A}
\vspace*{-4mm}
\begin{align*}
V_0 &= \sigma_\eta^{2a}V_a = (\tfrac{I}{\sigma_\epsilon^2} + \tfrac{\Lambda}{\sigma_\eta^2})^{-1}, \\
\bfm_{0\bfw} &= \sigma_\eta^{a}\bfm_{a\bfw} =  V_0 (\bfy -\mu \bone)/\sigma_\epsilon^2 + \mu \bar{\bfw}, \\
\bfm_{0\bone} & = \bfm_{0\bfw} - \mu \bar{\bfw} = V_0 (\bfy - \mu \bone)/\sigma_\epsilon^2, \\
\bfm_{0\bfzero} & = \bfm_{0\bfw} + \mu \bfw  = V_0 (\bfy/\sigma_\epsilon^2 + \mu \tfrac{\Lambda}{\sigma_\eta^2} \bone), \\
\bfy - \bfm_{0\bfzero} &= V_0 \tfrac{\Lambda}{\sigma_\eta^2} (\bfy - \mu \bone) = \sigma_\epsilon^2  \tfrac{\Lambda}{\sigma_\eta^2} \bfm_{0\bone}, \\
S &= \sigma_\epsilon^2 I + \sigma_\eta^2 \Lambda^{-1}, \\
S^{-1} &= \tfrac{\Lambda V_0}{\sigma_\eta^2 \sigma_\epsilon^2}  = \tfrac{V_0 \Lambda}{\sigma_\epsilon^2 \sigma_\eta^2} = \tfrac{I}{\sigma_\epsilon^2}  -\tfrac{V_0}{\sigma_\epsilon^4}, \\
\sigma_\eta^2 \Lambda^{-1} &= \tfrac{V_0 S}{\sigma_\epsilon^2}  = \tfrac{S V_0}{\sigma_\epsilon^2} = V_0 + \tfrac{V_0 S V_0}{\sigma_\epsilon^4}, \\
\tfrac{ \Lambda }{\sigma_\eta^2} V_0 \tfrac{ \Lambda }{\sigma_\eta^2}&= \tfrac{\Lambda}{\sigma_\eta^2} V_0 (V_0^{-1} - \tfrac{I}{\sigma_\epsilon^2} ) = \tfrac{\Lambda}{\sigma_\eta^2} - S^{-1}, \\
\tfrac{\Lambda}{\sigma_\eta^2} \bfm_{0\bone} &= \tfrac{\bfy-\bfm_{0\bfzero}}{\sigma_\epsilon^2} = S^{-1}(\bfy - \mu \bone), \\
\tau(\bfw) &= \bfw^T \bfw/\sigma_\epsilon^2  + \bar{\bfw}^T \Lambda \bar{\bfw}/\sigma_\eta^2, \\
\rho(\bfw) &= V_0^{-1} \bfw - \sigma^{-2} \Lambda \bone, \\
\bone^T S^{-1} \bone &= \tau(\bfw) - \rho(\bfw)^T V_0 \rho(\bfw).
\end{align*} 

\section{Identities for multivariate observations with univariate AR(1) state equation}\label{Appendix B}
\vspace{-3mm}
\begin{align*}
V_0 &=(Z^T H^{-1} Z + \sigma_\eta^{-2}\Lambda)^{-1}, \\
\bfm_{0\bfw} &= \sigma_\eta^{a}\bfm_{a\bfw} = V_0 Z^T H^{-1} (\bfy -\mu Z \bone) + \mu \bar{\bfw}, \\
\bfm_{0\bone} & = \bfm_{0\bfw} - \mu \bar{\bfw} = V_0 Z^T H^{-1}(\bfy - \mu Z\bone) , \\
\bfm_{0\bfzero} & = \bfm_{0\bfw} + \mu \bfw  = V_0 (Z^T H^{-1}\bfy + \mu \tfrac{\Lambda}{\sigma_\eta^2} \bone), \\
%\bfy &- \bfm_{0\bfzero} = V_0 \tfrac{\Lambda}{\sigma_\eta^2} (\bfy - \mu \bone), \\
S&= H + \sigma_\eta^2 Z \Lambda^{-1} Z^T, \\
S^{-1} & = H^{-1} - H^{-1} Z V_0 Z H^{-1}, \\
S^{-1} Z &= H^{-1} Z V_0 \tfrac{\Lambda}{\sigma_\eta^2}.
%\sigma_\eta^2 \Lambda^{-1} &= V_0 D_\epsilon^{-1} S = S D_\epsilon^{-1} V_0 = V_0 + V_0 D_\epsilon^{-1} S D_\epsilon^{-1} V_0, \\
%\tfrac{ \Lambda }{\sigma_\eta^2} V_0 \tfrac{ \Lambda }{\sigma_\eta^2}&= \tfrac{\Lambda}{\sigma_\eta^2} V_0 (V_0^{-1} - D_\epsilon^{-1}) = \tfrac{\Lambda}{\sigma_\eta^2} - S^{-1}, \\
%\tfrac{\Lambda}{\sigma_\eta^2} \bfm_{0\bone} &= D_\epsilon^{-1} (\bfy - \bfm_{0\bfzero}) = S^{-1}(\bfy - \mu \bone), \\
%\tau(\bfw) &= \bfw^T Z^T H^{-1} Z \bfw+  \sigma_\eta^{-2} \bar{\bfw}^T \Lambda \bar{\bfw}, \\
%\rho(\bfw) &= V_0^{-1} \bfw - \sigma^{-2} \Lambda \bone, \\
%\bone^T S^{-1} \bone &= \tau(\bfw) - \rho(\bfw)^T V_0 \rho(\bfw).
\end{align*}

\end{appendix}

\end{document}